\documentclass[a4paper,11pt]{article}

\usepackage{lmodern}
\usepackage[utf8]{inputenc}
\usepackage[T1]{fontenc}
\usepackage[USenglish]{babel}
\usepackage{amsmath,amssymb,amsthm}
\usepackage[margin=1in,footskip=30pt,includefoot]{geometry} 
\usepackage{graphicx}
\usepackage{enumerate}
\usepackage[numbers,sort&compress]{natbib}
\usepackage{stmaryrd}
\usepackage[pdftitle={Dynamic and certifying recognition of proper circular-arc graphs},pdfauthor={Francisco Soulignac},pdfcreator={},pdfsubject={Fully dynamic certification of proper circular-arc graphs},pdfkeywords={},colorlinks=true,linkcolor=blue,citecolor=blue]{hyperref}
\usepackage[NOTIKZ]{optional}  
\opt{NOTIKZ}{
  \newcommand{\imgs}{}
}
\opt{TIKZ}{
  \usepackage{tikz} 
  \usetikzlibrary{intersections}
  \usetikzlibrary{external} 
  \tikzexternalize[prefix=figs/] 
  \newif\ifhthree
  \newif\ifhfive
}

\newtheorem{theorem}{Theorem}[section]
\newtheorem{lemma}[theorem]{Lemma}
\newtheorem{corollary}[theorem]{Corollary}

\newtheorem{observation}[theorem]{Observation}

\newcommand{\B}{\mathcal{B}}
\renewcommand{\to}{\ensuremath{\longrightarrow}}
\newcommand{\nto}{\ensuremath{\longarrownot\to}}
\newcommand{\toeq}{\raisebox{-2pt}{\makebox[0pt][l]{\ensuremath{\overset{=}{\phantom\to}}}}{\to}}

\newcommand{\M}{\mathcal{M}}
\newcommand{\N}{\mathcal{N}}
\newcommand{\MP}{\mathcal{P}}
\newcommand{\Tuple}[1]{\ensuremath{\langle#1\rangle}}
\newcommand{\Bip}[2]{\Tuple{#1,#2}}
\newcommand{\Cat}{\ensuremath{\bullet}}

\newcommand{\RR}{\hat{R}}
\newcommand{\LL}{\hat{L}}
\newcommand{\UU}{\hat{U}}

\newcommand{\epTest}{\operatorname{ep}^+}
\newcommand{\epAdd}{\operatorname{ep}^+}
\newcommand{\epRemove}{\operatorname{ep}^-}

\newcommand{\DS}[1]{\tilde{#1}}

\newcommand{\itemref}[3]{\ensuremath{\rm (\hyperref[def:#2-#3]{#1}_{\ref{def:#2-#3}})}}

\title{A certifying and dynamic algorithm for the recognition of proper circular-arc graphs}
\author{%
  Francisco J.\ Soulignac~\thanks{\small mail: \tt francisco.soulignac@unq.edu.ar}%
}

\date{CONICET and Departamento de Ciencia y Tecnología, Universidad Nacional de Quilmes, 
Buenos Aires, Argentina.}

\begin{document}
\maketitle

\begin{abstract}
  We present a dynamic algorithm for the recognition of proper circular-arc (PCA) graphs, that supports the insertion and removal of vertices (together with its incident edges).  The main feature of the algorithm is that it outputs a minimally non-PCA induced subgraph when the insertion of a vertex fails.  Each operation cost $O(\log n + d)$ time, where $n$ is the number vertices and $d$ is the degree of the modified vertex.  When removals are disallowed, each insertion is processed in $O(d)$ time.  The algorithm also provides two constant-time operations to query if the dynamic graph is proper Helly (PHCA) or proper interval (PIG).  When the dynamic graph is not PHCA (resp.\ PIG), a minimally non-PHCA (resp.\ non-PIG) induced subgraph is obtained.  
  
  \vspace*{.2\baselineskip} {\bf Keywords:} dynamic representation, certifying algorithm, proper circular-arc graphs, proper interval graph, proper Helly circular-arc graphs.
\end{abstract}

\section{Introduction}
\label{sec:introduction}

A \emph{circular-arc (CA) model} is a family of arcs of a circle.  A graph $G$ \emph{admits} a CA model $\M$ when its vertices are in a one-to-one correspondence with the arcs of $\M$ in such a way that two vertices of $G$ are adjacent if and only if their corresponding arcs have a nonempty intersection.  Those graphs that admit a CA model are called \emph{circular-arc (CA) graphs}.  Proper circular-arc graphs and proper interval graphs form two of the most studied subclasses of CA graphs.  A CA model $\M$ is \emph{proper} when no arc of $\M$ is properly contained in another arc of $\M$, while $\M$ is an \emph{interval (IG) model} when the union of its arcs does not cover the entire circle.  A graph is a \emph{proper circular-arc (PCA) graph} when it admits a proper CA model, while it is a \emph{proper interval (PIG) graph} when it admits a proper IG model.

The \emph{(static) recognition problem} for PCA (resp.\ PIG) graphs asks if an input graph $G$ is PCA (resp.\ PIG).  A recognition algorithm that outputs YES or NO is not that useful in practice for two reasons.  First, there are many applications in which PIG and PCA models of $G$ are looked for, while several algorithms work more efficiently when a PIG or PCA model of $G$ is available~\cite{Golumbic2004,HsuTsaiIPL1991,PirlotVincke1997}.  Second, and not less important, a buggy implementation can lead to incorrect answers that a user cannot corroborate.  A \emph{certifying} recognition algorithm yields a \emph{witness} $W$ proving that the output is correct for $G$.  Besides proving correctness, two additional properties are required for $W$~\cite{McConnellMehlhornNaeherSchweitzerCSR2011}.  First, there exists a \emph{checker} with a ``trivial'' implementation that, given $G$ and $W$, \emph{authenticates} that the output is correct for $G$.  Second, there is a simple proof that the existence of $W$ implies the output on $G$.  With these two ingredients, a user can test that the output for $G$ is correct, even when the implemented algorithm has bugs~\cite{McConnellMehlhornNaeherSchweitzerCSR2011}.

Witnesses are classified as \emph{positive} or \emph{negative} according to the output given for $G$.  The former prove that $G$ is PCA (YES output), while the latter prove that $G$ is not PCA (NO output).  A priori, there are many certifying algorithms for the recognition of PCA graphs, as we can chose different kinds of witnesses.  Although all of them can be used to authenticate the output, they need not be equally useful for the user.  This statement is obvious in those application where the goal is to produce a PCA model of $G$, but it also holds for those applications on PCA graphs that require a specific kind of input.  Thus, a positive witness with the required interface is better than one that has to be further processed.  Similarly, a negative witness should highlight the reason why $G$ is not PCA.  Arguably, PCA models are the most useful positive witnesses, while \emph{minimally forbidden subgraphs} are the most useful negative witnesses.

Unfortunately, the positive witnesses that we use in this article are not PCA models, but \emph{round representations}~\cite{DengHellHuangSJC1996}.  Roughly speaking, a round (resp.\ straight) representation $\Phi$ is like a PCA (resp.\ PIG) model in which the actual position of the arcs is missing.  Instead, we know the order of the arcs and which are the leftmost and rightmost arcs intersected by a given arc.  Fortunately, $\Phi$ is enough for all those applications in which knowing the actual position of an arc is not required, e.g.~\cite{HsuTsaiIPL1991}.  Also, it is trivial to obtain a PCA (PIG) model $\M$ \emph{associated} to $\Phi$ in $O(n)$ time, where $n$ is the number of arcs of $\M$; by \emph{associated}, we mean that the arcs of $\Phi$ and $\M$ appear in the same order.

In this article we consider a \emph{dynamic} version of the recognition problem for PCA graphs.  The goal is to keep a round representation $\Phi$ of a graph $G$ while some operations are applied.  We allow two kinds of updates: the insertion of a new vertex (and the edges incident to it), and the removal of an existing vertex (and its incident edges).  Those insertions that yield non-PCA graphs have no effects on $\Phi$; instead, an error message is obtained.  Also, the algorithm must answer if $G$ is PIG or not and, if affirmative, then $\Phi$ must be a straight representation.  Consequently, $\Phi$ can be immediately applied on algorithms that work on PIG graphs.  When efficiency does not matter, the dynamic problem is solved by applying any static recognition algorithm for each update.  The idea, however, is to reduce the complexity of the operations.

To motivate the development of dynamic algorithms for PIG graphs, Hell et al.~\cite{HellShamirSharanSJC2001} describe an application to \emph{physical mapping} of DNA.  The problem is to find a straight representation $\Phi$ of an input graph $G$ that encodes some biological data, or to prove that no such model exists.  As time goes by, further experiments may prove that the initial biological data is not accurate.  The resulting changes in the data correspond to the insertion and removal of vertices and edges from $G$.  Instead of building a new straight representation from scratch, the goal is to ``fix'' $\Phi$ efficiently.  

The concerns about the reliability and usefulness of the outputs, that we had for the static recognition algorithms, hold also for the dynamic ones.  The existence of a round representation $\Phi$ proves that $G$ is a PCA graph, thus $\Phi$ can be taken as the positive witness.  However, when the algorithm rejects an update claiming that it leads to a non-PCA graph, can we trust this claim blindly?  And, even if we do trust, we still want a negative witness to check if the input data is incorrect.  This is particularly true for the above application to physical mapping of DNA, since we expect the experiments to be inaccurate at some point, and we cannot assume the erroneous data yields a PCA graph.  A \emph{certifying} and dynamic algorithm for the recognition of PCA graphs outputs a minimally forbidden subgraph when some update is rejected.

Authenticating that a round representation $\Phi$ encodes $G$ or that $F$ is minimally forbidden subgraph of $G$ are trivial tasks, as desired.  However, the time required for these authentications is linear on the size of $G$.  Thus, we cannot expect the user to authenticate the witnesses after each operation, as doing so throws out the efficiency benefits of the dynamic algorithm.  The difference between static and dynamic algorithms is that the latter are not, strictly speaking, algorithms.  Instead, they are \emph{abstract data types} that keep a certain \emph{data structure} that reacts to different operations.  Thus, $\Phi$ is not given as output when an insertion or removal is applied and, so, $\Phi$ should be authenticated against $G$ only occasionally.  

We can conceive three types of checkers, which we call \emph{static}, \emph{dynamic}, and \emph{monitors}.  Static checkers are static algorithms that authenticate the witness against the static graph $G$.  Dynamic checkers are also static algorithms, but they check one update of the dynamic algorithm against the round representation $\Phi$.   Finally, monitors are dynamic algorithms that ensure the correctness of data structure $\DS{\Phi}$ implementing $\Phi$~\cite{BrightSullivan1995,McConnellMehlhornNaeherSchweitzerCSR2011}.  Thus, a monitor is an abstract data type that sits between the user and the recognition algorithm.  The user interacts with the monitor as if it were a round representation.  In turn, the monitor forwards each operation to $\DS{\Phi}$, while it checks the correct behavior of $\DS{\Phi}$ and the generated output.  In case of an error, the monitor raises an exception.  The main difference between checkers and monitors is that the latter may require access to operations that the are restricted to the user.  Checkers are usually simpler, as they have no knowledge of $\DS{\Phi}$, and can be implemented even when the source code of the recognition algorithm is unavailable.  However, the same reason could make them less efficient.  Thus, checkers and monitors are complementary tools.

\paragraph{Previous work.}
Linear-time algorithms for generating PIG models of graphs are known since more than twenty years, e.g.~\cite{CorneilKimNatarajanOlariuSpragueIPL1995,DengHellHuangSJC1996,HerreraMeidanisPicininIPL1995}.  While dealing with the correctness of their algorithm, Deng et al.~\cite{DengHellHuangSJC1996} prove that a minimally forbidden subgraph $F$ must exist when the algorithm fails.  Although it is not discussed in~\cite{DengHellHuangSJC1996}, $F$ can be obtained in $O(n)$ time.  A second way to find a $F$ is to apply the dynamic recognition algorithm by Hell et al.~\cite{HellShamirSharanSJC2001}.  In a first phase, the algorithm finds a set of vertices $V$ such that the subgraph $G[V]$ induced by $V$ is PIG, and $G[V \cup \{v\}]$ is not PIG.  In a second phase, the algorithm transforms $G[V \cup \{v\}]$ into $F$ by removing vertices from $V$.  This strategy, which is discussed in~\cite{SoulignacA2015} for PCA graphs, costs linear time for PIG graphs.  A similar approach using only incremental graphs is discussed in~\cite{McConnellMehlhornNaeherSchweitzerCSR2011} for planar graphs and in~\cite{KratschMcConnellMehlhornSpinradSJC2006} for interval graphs.  Arguably, the simpler linear time algorithm to find $F$ was presented in~2004 by Hell and Huang~\cite{HellHuangSJDM2004/05}, who extend the LexBFS algorithm by Corneil~\cite{CorneilDAM2004} to exhibit such a forbidden when the input is not PIG.  Meister~\cite{MeisterDAM2005} also applies LexBFS to find a negative witness, but this witness is not always a minimally forbidden subgraph.

The problem of building a PCA model of a graph $G$ is also well settled, but fewer algorithms are known~\cite{DengHellHuangSJC1996,KaplanNussbaumDAM2009,SoulignacA2015}.  Regarding the certification problem, Hell and Huang~\cite{HellHuangSJDM2004/05} show how to obtain a minimally forbidden subgraph $F$ when $G$ is co-bipartite and not PCA.  The first algorithm that shows how to obtain a negative witness when $G$ is not co-bipartite was presented in 2009 by Kaplan and Nussbaum~\cite{KaplanNussbaumDAM2009}.  Unfortunately, their witnesses are not forbidden subgraphs, but odd cycles of incompatibility graphs.  Up to this date, the only algorithm that is able to compute $F$ in linear time for every PCA graphs was given by the author~\cite{SoulignacA2015} in 2015.  The idea is to apply a dynamic recognition algorithm in two phases as discussed above.

Lin and Szwarcfiter~\cite{LinSzwarcfiterDM2009} survey different algorithms for the recognition of other classes of circular-arc graphs, while McConnell et al.~\cite{McConnellMehlhornNaeherSchweitzerCSR2011} discuss a theoretical framework for certifying algorithms and explain why they are preferred over non-certifying ones.  McConnell et al.\ survey certifying recognition algorithms for other classes of graphs as well.

In the last years, dynamic recognition algorithms for many classes of graphs were developed~\cite{CrespellePaulDAM2006,Crespelle2010,CrespellePaulA2010,GioanPaulDAM2012,HeggernesManciniDAM2009,IbarraATA2008,Ibarra2009,IbarraA2010,NikolopoulosPaliosPapadopoulosTCS2012,ShamirSharanDAM2004,TedderCorneil2007}.  Among these examples, the only one providing negative witnesses is the one by Crespelle and Paul~\cite{CrespellePaulDAM2006} for the recognition of directed cographs.  We remark that the minimally forbidden subgraphs for directed cographs have $O(1)$ vertices, thus they are generated when required.  On the other hand, a minimally forbidden subgraph for PCA graphs can have $\Theta(n)$ vertices with degree $O(1)$.  Thus, the computation of such a forbidden is dynamic.

\paragraph{Our results.}
We conceive our manuscript as the forth in a series of articles.  The series begins in 1996 with the recognition algorithm for PCA graphs developed by Deng et al.~\cite{DengHellHuangSJC1996}.  As part of their algorithm, Deng et al.\ devise a vertex-only incremental algorithm for the recognition of connected PIG graphs that runs in $O(d)$ time per vertex insertion, where $d$ is the degree of the inserted vertex.  The data structure that supports their recognition algorithm is a straight representation.  The second article of the series dates back to 2001, where Hell et al.~\cite{HellShamirSharanSJC2001} extend the algorithm by Deng et al.\ to solve the dynamic recognition of PIG graphs.  Their algorithm runs in $O(d + \epTest)$ time per vertex insertion, $O(d + \epRemove)$ time per vertex removal, $O(\epTest)$ time per edge insertion, and $O(\epRemove)$ time per edge removal.  The values of $\epTest$ and $\epRemove$ depend on the data structure employed to implement the straight representations, as depicted in Table~\ref{tab:ep}.  Note that the algorithm is optimal if only insertions are allowed, while it is almost optimal when both operations are allowed.  Indeed, Hell et al.\ prove that at least $\Omega(\log n/(\log\log n + \log b))$ amortized time is required by the fully dynamic algorithm in the cell probe model of computation with word size $b$.  Finally, in 2015, the author~\cite{SoulignacA2015} extended the algorithm by Hell el al.\ for the recognition of PCA graphs.  The algorithm works with round representations and has the same complexity as the one by Hell et al.  Moreover, the round representation is straight when the input graph is PIG, thus the algorithm solves the dynamic recognition of PIG graphs as well.

\begin{table}
  \centering
  \begin{tabular}{l|c|c}
    Type of data structure & $\epTest$   & $\epRemove$ \\\hline
    Incremental            & $O(1)$      & $O(n)$      \\
    Decremental            & $O(n)$      & $O(1)$      \\
    Fully dynamic          & $O(\log n)$ & $O(\log n)$ 
  \end{tabular}

  \caption{The actual values of $\epTest$ and $\epRemove$ according to the data structure employed}\label{tab:ep}
\end{table}

In this article we further extend the algorithm in~\cite{SoulignacA2015} to provide a certifying and dynamic algorithm for the recognition of PCA graphs as discussed above.  The algorithm is restricted to the insertion and removal of vertices, and we ignore the problem for edge operations.  Specifically, the algorithm implements a round representation $\Phi$ of the input graph $G$, and it yields a minimally forbidden subgraph when a vertex insertion fails.  Our algorithm is as efficient as the one by Hell et al., as it handles the insertion of a new vertex $v$ in $O(d + \epTest)$ time, while the removal of $v$ costs $O(d + \epRemove)$ time.  The user can also ask, at any point, if $G$ is a PIG graph; this query costs $O(1)$ time.  If affirmative, then $\Phi$ is a straight representation of $G$.  Otherwise, a minimally forbidden subgraph is obtained.  

We remark that when the insertion of a vertex $v$ with $O(1)$ neighbors fails, the minimally forbidden subgraph $F$ can be of size $\Theta(n)$.  However, only $O(\epTest)$ time is available to generate $F$.  Thus, besides keeping $\Phi$, the dynamic algorithm stores a \emph{partial} forbidden subgraph $\MP(\Phi)$.  When an insertion fails, $\MP(\Phi)$ is extended with $v$ to yield $F$.  This scheme is similar to the one used for the positive witness.  The difference is that $\MP(\Phi)$ is not accessible by the user, who observes the dynamic algorithm as an implementation of $\Phi$.  As is the case with $\Phi$, the output $F$ must provide an efficient and convenient interface to the user.  Of course, because of the inevitable aliasing between $F$ and $G$, no updates on $F$ are possible, and any modification on $G$ invalidates $F$.  If required, a copy of $F$ can be obtained in $O(|E(F)|) = O(n)$ time.

\paragraph{Organization of the manuscript.}  Section~\ref{sec:preliminaries} introduces the basic terminology and notation.  Section~\ref{sec:reception theorem} presents the \emph{Reception Theorem}, which characterizes when a graph $H$ is PCA knowing that $H \setminus \{v\}$ is PCA.  The Reception Theorem sums up the results we require from~\cite{DengHellHuangSJC1996,HellShamirSharanSJC2001,SoulignacA2015}, and guides the certifying algorithm that we develop later.  Section~\ref{sec:data structure} describes the data structure that we employ, and the algorithms required to update the partial forbidden subgraph $\MP(\Phi)$.  Section~\ref{sec:incremental} shows the certifying algorithm we use when a vertex is inserted.  Section~\ref{sec:recognition algorithms} depicts the complete recognition algorithm, including the insertion and removal of vertices and the query for the recognition of PIG graphs.  Section~\ref{sec:recognition algorithms} also discusses the authentication problems.  Finally, Section~\ref{sec:conclusions} has some further remarks and open problems.

\section{Preliminaries}
\label{sec:preliminaries}

For a graph $G$, we use $V(G)$ and $E(G)$ to denote the sets of vertices and edges of $G$, respectively, and $n = |V(G)|$ and $m = |E(G)|$.  The \emph{neighborhood} of a vertex $v$ is the set $N_G(v)$ of all the neighbors of $v$, while the \emph{closed neighborhood} of $v$ is $N_G[v] = N_G(v) \cup \{v\}$.  If $N_G[v] = V(G)$, then $v$ is a \emph{universal vertex}, while if $N_G(v) = \emptyset$, then $v$ is an \emph{isolated vertex}.  Two vertices $v$ and $w$ are \emph{twins} when $N_G[v] = N_G[w]$.  The cardinality of $N_G(v)$ is the \emph{degree} of $v$ and is denoted by $d_G(v)$.  We omit the subscripts from $N$ and $d$ when there is no ambiguity about $G$.

The subgraph of $G$ \emph{induced} by $V \subseteq V(G)$ is the graph $G[V]$ that has $V$ as its vertex set, where two vertices of $G[V]$ are adjacent if and only if they are adjacent in $G$.  A \emph{clique} is a subset of pairwise adjacent vertices.  We also use the term \emph{clique} to refer to the corresponding subgraph.  An \emph{independent set} is a set of pairwise non-adjacent vertices.  A \emph{semiblock} of $G$ is a nonempty set of twin vertices, and a \emph{block} of $G$ is a maximal semiblock.

The \emph{complement} of $G$, denoted by $\overline{G}$, is the graph that has the same vertices as $G$ and such that two vertices are adjacent in $\overline{G}$ if and only if they are not adjacent in $G$.  Each component of $\overline{G}$ is called a \emph{co-component} of $G$, and $G$ is \emph{co-connected} when $\overline{G}$ is connected.   The \emph{union} of two vertex-disjoint graphs $G$ and $H$ is the graph $G \cup H$ with vertex set $V(G) \cup V(H)$ and edge set $E(G) \cup E(H)$.  The \emph{join} of $G$ and $H$ is the graph $G + H = \overline{\overline{G} \cup \overline{H}}$, i.e., $G + H$ is obtained from $G \cup H$ by inserting all the edges $vw$, for $v \in V(G)$ and $w \in V(H)$. 

A graph $G$ is \emph{bipartite} when there is a partition $V_1, V_2$ of $V(G)$ such that both $V_1$ and $V_2$ are independent sets.  Contrary to the usual definition of a partition, we allow one of the sets $V_1$ and $V_2$ to be empty.  So, the graph with one vertex is bipartite for us.  The partition of $V(G)$ into $V_1, V_2$, denoted by $\Bip{V_1}{V_2}$, is called a \emph{bipartition} of $G$.  When $\overline{G}$ is bipartite, $G$ is a \emph{co-bipartite} graph and each bipartition of $\overline{G}$ is a \emph{co-bipartition} of $G$.

For $v \in V(G)$ and $W \subseteq V(G)$, we say that $v$ and $W$ are: \emph{adjacent} if $N(v) \cap W \neq \emptyset$, \emph{co-adjacent} if $W \setminus N(v) \neq \emptyset$, and \emph{fully adjacent} if $W \subseteq N(v)$.   Two disjoint semiblocks $B$ and $W$ are \emph{adjacent} if some vertex in $B$ is adjacent to some vertex in $W$; observe that $B$ and $W$ are adjacent if and only if every vertex in $B$ is fully adjacent to $W$.  If $B \cup W$ is a semiblock, then $B$ is a \emph{twin} of $W$.  A semiblock $B$ is \emph{universal} when its vertices are universal, and \emph{isolated} when the vertices in $B$ are not adjacent to $V(G) \setminus B$.  For a disjoint family of semiblocks $\B$, the subgraph $G[\B]$ of $G$ \emph{induced by $\B$} is obtained from $G[\bigcup \B]$ by removing all but one vertex from each semiblock of $\B$.  Clearly, $G[\B]$ is an induced subgraph of $G$.

\subsection{Orderings and ranges}

In this article, an \emph{order} is a pair $(S, R)$ where $S$ is a finite, and possibly empty, set that admits an enumeration $X = x_1, \ldots, x_n$ of its elements such that $R(x_{i}) = x_{i+1}$ for every $1 \leq i < n$ and $R(x_n) \in \{x_1, \bot\}$ for some undefined value $\bot \not\in S$.   We say that $(S, R)$ is \emph{linear} when $R(x_n) = \bot$, while $(S, R)$ is \emph{circular} when $R(x_n) = x_1$.  When $(S, R)$ is linear, $x_1$ and $x_n$ are the \emph{leftmost} and \emph{rightmost} elements of $(S, R)$.  The enumeration $X$ of $S$ is said to be an \emph{ordering} of $(S, R)$.  

Clearly, every enumeration $X$ of a finite set $S$ defines a linear order and a circular order $(S, R)$, both of which have $X$ as its ordering.  Thus, we say that $X$ is a \emph{linear (resp.\ circular) ordering} to mean that $(S, R)$ is a linear (resp.\ circular) order.  In such a case, we write $R^X$ as a shortcut for $R$; we omit the superscript $X$ when $X$ is clear by context.  For each $x \in X$, the element $R(x)$ is the \emph{right near neighbor} of $x$.   When we want to make no distinctions about $X$ being linear or circular, we simply state that $X$ is an \emph{ordering}.  Note, however, that every ordering is either linear or circular, and it cannot be both at the same time.

Every ordering $X = x_1, \ldots, x_n$ that we shall consider is embedded in some larger circular order.  Hence, all the operations on the indices of $X$ are taken modulo $n$, regardless of whether $X$ is linear or circular.  We use the standard interval notation applied to orders, though we call them \emph{ranges} to avoid confusions with interval graphs.  So, the \emph{range} $[x_i, x_j]$ of $X$ is the linear ordering $x_i, \ldots, x_j$ where, as stated before, $x_1 = x_{n+1}$.  Similarly, the \emph{range} $[x_i, x_j)$ is obtained by removing the rightmost element of $[x_i, x_j]$, the \emph{range} $(x_i, x_j]$ is obtained by removing the leftmost element from $[x_i, x_j]$, and the \emph{range} $(x_i, x_j)$ is obtained by removing both the leftmost and rightmost elements from $[x_i, x_j]$.  The \emph{reverse} of $X$ is the ordering $X^{-1}$ $=$ $x_n$, $\ldots$, $x_1$, where $X^{-1}$ is linear if and only if $X$ is linear.  We write $L^X$ as a shortcut for $R^{X^{-1}}$ and we omit $X$ when it is clear from context.  Note that $L(x_{i+1}) = x_i$ for every $1 \leq i < n$, while $L(x_1) \in \{x_n, \bot\}$ equals $x_n$ if and only if $X$ is circular.  For each $x \in X$, the element $L(x)$ is the \emph{near left neighbor} of $x$.  If $X$ and $Y$ $=$ $y_1$, $\ldots$, $y_m$ are linear orderings, then $X \Cat Y$ denotes the linear ordering $x_1$, $\ldots$, $x_n$, $y_1$, $\ldots$, $y_m$.

The range notation that we use for ranges clashes with the usual notation for ordered pairs.  Thus, we write $\Bip{x}{y}$ to denote the ordered pair $(x, y)$.  The unordered pair formed by $x$ and $y$ is, as usual, denoted by $\{x,y\}$.  Also, for the sake of notation, we sometimes write $\#S$ to denote the cardinality of a range $S$.  For any function $f$, we write $f^0$ to mean the identity and $f^{k+1}(x) = f^k(f(x))$.

\subsection{Contigs, round representations, and proper circular-arc graphs}

In this section we present an alternative definition of proper circular-arc graphs that follows from~\cite{DengHellHuangSJC1996,HuangJCTSB1995}.  These definitions are based on the notion of round representations, which are combinatorial views of proper circular-arc models; see~\cite{Bang-JensenGutin2009}.

A \emph{contig} is a pair $\phi = \Bip{\B(\phi)}{F_r^\phi}$ where $\B(\phi) = B_1, \ldots, B_n$ is an ordering of pairwise disjoint sets, and $F_r^\phi$ is a mapping from $\B(\phi)$ to $\B(\phi)$ such that:
\begin{enumerate}[(i)]
  \item $F_r^\phi(B_i) \in (B_i, F_r^\phi(B_{i+1})]$, for every $1 \leq i < n$, 
  \item if $\B(\phi)$ is linear, then $F_r^\phi(B_n) = B_n$; otherwise $F_r^\phi(B_n) \in [B_1, F_r^\phi(B_1)]$, and
  \item $B_i \not\in (F_r^\phi(B_i), F_r^\phi(F_r^\phi(B_i))]$ for every $1 \leq i \leq n$.
\end{enumerate}
We classify each contig $\phi$ as either linear or circular according to whether $\B(\phi)$ is linear or circular.  Note that $\phi$ is linear if and only if $F_r^\phi(B_n) = B_n$.  

We use a convenient notation for dealing with the range $(B, F_r^\Phi(B)]$.  For $B, W \in \B(\phi)$, we write $B \to_{\phi} W$ to mean that $W \in (B, F_r^\phi(B)]$.  Similarly, we write $B \nto_{\phi} W$ to indicate that $W \not \in (B, F_r^\phi(B)]$, and $B \toeq_{\phi} W$ to indicate that either $B = W$ or $B \to_{\phi} W$.  With the $\to$-notation, we can rewrite conditions (i)--(iii) as follows:
\begin{enumerate}[(i)]
  \item $B_i \to_{\phi} B_{i+1}$ and $B_{i+1} \toeq_{\phi} F_r(B_i)$ for every $1 \leq i < n$,
  \item if $\phi$ is linear, then $B_n \nto_{\phi} B_1$; otherwise, $B_n\to_{\phi} B_1$ and $B_1 \toeq_{\phi} F_r(B_n)$, and
  \item either $B_i \nto_{\phi} B_j$ or $B_j \nto_{\phi} B_i$ for every $1 \leq i \leq j \leq n$.
\end{enumerate}
Figure~\ref{fig:round representation} depicts three contigs with their corresponding $\to$ relation.

The sets in $\B(\phi)$ are referred to as \emph{semiblocks} of $\phi$, while $V(\phi) = \bigcup \B(\phi)$ is the set of \emph{vertices} of $\phi$.  For simplicity, we write $L^\phi$ and $R^\phi$ as shortcuts for $L^{\B(\phi)}$ and $R^{\B(\phi)}$.  Recall that $L^\phi(B)$ and $R^\phi(B)$ are the left and right near neighbors of $B$ for every $B \in \B(\phi)$.  Similarly, we say that $F_r^\phi(B)$ is the \emph{right far neighbor} of $B$.  The \emph{left far neighbor} of $B$ is the unique semiblock $F_l^\phi(B)$ such that $W \to_{\phi} B$ if and only if $W \in [F_l^\phi(B), B)$, for every $W \in \B(\phi)$.  Two other mappings are highly used in this manuscript: $U_r^\phi(B) = R^\phi(F_r^\phi(B))$ and $U_l^\phi(B) = L^\phi(F_l^\phi(B))$ are the \emph{unreached right} and \emph{unreached left} semiblocks, respectively.  As usual, we do not write the subscript and superscript $\phi$ for $L$, $R$, $F_l$, $F_r$, $U_l$, $U_r$, and $\to$ when $\phi$ is clear by context.  Figure~\ref{fig:round representation} shows the values of $R$, $F_r$, and $U_r$ for some contigs.  Note that $\phi$ is linear if and only if $R(B) = \bot$ (and $F_r(B) = B$) for some semiblock $B$.

\begin{figure}
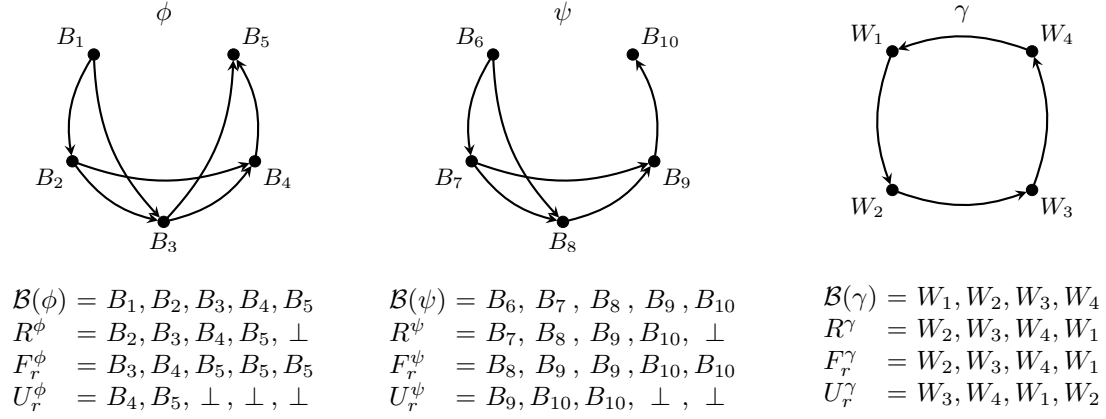

  \small
  \parbox{.33\linewidth}{%
    \centering
    $\phi$\\
    \input{srcFigContigA}
    
    ~
    
    \begin{tabular}{l@{\ $=$\ }c@{,\,}c@{,\,}c@{,\,}c@{,\,}c}
      $\B(\phi)$  & $B_1$ & $B_2$ & $B_3$ & $B_4$ & $B_5$ \\
      $R^\phi$   & $B_2$ & $B_3$ & $B_4$ & $B_5$ & $\bot$ \\
      $F_r^\phi$ & $B_3$ & $B_4$ & $B_5$ & $B_5$ & $B_5$ \\
      $U_r^\phi$ & $B_4$ & $B_5$ & $\bot$ & $\bot$ & $\bot$ 
    \end{tabular}%
  }%
  \parbox{.33\linewidth}{%
    \centering
    $\psi$\\
    \input{srcFigContigB}
    
    ~
    
    \begin{tabular}{l@{\ $=$\ }c@{,\,}c@{,\,}c@{,\,}c@{,\,}c}
      $\B(\psi)$  & $B_6$ & $B_7$    & $B_8$    & $B_9$    & $B_{10}$\\
      $R^\psi$    & $B_7$ & $B_8$    & $B_9$    & $B_{10}$ & $\bot$  \\
      $F_r^\psi$  & $B_8$ & $B_9$    & $B_9$    & $B_{10}$ & $B_{10}$ \\
      $U_r^\psi$  & $B_9$ & $B_{10}$ & $B_{10}$ & $\bot$ & $\bot$ 
    \end{tabular}%
  }%
  \parbox{.33\linewidth}{%
    \centering
    $\gamma$\\
    \input{srcFigContigC}
    
    ~
    
    \begin{tabular}{l@{\ $=$\ }c@{,\,}c@{,\,}c@{,\,}c}
      $\B(\gamma)$  & $W_1$ & $W_2$ & $W_3$ & $W_4$ \\
      $R^\gamma$   & $W_2$ & $W_3$ & $W_4$ & $W_1$ \\
      $F_r^\gamma$ & $W_2$ & $W_3$ & $W_4$ & $W_1$ \\
      $U_r^\gamma$ & $W_3$ & $W_4$ & $W_1$ & $W_2$ 
    \end{tabular}%
  }

  \caption{Two linear contigs $\phi$ and $\psi$ and a circular contig $\rho$.  Each contig $\bullet$ is depicted with their corresponding $\to_\bullet$ relations that corresponds to the semiblock ordering $\B(\bullet)$ and the mappings $R^\bullet$ and $F_r^\bullet$.  The mapping $U_r^\bullet = R^\bullet(F_r^\bullet)$ is also shown.}\label{fig:round representation}
\end{figure}

A \emph{round representation} is a family $\Phi = \{\phi_1, \ldots, \phi_k\}$ of vertex-disjoint contigs such that either $k = 1$ or $\phi_i$ is linear for every $1 \leq i \leq k$.  A \emph{straight representation} is a round representation whose contigs are all linear.  We extend the notation used for contigs to round representations:
\begin{itemize}
  \item $\B(\Phi) = \bigcup_{1 \leq i \leq k}\B(\phi_i)$ and $V(\Phi) = \bigcup_{1 \leq i \leq k}V(\phi_i)$, 
  \item if $B \in \B(\phi_i)$, then $f^\Phi(B) = f^{\phi_i}(B)$ for every $f \in \{L, R, F_l, F_r, U_l, U_r\}$,
  \item if $W \in \B(\phi_j)$, then $B \to_{\Phi} W$ if and only if $i = j$ and $B \to_{\phi_i} W$, and
  \item $B \nto_{\Phi} W$ if and only if $j \neq i$ or $B \nto_{\phi_i} W$.
\end{itemize}
As usual, we omit $\Phi$ from the previous notation.  Note that $\Phi$ is uniquely determined by the triplet $\Tuple{\B(\Phi), R, F_r}$, thus sometimes we write $\Tuple{\B(\Phi), R, F_r}$ as an alternative definition of $\Phi$.  (Note that $\B(\Phi)$ is a family of semiblocks and not a family of orderings.)  Any (linear) contig $\phi$ can be regarded as the (straight) round representation $\{\phi\} = \Tuple{\B(\phi), R^\phi, F_r^\phi}$, thus all the definitions that follow for round representations hold for contigs as well.  We say that a semiblock $B \in \B(\Phi)$ is a \emph{left (resp.\ right) end} semiblock when $F_l(B) = B$ (resp.\ $F_r(B) = B$).  Equivalently, $B$ is a left (resp.\ right) semiblock of $\Phi$ if and only if $B$ is the leftmost (resp.\ rightmost) of its contig, which happens if and only if $L(B) = \bot$ (resp.\ $R(B) = \bot$).  We treat $\bot$ as a special semiblock outside $\B(\Phi)$, one for which $f(\bot) = \bot$ for every $f \in \{L, R, F_r, F_l\}$.   In Figure~\ref{fig:round representation}, $\Phi = \{\phi, \psi\}$ and $\Gamma = \{\gamma\}$ are round representations, whereas $\{\phi, \gamma\}$ is not.  Moreover, $B_1$ and $B_6$ are the left end semiblocks of $\Phi$, while $\gamma$ has no end semiblocks.  Indeed, a round representation is straight if and only if it contains end semiblocks.

Each round representation $\Phi$ defines a graph $G(\Phi)$ that has $V(\Phi)$ as it vertex set, where $v \in B$ and $w \in W$ are adjacent, for $B, W \in \B(\Phi)$, if and only if $B \toeq W$, or $W \toeq B$.  Observe that: each contig $\phi \in \Phi$ defines a component $G(\phi)$ of $G(\Phi)$, each semiblock of $\Phi$ is a semiblock of $G(\Phi)$, and $N_{G(\Phi)}[v] = \bigcup [F_l(B), F_r(B)]$ for every vertex $v \in B$ of $\Phi$.  We say a semiblock of $\B(\Phi)$ is \emph{isolated} or \emph{universal} according to whether it is isolated or universal in $G(\Phi)$.  Similarly, two semiblocks of $\B(\Phi)$ are adjacent or twins when they are adjacent or twins in $\Phi$.  We write $N_\Phi(B)$ to denote the set of semiblocks adjacent to $B$ and $N_\Phi[B] = N_\Phi(B) \cup \{B\}$; we drop the subindex $\Phi$ as usual.  Note that $N[B] = [F_l(B), F_r(B)]$.  In Figure~\ref{fig:round representation}, $G(\{\gamma\})$ is obtained from a cycle with four vertices $w_1$, $w_2$, $w_3$, $w_4$ by adding $|W_i| - 1$ twins of $w_i$, for $1 \leq i \leq 4$.  Also, $B_3$ is universal in $\{\phi\}$ but not in $\{\phi, \psi\}$. 

A graph $G$ is a \emph{proper circular-arc (PCA)} graph if it is isomorphic to $G(\Phi)$ for some round representation $\Phi$.  In such a case, $G$ \emph{admits} $\Phi$, while $\Phi$ \emph{represents} $G$.  It is a well known fact that PCA graphs are precisely those graphs that admit a PCA model, as defined in Section~\ref{sec:introduction}~\cite{DengHellHuangSJC1996,HuangJCTSB1995,SoulignacA2015}.  PCA graphs are characterized by a family of minimal forbidden induced subgraphs, as in Theorem~\ref{thm:forbiddens PCA}.  There, $H^*$ denotes the graph that is obtained from $H$ by inserting an isolated vertex, while $C_n$ denotes the cycle with $n$ vertices.  Graph $\overline{C_3^*}$ is also denoted by \emph{$K_{1,3}$}.

\begin{theorem}[\cite{TuckerDM1974}]\label{thm:forbiddens PCA}
  A graph is a PCA graph if and only if it does not contain as induced subgraphs any of the following graphs: $C_n^*$ for $n \geq 4$, $\overline{C_{2n}}$ for $n \geq 3$, $\overline{C^*_{2n+1}}$ for $n \geq 1$, and the graphs $\overline{S_3}$, $\overline{H_2}$, $\overline{H_3}$, $\overline{H_4}$, $\overline{H_5}$ and $S_3^*$ (see Figure~\ref{fig:forbiddens PCA}).
\end{theorem}

\begin{figure}
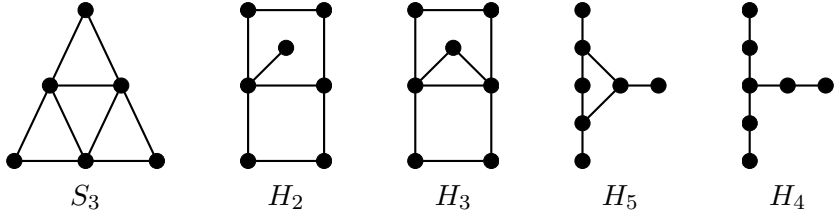

\centering
 \begin{tabular}{c@{\hspace{1cm}}c@{\hspace{1cm}}c@{\hspace{1cm}}c@{\hspace{1cm}}c}
  \input{srcS3} & \input{srcH2} & \input{srcH3} & \input{srcH4} & \input{srcH5}\\
  $S_3$ & $H_2$ & $H_3$ & $H_5$ & $H_4$
 \end{tabular}
 \caption{Complements of the forbidden induced subgraphs for PCA graphs}\label{fig:forbiddens PCA}
\end{figure}

Proper interval graphs are defined as PCA graphs, by replacing round representations with straight representations.  That is, a graph is a \emph{proper interval graph (PIG)} graph when it is isomorphic to $G(\Phi)$ for some straight representation $\Phi$.  It is well known that PIG graphs are precisely those graphs that admit PIG models~\cite{DengHellHuangSJC1996,HellShamirSharanSJC2001}.  PIG graphs are also characterized by minimal forbidden induced subgraphs.

\begin{theorem}[\cite{LekkerkerkerBolandFM1962/1963}]\label{thm:forbiddens PIG}
  A PCA graph is a PIG graph if and only if it does not contain $C_k$ for $k \geq 4$, and $S_3$ as induced subgraphs.
\end{theorem}

Two semiblocks $B, W$ of a round representation $\Phi$ are \emph{indistinguishable} when $F_l(B) = F_l(W)$ and $F_r(B) = F_r(W)$ (e.g., $B_7$ and $B_8$ in Figure~\ref{fig:round representation}).  Clearly, if $B \to W$, then all the semiblocks in $[B, W]$ are pairwise indistinguishable in $\Phi$.  It is not hard to see that $B$ and $W$ are twins when they are indistinguishable.  
We say that $\Phi$ and a round representation $\Psi$ are \emph{equal} when $\Phi$ can be obtained from $\Psi$ by permuting some indistinguishable semiblocks in the contigs of $\Psi$.  Of course, $\Psi$ is an alternative round representation of $G(\Phi)$.  
A PCA graph can also have non-equal representations.  Indeed, $\Phi^{-1} = \Tuple{\B(\Phi), L^\Phi, F_l^\Phi}$, which is called the \emph{reverse} of $\Phi$, is a representation of $G(\Phi)$.  By definition, $R^{\Phi^{-1}} = L^\Phi$, $F_l^{\Phi^{-1}} = F_r^\Phi$, and $\Phi^{-1} = \{\phi^{-1} \mid \phi \in \Phi\}$.  

For $\B \subseteq \B(\Phi)$, we write $\Phi|\B$ to denote the round representation $\Psi$ such that $\B(\Psi) = \B$ and $B \to_{\Psi} W$ if and only if $B \to_{\Phi} W$ for every $B, W \in \B$.  Observe that $\Psi$ is a round representation of $G(\Phi)[V(\Psi)]$, thus $\Psi$ is referred to as the representation of $\Phi$ \emph{induced by} $\B$.  Similarly, the \emph{removal of} $\B$ from $\Phi$ is the representation $\Phi \setminus \B = \Phi|(\B(\Phi) \setminus \B)$; this time, $G(\Phi \setminus \B) = G(\Phi) \setminus (\bigcup \B)$.

We extend the notion of ranges to round representations.  Let $B_i$ be a semiblock of a contig $\phi_i$ of the round representation $\Phi$, for $i \in \{1,2\}$.  When $\phi_1 = \phi_2$, the \emph{range} $[B_1, B_2]$ of $\Phi$ is defined as the range $[B_1, B_2]$ of $\B(\phi_1)$.  When $\phi_1 \neq \phi_2$, the \emph{range} $[B_1, B_2]$ of $\Phi$ is defined as the range $[B_1, B_2]$ of $\B(\phi_1) \Cat \B(\phi_2)$.  That is, $[B_1, B_2]$ is obtained by traversing $\phi_1$ from $B_1$ to its right end semiblock, followed by the range obtained by traversing $\phi_2$ from its left end semiblock to $B_2$.  This non-standard definition is useful to deal with the different possible orderings of the contigs of $\Phi$; in this case, $\phi_1$ would immediately precede $\phi_2$.  For instance, $[B_4, B_7]$ of $\{\phi, \psi\}$ in Figure~\ref{fig:round representation} is $B_4, B_5, B_6, B_7$.  The ranges $(B, W]$, $[B, W)$, and $(B, W)$ of $\Phi$ are defined analogously.

Our definition of ranges allows us to define some robust versions of $L$, $R$, $U_l$, and $U_r$.  By definition, any range $\B = [B_1, B_2]$ of $\Phi$ is a linear ordering, thus $R^\B(B)$ is the semiblock that follows $B$ in $\B$, for any $B \in [B_1, B_2)$.  Let $\phi_1$ and $\phi_2$ be the contigs that contain $B_1$ and $B_2$, respectively.  By definition, $\B$ could contain the right end semiblock $B_r$ of $\phi_1$ followed by the left end semiblock $B_l$ of $\phi_2$.  Although $R^\Phi(B_r) = \bot$ and $L^\Phi(B_l) = \bot$, the semiblocks $R^\B(B_r) = B_l$ and $L^\B(B_l) = B_r$ are well defined.  The \emph{robust} version $R^{\Bip{\Phi}{\B}}$ of $R^\Phi$ behaves exactly as $R^\Phi$, with the exception that it maps $B_r$ to $B_l$.  Similarly, the \emph{robust} version $L^{\Bip{\Phi}{\B}}$ of $L^\Phi$ maps $B_l$ to $B_r$ and $B$ to $L^\Phi(B)$ for $B \neq \B(\Phi) \setminus \{B_l\}$.  The \emph{robust} versions $U_l^{\Bip{\Phi}{\B}}$ of $U_l^\Phi$, and $U_r^{\Bip{\Phi}{\B}}$ of $U_r^\Phi$ are defined analogously.  For the sake of notation, we write $\LL$, $\RR$, $\UU_l$ and $\UU_r$ when $\Phi$ and $\B$ are clear.  Thus, if we consider the range $[B_4, B_7]$ of $\{\phi, \psi\}$ in Figure~\ref{fig:round representation}, then $\RR(B_5) = B_6$, $\LL(B_6) = B_5$, thus $\UU_r(B_3) = B_6$, and $\UU_l(B_8) = B_5$; note, however, that $\LL(B_1) = \RR(B_{10}) = \bot$.

Before we advance, we describe the rationale behind our definitions of round representations, ranges, and the robust mappings $\LL$ and $\RR$.  Our main goal in this article is to insert a new vertex $v$ into a round representation $\Phi$ to obtain a new round representation $\Psi$.  As we shall see, $v$ can have neighbors in at most two contigs $\phi_1$ and $\phi_2$ of $\Phi$ (possibly $\phi_1 = \phi_2$).  To insert $v$, we must join the semiblocks in $\B(\phi_1) \cup \B(\phi_2)$ together with $v$ into a new contig $\psi$ that ``replaces'' both $\phi_1$ and $\phi_2$, i.e., $\Psi = (\Phi \setminus \{\phi_1, \phi_2\}) \cup \{\psi\}$.  Since $\psi$ is a contig, $\B(\phi_1) \cup \B(\phi_2)$ must be somehow together in $\Psi$.  Prior to the insertion of $v$, any pair of contigs of $\Phi$ could play the role of $\phi_1$ and $\phi_2$, thus it is inconvenient to prefix an ordering of the contigs of $\Phi$.  As this ordering is absent, it makes no sense to define the follower (resp.\ predecessor) of a right (resp.\ left) end semiblock.  However, once $v$ and $N(v)$ are given, we have access to the neighbor semiblocks in $\phi_1$ and $\phi_2$.  A priori, there is no way of knowing if $\phi_1 = \phi_2$; all we can query is if $v$ is adjacent to end semiblocks.  Yet, since $\psi$ is a contig, the semiblocks adjacent to $v$ must appear consecutively in $\psi$.  In other words, $N(v)$ should be a range of $[B_1, B_2]$ of $\B(\phi_1) \Cat \B(\phi_2)$.  We want to make no case distinctions according to whether $\phi_1 = \phi_2$ or whether $[B_1, B_2]$ has end semiblocks.  This is the reason why ranges are defined for semiblocks in different contigs, and why the range of an ordering can include its rightmost element.  Finally, to test if $v$ can be inserted, we have to check some conditions on $R(B_m)$ for $B_m \in [B_1, B_2]$.  However, this semiblock is not well defined when $R(B_m) = \bot$ and, in this case, the role of this semiblock is played by the left end semiblock of $\phi_2$.  The robust definition of $\RR$ allows us to treat the case in which $R(B_m) = \bot$ in the same way as we treat the other case.

Although we need access to $\B$ for the robust versions to be efficient, there is one case in which specifying $\B$ is not required.  If $\Phi = \{\phi_1, \phi_2\}$ for (possibly equal) contigs $\phi_1, \phi_2$, then $\RR(B_r) = B_l$ and $\LL(B_l) = B_r$ for the left end semiblock $B_l$ of $\phi_i$ and the right end semiblock $B_r$ of $\phi_j$ (if existing), $i,j = \{1,2\}$, while $\RR(B) = R(B)$ and $\LL(B) = L(B)$ for every other semiblock.  In this case, we refer to $\Phi$ is being \emph{robust}.

By definition, each contig $\phi$ of a straight representation $\Phi$ is ``equivalent'' to a range $[B_l, B_r]$ of $\Phi$, where $B_l$ is a left end semiblock, $B_r$ is a right end semiblock, and $(B_l, B_r)$ has no end semiblocks.  The term ``equivalent'' employed here means that $\{\phi\}$ equals $\Phi|[B_l, B_r]$; moreover, $\Phi|[B_l, B_r]$ is a round representation of some component of $G'$ of $G(\Phi)$.  We refer to $[B_l, B_r]$ simply as a \emph{contig range} of $\Phi$ that \emph{describes} $G'$.  The following observation then follows.

\begin{observation}
  If $\Phi$ is a straight representation of a graph $G$, then every component of $G$ is described by a contig range.
\end{observation}

In a similar way, the co-components of co-bipartite PCA graphs are described by co-contig pairs (see Figure~\ref{fig:co-contigs}).  Let $\Phi$ be a round representation of a co-bipartite graph $G$; observe $\Phi$ is robust.  Say that a non-universal semiblock $B \in \B(\Phi)$ is a \emph{left co-end semiblock} of $\Phi$ if $B = \UU_r^2(B)$; similarly, $B$ is a \emph{right co-end semiblock} of $\Phi$ when $B = \UU_l^2(B)$.  A \emph{co-contig range} is a range $[B_l, B_r]$ such that $B_l$ is a left co-end semiblock, $B_r$ is a right co-end semiblock, and $(B_l, B_r)$ has no co-end semiblocks.  As it is shown in~\cite{HuangJCTSB1995,SoulignacA2015}, $\overline{\B} = [\UU_r(B_l), \UU_l(B_r)]$ is also a co-contig range for any co-contig range $\B = [B_l, B_r]$.  Moreover, $G' = G[\bigcup(\B \cup \overline{\B})]$ is a co-component of $G$, and $\Bip{\bigcup\B}{\bigcup \overline{\B}}$ is a co-bipartition of $G'$.  The pair $\Bip{\B}{\overline{\B}}$ is said to be a \emph{co-contig pair} of $\Phi$ that \emph{describes} $G'$, while $\Phi|(\B \cup \overline{\B})$ is a \emph{co-contig} of $\Phi$.  

\begin{theorem}[\cite{HuangJCTSB1995,SoulignacA2015}]
  If $\Phi$ is a round representation of a co-bipartite graph $G$, then $\Phi$ is robust and every co-component of $G$ with at least two vertices is described by a co-contig pair.
\end{theorem}

Our definition of co-contig pairs above explicitly excludes universal semiblocks.  Clearly, each vertex in a universal semiblock induces a co-component of $G$.  We say that a universal semiblock $B$ is both a \emph{left co-end} and \emph{right co-end} semiblock.  Hence, $\B = [B, B]$ is a \emph{co-contig range}, $\Bip{\B}{\emptyset}$ is a \emph{co-contig pair} that describes $G' = G[\bigcup \B]$, and $\Phi|\{B\}$ is a \emph{co-contig} of $\Phi$.  

As defined so far, co-contigs only represent co-bipartite graphs.  For the sake of generality, we say that a round representation $\Phi$ is a \emph{co-contig} of $\Phi$ when $G(\Phi)$ is co-connected.  Note that, consequently, we may not assume that co-contigs are robust or have co-contig ranges.  Also, to make clear the parallelism between contigs and co-contigs, we use lowercase Greek letters to name the co-contigs of $\Phi$ when $G(\Phi)$ needs not be co-connected. 

\begin{figure}
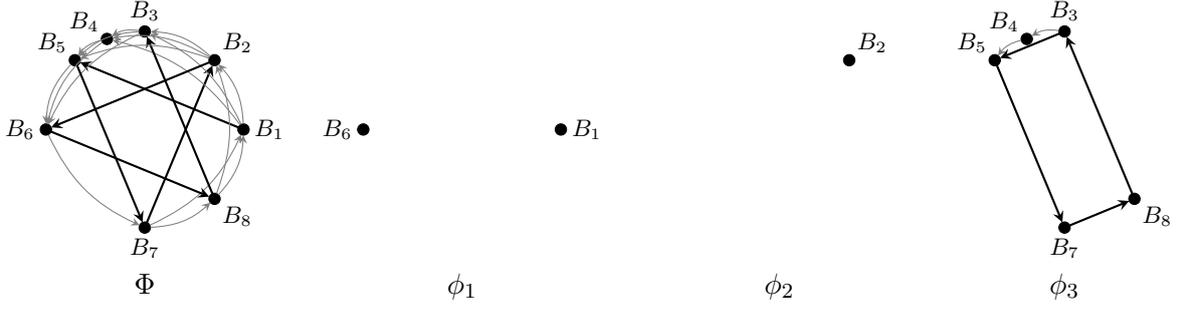

 \begin{tabular}{c@{\hspace{5mm}}c@{\hspace{5mm}}c@{\hspace{5mm}}c}
  \input{srcFigCoContigs} & \input{srcFigCoContigA} & \input{srcFigCoContigB} & \input{srcFigCoContigC}\\
  $\Phi$ & $\phi_1$ & $\phi_2$ & $\phi_3$
 \end{tabular}
  
  \caption{A round representation $\Phi$ with three co-contigs $\phi_1$, $\phi_2$, and $\phi_3$ (gray lines are used to depict the edges that are implied by black edges).  The left co-end semiblocks of $\Phi$ are $B_1$, $B_2$, $B_3$, $B_6$, and $B_7$, while its right co-end semiblocks are $B_1$, $B_2$, $B_5$, $B_6$, and $B_8$.  Note that each co-contig pair other than $\Bip{B_2}{\emptyset}$ describes a co-component of $G(\Phi)$.}\label{fig:co-contigs}
\end{figure}

\section{The Reception Theorem: a certification roadmap}
\label{sec:reception theorem}

Receptive pairs are the main concept used in~\cite{SoulignacA2015} for dealing with the insertion of a non-isolated vertex $v$ into $G$.  In simple terms, a pair of semiblocks is receptive when it witnesses that $H = G \cup \{v\}$ is PCA.  Its definition, however, depends on whether $v$ belongs to an end semiblock or not.  Suppose $H$ is represented by a round representation $\Psi$ and $\{v\}$ is a semiblock of $\Psi$.  Let $B_l = F_l(\{v\})$ if $\{v\}$ is not a left end semiblock, and $B_l = R(\{v\})$ otherwise.  Similarly, $B_r = F_r(\{v\})$ if $\{v\}$ is not a right end semiblock, while $B_r = L(\{v\})$ otherwise.  By definition, $\Phi = \Psi \setminus \{v\}$ is a round representation of $G$.  The pair $\Bip{B_l}{B_r}$ is referred to as being \emph{receptive} in $\Phi$, while $\Psi$ is the \emph{$\{v\}$-reception of $\Bip{B_l}{B_r}$} in $\Phi$.  Strictly speaking, $v$ plays no role when deciding if a pair is receptive; $\Bip{B_l}{B_r}$ is receptive if and only if $G \cup \{w\}$ is a PCA graph for any $w$ with $N(w) = N(v) = \bigcup [B_l, B_r]$ (recall $v$ is not isolated).  When applied to $H$ and $v$, we obtain that $H$ is a PCA graph if and only if $G$ admits a round representation $\Phi$ with a receptive pair $\Bip{B_l}{B_r}$ such that $N(v) = \bigcup [B_l, B_r]$.

As defined, the concept of receptive pairs applies to any round representation.  Yet, the dynamic algorithm deals with a rather restricted subset of round representations.  Say that a semiblock $B$ of a round representation $\Psi$ is a \emph{block} when $B$ is a block of $G(\Psi)$.  If every semiblock of $\Psi$ is a block, then $\Psi$ is a \emph{round block representation} and all its (co-)contigs are referred to as \emph{block (co-)contigs}.  When $\Psi$ is a round block representation, the round representation $\Phi = \Psi \setminus \{v\}$ is almost a block representation.  In fact, it can be easily observed that $\{L(B_l), B_l\}$ and $\{B_r, R(B_r)\}$ are the unique possible pairs of indistinguishable semiblocks of $\Phi$, while $\Phi$ has at most two universal semiblocks, one in $[B_l, B_r]$ and the other outside $[B_l, B_r]$.  For the sake of notation, we refer to $\Phi$ as \emph{$v$-receptive} when it contains a receptive pair $\Bip{B_l}{B_r}$ such that:
\begin{itemize}
  \item $N(v) = \bigcup [B_l, B_r]$, and
  \item no pair of semiblocks in $\B(\Phi) \setminus \{B_l, B_r\}$ are indistinguishable.
\end{itemize}

\begin{theorem}\label{thm:H is PCA iff v-reception}
  Let $H$ be a graph such that $v \in V(H)$ is not isolated.  Then, $H$ is PCA if and only if $H \setminus \{v\}$ admits a round representation $\Phi$ that is $v$-receptive.  Furthermore, the $\{v\}$-reception $\Phi$ is a round representation of $H$.
\end{theorem}

The above observation is quite straightforward, but it tells us little about the $v$-receptive representations of $G$.  In~\cite{DengHellHuangSJC1996,HellShamirSharanSJC2001,SoulignacA2015}, tools for efficiently finding and transforming $\Phi$ into a round representation of $H$ are developed.  The Receptive Pair Lemma of~\cite{SoulignacA2015}, that generalizes some results in~\cite{DengHellHuangSJC1996,HellShamirSharanSJC2001}, is one of such tools.  For the sake of simplicity, we present a unified view of~\cite{DengHellHuangSJC1996,HellShamirSharanSJC2001,SoulignacA2015}.

Let $\Phi$ be a round representation and $B_l \neq B_r$ be semiblocks of $\B(\Phi)$.  A semiblock $B_m \in \B(\Phi)$ \emph{witnesses that $\Bip{B_l}{B_r}$ is receptive} in $\Phi$ when (see Figure~\ref{fig:receptive}):
\begin{enumerate}[({wit}$_a$)]
  \item $B_m$ is an end semiblock, $B_l \toeq B_m$, and $B_m \toeq B_r$, or \label{def:wit-1}
  \item $B_l \toeq B_m$, $B_m \nto R(B_r)$, $L(B_l) \nto R(B_m)$, and $\RR(B_m) \toeq B_r$. \label{def:wit-2}
\end{enumerate}
\newcommand{\witnessref}[1]{\itemref{wit}{wit}{#1}}
The essence of the insertions methods in~\cite{DengHellHuangSJC1996,HellShamirSharanSJC2001,SoulignacA2015} is captured in the next lemma.  

\begin{figure}
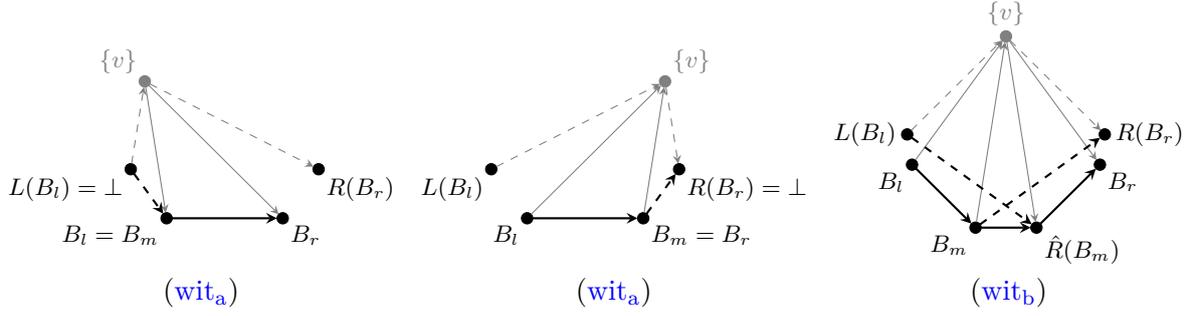

  \centering
  \begin{tabular}{c@{\hspace{3.5mm}}c@{\hspace{3.5mm}}c}
    \input{srcFigReceptiveA1} & \input{srcFigReceptiveA2} & \input{srcFigReceptiveB} \\
    \witnessref{1} & \witnessref{1} & \witnessref{2}
  \end{tabular}
  \caption{The possibilities for a semiblock $B_m$ to witness that $\Bip{B_l}{B_r}$ is receptive. A dashed edge from $X$ to $Y$ means that $X \nto Y$, while $X \to Y$ is possible when the edge is missing (and $\Phi$ is a round representation). Also, some of the depicted semiblocks could be equal.  The corresponding $\{v\}$-receptions are drawn in gray.}\label{fig:receptive}
\end{figure}

\begin{lemma}[Receptive Pair Lemma~\cite{DengHellHuangSJC1996,HellShamirSharanSJC2001,SoulignacA2015}]\label{lem:receptive representation}
  Let $\phi_1, \phi_2$ be two (possibly equal) contigs that contain the semiblocks $B_l$ and $B_r$, respectively.  Then, $\Bip{B_l}{B_r}$ is receptive in $\{\phi_1, \phi_2\}$ if and only if $B_m \in [B_l, B_r]$ witnesses that $\Bip{B_l}{B_r}$ is receptive in $\{\phi_1, \phi_2\}$.
\end{lemma}
\newcommand{\receptionref}[1]{\itemref{rec}{rec}{#1}}

The Receptive Pair Lemma can be proved with not much effort by following the definitions (see~\cite{SoulignacA2015} and Figure~\ref{fig:receptive}).  Indeed, if $B_m$ witnesses that $\Bip{B_l}{B_r}$ is receptive, then a $\{v\}$-reception is obtained by: inserting $\{v\}$ immediately to the right of $B_m$ if \witnessref{2} or $B_m \neq B_l$; or inserting $\{v\}$ immediately to the left of $B_m = B_l$ if \witnessref{1}.  On the other hand, if $\Psi$ is a $\{v\}$-reception of $\Bip{B_l}{B_r}$, then either $L(\{v\})$ or $R(\{v\})$ (if $L(\{v\}) = \bot$) witnesses that $\Bip{B_l}{B_r}$ is receptive.  The Receptive Pair Lemma is an asymmetric tool: it suffices to find one $v$-receptive representation of $G$ to claim that $H$ is PCA, while all the round representations of $G$ must be discarded before claiming that $H$ is not PCA.  Moreover, a round representation of $H$ is available when $H$ is PCA, whereas no minimally forbidden is found when $H$ is not PCA.  The Reception Theorem combines Theorem~\ref{thm:H is PCA iff v-reception} with a slight generalization of the Receptive Pair Lemma that takes $N(v)$ into account. For a better exposition, we consider only the case in which $H$ is connected.  Nevertheless, conditions \receptionref{1}--\receptionref{3} are general.

\begin{corollary}[Reception Theorem]\label{thm:reception}
  Let $H$ be a connected graph with a vertex $v$.  Then, $H$ is PCA if and only if $H \setminus \{v\}$ admits a round representation $\Phi$ that contains two semiblocks $B_l$, $B_r$ such that:
  \begin{enumerate}[({rec}$_1$)]
    \item $N(v) = \bigcup [B_l, B_r]$, \label{def:rec-1}
    \item no pair of semiblocks in $\Phi \setminus \{B_l, B_r\}$ are indistinguishable, \label{def:rec-2}
    \item $B_m \in [B_l, B_r]$ witnesses that $\Bip{B_l}{B_r}$ is receptive in $\Phi$. \label{def:rec-3}
  \end{enumerate}
\end{corollary}

Technically speaking, \receptionref{1}--\receptionref{3} are statements dealing with pairs of semiblocks.  For the sake of simplicity, we say that a round representation $\Phi$ \emph{satisfies} a subset $P$ of \receptionref{1}--\receptionref{3} when $\Phi$ has two semiblocks $B_l$ and $B_r$ that simultaneously satisfies all the conditions in $P$.  

Despite the simplicity of the Reception Theorem, the problem of finding a $v$-receptive representation is not an easy task, specially when the time constraints are tight.  Most of the effort in~\cite{DengHellHuangSJC1996,HellShamirSharanSJC2001,SoulignacA2015} is spent on finding such a $v$-receptive representation efficiently.  The problem of finding a minimally forbidden is mostly, but not completely, ignored in these articles.  In fact, the Reception Theorem made its first appearance in~\cite{DengHellHuangSJC1996}, where the authors consider a rather restricted case in which $G$ is PIG and both $G$ and $H$ are connected.  The advantage of this case is that $G$ admits only two contigs, namely $\gamma$ and $\gamma^{-1}$.  By~\receptionref{1}, $N(v)$ must be a range of $\gamma$, which implies that there are exactly two contigs $\phi$ and $\phi^{-1}$ that can satisfy \receptionref{1} and \receptionref{2}.  In their proof of the Reception Theorem, Deng et al.~\cite{DengHellHuangSJC1996} exhibit a minimally forbidden of $H$ when either $N(v)$ is not consecutive in $\gamma$ or $\phi$ does not satisfy \receptionref{3}.  Although it is not explicit in~\cite{DengHellHuangSJC1996}, an $O(n)$ time algorithm for obtaining such a minimally forbidden, when $\gamma$ and $N(v)$ are give as input, follows as a by-product.  It is not hard to extend this certified algorithm to the case in which the PIG graph $G$ can be disconnected.

Our aim in the present manuscript is to design a certifying and dynamic algorithm for the recognition of PCA graphs, following the framework discussed in Section~\ref{sec:introduction}.  Thus, we ought to compute a minimally forbidden each time an insertion of a vertex $v$ fails.  The main idea is to prove the Reception Theorem following the same path as Deng et al.  That is, we show a minimally forbidden of $H$ when \textbf{no} round representation of $H \setminus \{v\}$ is $v$-receptive.  However, we spend $O(d(v)+\epTest)$ time to build the minimally forbidden.

\section{The data structure}
\label{sec:data structure}

As anticipated, the algorithm keeps two differentiated data structures.  One implements a round block representation to witness that $G$ is PCA, while the other represents an induced path of $G$ that is extended to a negative witness when the insertion of a vertex fails.  The following sections present the data types involved in the dynamic algorithm.

\subsection{Contigs}
\label{sec:data structure:contig}

The data structure we use to implement contigs is the same as discussed in~\cite{SoulignacA2015}.  For completeness, we describe its interface as an abstract data type; for implementations details see~\cite{SoulignacA2015}.

Each contig $\phi$ presents itself as the collection of semiblocks $\B(\phi)$, where each $B \in \B(\phi)$ stores its set of vertices.  Also, each $B \in \B(\phi)$ and each $v \in B$ allow the user to store some additional data.  The internal structure and the semiblocks and vertices of $\phi$ are exclusively owned by $\phi$, thus the modifications applied on $\phi$ have no impact on the data structures of other contigs.  Moreover, a user cannot directly access $\phi$, a semiblock $B \in \B(\phi)$, or a vertex $v \in B$.  Instead, a \emph{semiblock (resp.\ vertex) pointer} $\DS{B}$ (resp.\ $\DS{v}$) \emph{associated to} $\phi$ must be employed to access $B$ (resp.\ $v$).  For simplicity, we also say that $\phi$ is \emph{referenced by} $\DS{B}$ (resp.\ $\DS{v}$).  The pointer $\DS{B}$ is aware of the internal structure of $\phi$, thus it can be used used to manipulate both $B$ and $\phi$.  However, $\DS{B}$ knows nothing about the other semiblock pointers associated to $\phi$ or the semiblocks in $\B(\phi) \setminus \{B\}$.   Hence, there is no way to answer, in $O(1)$ time, if $\DS{B}$ is associated to the same contig as other pointer $\DS{W}$.  (Roughly speaking, $\phi$ is similar to doubly linked lists in which the semiblocks play the role of nodes and semiblock pointers are pointers to the nodes.)

The following functions that operate on contigs and semiblocks are provided in~\cite{DengHellHuangSJC1996,HellShamirSharanSJC2001,SoulignacA2015}.  As usual, we use lower case Greek letters for contigs, capital Roman letters for semiblocks, and tildes to indicate pointers.
\begin{description}
  \item [\tt newContig()] creates a new contig that contains only one block $B = \{v\}$ and returns the pointers to $B$ and $v$.  Complexity: $O(1)$ time.
  \item [{\tt vertices($\DS{B}$)}] returns (an iterator to) $\{\DS{v} \mid v \in B\}$. Complexity: $O(1)$ time.
  \item [{\tt semiblock($\DS{v}$)}] returns a pointer to the semiblock that contains $v$. Complexity $O(1)$ time.
  \item [$f(\DS{B})$] returns a semiblock pointer to $f(B)$ for $f \in \{L, R, F_l, F_r, U_r, U_l\}$.  Complexity: $O(1)$ time.
  \item [${[\DS{B}_1, \DS{B}_2]}$] returns a list of semiblocks pointers that represents the range $[B_1, B_2]$ of $\{\phi_1,\phi_2\}$, where $\phi_i$ is the contig referenced by $\DS{B}_i$ for $i \in \{1,2\}$.  The ranges $(\DS{B}_1, \DS{B}_2]$, $[\DS{B}_1, \DS{B}_2)$ and $(\DS{B}_1, \DS{B}_2)$ work similarly. Complexity: $O(\#[B_1, B_2])$ time.  
  \item [\tt receptive($\DS{B}_l$, $\DS{B}_r$)] is true when $\Bip{B_l}{B_r}$ is receptive in $\{\phi_1, \phi_2\}$, where $\phi_1$ and $\phi_2$ are the contigs referenced by $\DS{B}_l$ and $\DS{B}_r$, respectively.  Complexity: $O(\#[B_l, B_r])$ time.  
  \item [\tt reception($\DS{B}_l$, $\DS{B}_r$)] transforms $\phi_1$ and $\phi_2$ into the $\{v\}$-reception $\psi$ of $\Bip{B_l}{B_r}$, where $\phi_1$ and $\phi_2$ are the contigs referenced by $\DS{B}_l$ and $\DS{B}_r$, respectively.  Returns a pointer to $\{v\} \in \B(\psi)$.  Requires $\Bip{B_l}{B_r}$ to be receptive in $\{\phi_1, \phi_2\}$.  Complexity: $O(\#[B_l, B_r])$ time.  
  \item [\tt remove($\DS{B}$)] transforms the contig $\psi$ referenced by $\DS{B}$ into the contigs of $\{\psi\} \setminus \{B\}$ and the contig $\gamma$ whose only semiblock is $B$.  Note that $\{\psi\} \setminus \{B\}$ has either one or two contigs, thus at most three contigs are generated.   Complexity: $O(\#[F_l(B), F_r(B)])$ time.  
  \item [\tt separate($\DS{B}$, $W$)] transforms the contig $\gamma$ referenced by $\DS{B}$ into a contig $\phi$ representing $G(\gamma)$ that is obtained by splitting $B$ into two indistinguishable semiblocks $B \setminus W$ and $W$ in such a way that $R^\phi(W) = R^\gamma(B)$, $L^\phi(W) = B \setminus W$, and $L^\phi(B \setminus W) = L^\gamma(B)$.  The other semiblocks of $\gamma$ are not affected by this operation.  It requires $W \subseteq B$, and it has no effects when either $W = B$ or $W = \emptyset$.  Note that $W$ is not a semiblock pointer, but a set of vertex pointers.  Complexity: $O(|W|)$ time.  (See Figure~\ref{fig:operations}.)
  \item [\tt separate($W$, $\DS{B}$)] does the same as {\tt separate($\DS{B}$, $B\setminus W$)}.  Complexity: $O(|W|)$ time.
  \item [\tt compact($\DS{B}$)] transforms the contig $\phi$ referenced by $\DS{B}$ into a contig $\gamma$ representing $G(\phi)$ that is obtained by joining the indistinguishable semiblocks $B$ and $R(B)$ into one semiblock $B \cup R(B)$ in such a way that $L^\gamma(B \cup R(B)) = L^\phi(B)$ and $R^\gamma(B\cup R(B)) = R^\phi(B)$.  The other semiblocks of $\phi$ are not affected by this operation.  It has no effects when $B$ and $R(B)$ are not indistinguishable.  Complexity: $O(\min(|B|, |R(B)|)$ time. (See Figure~\ref{fig:operations}.)
  \item [\tt join($\DS{B}_1$, $\DS{B}_2$)] transforms $\phi_1$ and $\phi_2$ into a new contig $\psi$ that represents $G(\phi_1) + G(\phi_2)$, where $\phi_i$ is the contig referenced by $\DS{B}_i$ for $i \in \{1,2\}$.  The semiblocks of $\psi$ appear as in the ordering $[B_1, F_r^{\phi_1}(B_1)]$ $\Cat$ $[B_2, F_r^{\phi_2}(B_2)]$ $\Cat$ $(F_r^{\phi_1}(B_1), B_1)$ $\Cat$ $(F_r^{\phi_2}(B_2), B_2)$.  It has an undefined behavior when either 1.~$B_1$ or $B_2$ is not a left co-end block or 2.~$\DS{B}_1$ and $\DS{B}_2$ are associated to the same contig.  Time complexity: $O(u)$ time, where $u$ is the number of universal semiblocks in $\psi$. (See Figure~\ref{fig:operations}.)
  \item [\tt split($\DS{B}_1$, $\DS{B}_2$)] transforms the contig $\psi$ referenced by $\DS{B}_1$ and $\DS{B}_2$ into two co-contigs $\phi_1$ and $\phi_2$ in such a way that $G(\psi) = G(\phi_1) + G(\phi_2)$.  The semiblocks of $\phi_1$ appear as in the ordering $[B_1, B_2) \Cat (F_r^\psi(B_1), F_r^\psi(B_2)]$, while the semiblocks of $\phi_2$ appear as in the ordering $[B_2, F_r^\psi(B_1)] \Cat (F_r^\psi(B_2), B_1)$.  It has an undefined behavior when either 1.~$B_1$ or $B_2$ is not a left co-end block or 2.~$\DS{B}_1$ and $\DS{B}_2$ are associated to different contigs.  Note that this operation requires $B_1 \to B_2$; if this is not the case, then $B_2 \to B_1$, thus the operation works as if the parameters where inverted.  Time complexity: $O(u)$ time, where $u$ is the number of universal semiblocks in $\psi$. (See Figure~\ref{fig:operations}.)
\end{description}
Several of the above operations \emph{transform} a contig $\phi$ into a new contig $\psi$.  This means that the physical structure of $\phi$ is changed to obtain $\psi$, thus $\phi$ ceases to exist as such.  The semiblock pointers associated to $\phi$ are not invalidated, though; instead, they get associated to $\psi$.  We remark, also, that {\tt split($B_1, B_2$)} reverses the effects of {\tt join($B_1$, $B_2$)}, but the converse is not true in general.  The subtle problem is that {\tt split} outputs co-contigs, and co-contigs involve two contigs when they represents disconnected graphs (see Figure~\ref{fig:co-contigs} ($\phi_1$) and Figure~\ref{fig:operations} ($\phi_i$)).  We deal with this problem in Section~\ref{sec:data structure:round representations}.

\begin{figure}
  \begin{tabular}{c@{\hspace{3.5mm}}c@{\hspace{3.5mm}}c@{\hspace{3.5mm}}c}
    \input{srcFigCompact} & \input{srcFigSeparate} & \input{srcFigJoin} & \input{srcFigSplit}\\
    $\gamma$ & $\phi$ & $\psi$ & $\phi_i$
  \end{tabular}
  \caption{Effects of some operations on contigs: $\gamma$ equals {\tt compact($\DS{B}$)} for $B \in \B(\phi)$; $\phi$ equals {\tt separate($\DS{B}$, $W$)} for $B \in \B(\gamma)$; $\psi$ equals {\tt join($\DS{B}_1$, $\DS{B}_2$)} for $B_1 \in \phi_1$ and $B_2 \in \phi_2$; and $\{\phi_i,\phi_2\}$ equals {\tt split($\DS{B}_1$, $\DS{B}_2$)} for $B_1, B_2 \in \psi$.  Dashed lines represent non-edges, while missing lines represent edges that could be present of absent.  Note that some blocks may be equal, while $B_3$ (resp.\ $B_4$) does not exist when $B_1$ (resp.\ $B_2$) is universal.}\label{fig:operations}
\end{figure}

\subsection{Semiblock paths}

Together with each contig $\phi$, the dynamic algorithm keeps a path of $\phi$ that is used for the generation of negative witnesses.  Say that a semiblock $B \in \B(\phi)$ is \emph{long} when $F_r(B) \to F_l(B)$; those semiblocks that are not long are referred to as \emph{short}.  A \emph{semiblock path} $\MP$ of $\phi$ is an ordering $B_1, \ldots, B_k$ of a subset of $\B(\phi)$ such that:
\begin{itemize}
 \item $B_i \to B_{i+1}$ and $B_i \nto B_j$ for every $1 \leq i <  k$ and $j \neq i+1$.
 \item If $\phi$ is linear, then $\MP$ is linear and $B_1$ and $B_k$ are the end semiblocks of $\phi$.
 \item If $\phi$ is circular, then $\MP$ is circular and $B_k \to B_1$ and $B_k \nto B_2$.
 \item If $\phi$ contains a long semiblock, then $k = 3$.
\end{itemize}
Each semiblock path $\MP(\phi) = B_1, \ldots, B_k$ is implemented with a doubly linked list of semiblock pointers $\DS{B}_1, \ldots, \DS{B}_k$; the list is circular if and only if $\MP(\phi)$ is circular.  Conversely, each semiblock $B \in \B(\phi)$ has a \emph{path pointer} to the position that $\DS{B}$ occupies in $\MP(\phi)$; this is a null value when $B \not\in \MP(\phi)$.  Thus, $O(1)$ time is enough to detect if $B$ belongs $\MP(\phi)$, to access and remove $\DS{B}$ from $\MP(\phi)$, and to insert new semiblock pointers in $\MP(\phi)$ to the left or the right of $\DS{B}$.  

We now show how to efficiently update $\MP(\phi)$ (and the path pointers of $\phi$) each time a contig $\phi$ is updated.  As discussed in the previous section, there are eight operations that change the structure of a contig: {\tt newContig}, {\tt leftSeparate}, {\tt rightSeparate}, {\tt compact}, {\tt join}, {\tt split}, {\tt reception}, and {\tt remove}.  Updating $\MP(\phi)$ in $O(1)$ time is trivial for the first four operations.  Regarding {\tt join} and {\tt split}, note that the input and output contigs represent co-bipartite graphs.  Thus, the semiblock paths of the input contigs can be erased in $O(1)$ time because they have at most four semiblock pointers.  After the  semiblock paths are erased, we build the new semiblock paths from scratch as in the next lemma. 

\begin{lemma}
  Let $B_l$ be a left co-end semiblock of a contig $\phi$.  If a semiblock pointer to $B_l$ is given, then a semiblock path can be computed in $O(1)$ time.
\end{lemma}

\begin{proof}
  First suppose $\phi$ is a linear contig and observe that, in this case, $W_l = F_l^3(B)$ is the left end semiblock of $\phi$.  Indeed, if $W_l$ is not a left end semiblock, then $F_l(W_l)$, $F_l^2(B_l)$, and $B_l$ are pairwise non-adjacent semiblocks, which contradicts the fact that $G(\phi)$ is co-bipartite.  Similarly, $F_r^3(W_l)$ is the right end semiblock of $\phi$.  Thus, $F_r^0(W_l), \ldots, F_r^i(W_l)$ is a semiblock path, where $0 \leq i \leq 3$ is the minimum such that $F_r^i(W_l)$ is a right co-end semiblock.  Clearly, this semiblock path can be computed in $O(1)$ time.

  Now suppose $\phi$ is circular and let $B_r = F_r(B_l)$, $W_l = U_r(B_l)$ and $W_r = F_r(W_l)$.  Recall that $W_l$ is a left co-end block while $B_r$ and $W_r$ are right co-end blocks.  If $F_l(B_l) = W_l$, then $B_l$ is long and $B_l$, $B_r$, $W_l$ is a semiblock path.  Otherwise, $[B_l, B_r]$ and $[W_l, W_r]$ is a partition of $\B(\phi)$.  Moreover, $B_r \to W_l$ and $W_r \to B_l$ because $\phi$ is circular.  We consider two cases:
  \begin{description}
    \item [Case 1:] $\phi$ contains a semiblock path $B_1$, $B_2$, $B_3$.  Note that at least one of $B_1$, $B_2$, $B_3$ belongs to $[B_l, B_r]$ (resp.\ $[W_l, W_r]$); say $B_1 \in [B_l, B_r]$ and $B_3 \in [W_l, W_r]$.  If $B_2 \in [B_l, B_r]$, then $B_r \to B_3$ and $B_3 \to B_l$, which implies that $B_l$, $B_r$, $F_r(B_r)$ is a semiblock path.  Similarly, if $B_2 \in [W_l, W_r]$, then $W_l$, $W_r$, $F_r(W_r)$ is a semiblock path.
    \item [Case 2:] $\phi$ contains no semiblock $B$ such that $F_r(B) \to F_l(B)$.  In this case, $B_l$, $B_r$, $W_l$, $W_r$ is a semiblock path.
  \end{description}
  Note that $F_r(B_r) \toeq W_r$, thus $B_l$, $B_r$, $F_r(B_r)$ is a semiblock path if and only if $F_r^2(B_r) \neq W_r$.  Similarly, $W_l$, $W_r$, $F_r(W_r)$ is a semiblock path if and only if $F_r^2(W_r) \neq B_r$.  By Cases 1~and~2, we can compute a semiblock path in $O(1)$ time.
\end{proof}

In the case of {\tt reception($\DS{B}_l$, $\DS{B}_r$)} we have to modify both $\MP(\phi_1)$ and $\MP(\phi_2)$ to obtain $\MP(\psi)$, where $\phi_1$ and $\phi_2$ are the contigs referenced by $\DS{B}_l$ and $\DS{B}_r$, and $\psi$ is the $\{v\}$-reception of $\Bip{B_l}{B_r}$.  This update is applied after {\tt reception} is completed, thus we have access to a semiblock pointer of $\{v\}$.  The following lemma shows how to obtain $\MP(\psi)$ spending no more time than the required for {\tt reception}.

\begin{lemma}
  Let $\phi_1$ and $\phi_2$ be two (possibly equal) contigs such that $\{\phi_1,\phi_2\}$ contains a receptive pair $\Bip{B_l}{B_r}$ for $B_l \in \B(\phi_1)$ and $B_r \in \B(\phi_2)$, and $\psi$ be the $\{v\}$-reception of $\Bip{B_l}{B_r}$ in $\{\phi_1,\phi_2\}$.  Given a semiblock pointer to $B = \{v\}$ in $\psi$, it takes $O(\#[B_l, B_r])$ time to transform the semiblocks paths $\MP(\phi_1)$ and $\MP(\phi_2)$ into a semiblock path of $\psi$. 
\end{lemma}

\begin{proof}
  Recall that, by definition, $N[B] = [B_l, B_r]$ and $\{\phi_1, \phi_2\} = \{\psi\} \setminus \{B\}$.   Consider the following alternatives.
  \begin{description}
    \item [Alternative 1:] $B$ is an end semiblock of $\psi$.  Suppose $B$ is the left end semiblock as the other case is analogous.  By definition, $\phi_1 = \phi_2$, $B_l = R(B)$ is the left end semiblock of $\phi_1$, and $B_r = F_r(B)$.  Traversing $[B_l, B_r]$ in $\psi$, we can check if $B \to B_2$ for the semiblock $B_2$ that follows $B_l$ in $\MP(\phi_1)$.  If affirmative, then a semiblock path of $\psi$ is obtained by replacing $B_l$ with $B$ in $\MP(\phi_1)$; otherwise, a semiblock path of $\psi$ is obtained by inserting $B$ before $B_1$ in $\MP(\phi_1)$. 
    \item [Alternative 2:] $B$ is not an end block of $\psi$, thus $B \in (B_l, B_r)$ in $\psi$, $B_l = F_l(B)$, and $B_r = F_r(B)$.  Traversing $[B_l, B)$, we can check if $F_r(B_r) \in [B_l, B)$.  If affirmative, then $B$ is long and $B$, $B_r$, $B_l$ is a semiblock path.  Suppose, from now on, that $B_r \nto B_l$.  By traversing $[B_l, B)$ we can find the leftmost semiblock $B_a$ and the rightmost semiblock $B_b$ of $\MP(\phi_1)$ such that $B_a \to B$ and $B_b \to B$ (possibly $B_a = B_b$).  Similarly, we can obtain the leftmost semiblock $B_c$ and the rightmost semiblock $B_d$ of $\MP(\phi_2)$ such that $B \to B_c$ and $B \to B_d$.  Note that these semiblocks exist because $\psi$ is a contig.  If $F_r^\psi(B_b) \neq B$, then $B_b$ is not an end semiblock of $\phi_1$, thus $\phi_1 = \phi_2$ and $B_b \to B_c$; consequently, $\MP(\phi_1)$ is a semiblock path of $\MP(\psi)$.  Otherwise, $B_b$ is the right end semiblock of $\phi_1$ and $B_c$ is the leftmost end semiblock of $\phi_2$ (perhaps $\phi_1 = \phi_2$).  Then, the ordering obtained from $\MP(\phi_1)$ by inserting $B$ between $B_a$ and $B_d$ (removing $B_b$ if $B_a \neq B_b$ and $B_c$ if $B_c \neq B_d$) is a semiblock path of $\psi$.
  \end{description}
  Using the path pointers, we can apply all the modifications required on $\MP(\phi_1)$ and $\MP(\phi_2)$ in $O(1)$ time. We conclude, therefore, that $O(\#[B_l, B_r])$ time suffices to transform $\MP(\phi_1)$ and $\MP(\phi_2)$ into a semiblock path of $\psi$.
\end{proof}

Finally, to update $\MP(\psi)$ after {\tt remove($\DS{B}$)}, where $\psi$ is the contig referenced by $\DS{B}$, we have to revert the process done in the previous lemma.  After the completion of {\tt remove}, we obtain at most three contigs, namely $\phi_1$, $\phi_2$, and $\gamma$, where $\phi_1$ is referenced by $L(\DS{B})$, $\phi_2$ is referenced by $R(\DS{B})$, and $\gamma$ is referenced by $\DS{B}$.  The computation of $\MP(\gamma)$ is trivial; to compute $\MP(\phi_1)$ and $\MP(\phi_2)$ we apply the following lemma \textbf{before} invoking {\tt remove}.

\begin{lemma}\label{lem:semiblock path remove}
  Let $B$ be a semiblock of a contig $\psi$.  Given a semiblock pointer to $B$, it takes $O(\#[F_l(B), F_r(B)])$ time to transform $\MP(\psi)$ into $\{\MP(\phi_1), \MP(\phi_2)\}$, where $\{\phi_1, \phi_2\}$ is the family of contigs of $\{\psi\} \setminus \{B\}$.
\end{lemma}

\begin{proof}
  If $B \not\in \MP(\psi)$, then $\phi_1 = \phi_2$ and $\MP(\psi)$ is a semiblock path of $\phi_1$.  Suppose, then, that $B \in \MP(\psi)$ and consider the following alternatives.
  \begin{description}
    \item [Alternative 1:] $B$ is an end semiblock of $\psi$. Suppose $B$ is a left end semiblock, as the other case is analogous, and note that $\phi_1 = \phi_2$ and $B_l = R(B)$ is the leftmost end semiblock of $\phi_1$.  Let $\MP$ be the ordering obtained by replacing $B$ with $B_l$ in $\MP(\phi_1)$.  In $O(1)$ time we can obtain the first two semiblocks $B_2$ and $B_3$ that follow $B$ in $\MP(\phi_1)$.  If $F_r(B_l) \neq B_3$, then $\MP$ is a semiblock ordering of $\phi_1$; otherwise, the ordering obtained by removing $B_2$ from $\MP(\phi_1)$ is a semiblock ordering of $\phi_1$.
    \item [Alternative 2:] $B$ is not an end semiblock of $\psi$, thus $B \in (B_l, B_r)$ for $B_l = F_l(B)$ and $B_r = F_r(B)$.  First we search if $\phi_1$ has a long semiblock.  This happens only if $|\MP(\psi)| = 3$, in which case $\phi_1 = \phi_2$, $B_r \to B_l$, and, thus, $F_r^2(W) \in [B_l, W]$ for every $W \in (B, B_r]$.  Marking the position of every semiblock in $[B_l, B_r]$, we can check in $O(1)$ time if $(F_r^{\phi_1})^2(W)$ appears after $F_l(W)$ in $[B_l, B_r]$ for some $W \in (B, B_r]$.  If affirmative, then $W$ is long and $F_l(W)$, $W$, $F_r(W)$ is a semiblock path of $\phi_1$; otherwise, $\phi_1$ has no long semiblocks.  When $\phi_1$ has no long semiblocks, we traverse $[B_l, B_r]$ to check if 1.~$B_a \to R(B)$, 2.~$L(B) \to B_b$, 3.~$R(B) \to B_{b+1}$, and $B_{a-1} \to L(B)$, where $B_{a-1}$, $B_a$, $B_b$, and $B_{b+1}$ are the semiblocks of $\MP(\psi)$ such that $B_{a-1} \to B_a$, $B_a \to B$, $B \to B_b$, and $B_b \to B_{b+1}$ (unless $L(B_a) = \bot$ or $R(B_{b}) = \bot$ in which case $B_{a-1} = \bot$ and $B_{b+1} = \bot$, respectively).  Replacing $B$ with $R(B)$ if 1.\ and removing $B_b$ if 3., or replacing $B$ with $L(B)$ if 2.~and removing $B_a$ if 4., we obtain a semiblock path of $\phi_1 = \phi_2$.  If neither 1.~nor 2.~holds, then we transform $\MP(\psi)$ into the ordering $\MP$ that is obtained by first replacing $B$ with $L(B), R(B)$ and then removing $B_b$ if 3.\ and $B_a$ if 4.  If $L(B) \to R(B)$ or $\MP(\psi)$ is circular, then $\MP$ is a semiblock path of $\phi_1 = \phi_2$; otherwise, $\phi_1 \neq \phi_2$, thus we split $\MP$ into the suborderings that have $L(B)$ as rightmost and $R(B)$ as leftmost to obtain semiblock paths of $\phi_1$ and $\phi_2$, respectively.
  \end{description}
\end{proof}

\subsection{Round representations}
\label{sec:data structure:round representations}

To implement a round representation $\Phi$ we use a pair of doubly linked list $\{\DS{\Phi}, \DS{\Phi}^{-1}\}$ and a \emph{connectivity} structure (see below).  For each $\phi \in \Phi$, a semiblock pointer associated to $\phi$ (resp.\ $\phi^{-1}$) is kept in $\DS{\Phi}$ (resp.\ $\DS{\Phi}^{-1}$).  (The semiblock $\DS{B}$ of $\phi$ plays the same role as the pointer to the first node in a linked lists when implementing the abstract data type; that is, $\DS{B}$ is used to access $\phi$.)  Thus, both \emph{physical} contigs $\phi$ and $\phi^{-1}$ are stored for each contig $\phi \in \Phi$.  The reason why $\phi^{-1}$ is kept is to avoid the cost of reversing $\phi$.  If $\DS{B} \in \DS{\Phi} \cup \DS{\Phi}^{-1}$ is associated to $\phi$, then $B \in \MP(\phi)$; moreover, $B$ is the left end semiblock of $\phi$ when $\phi$ is linear.  Conversely, $B$ keeps a \emph{contig pointer} to the position of $\DS{B}$ inside $\DS{\Phi} \cup \DS{\Phi}^{-1}$.  The contig pointer is used, among other things, to remove $\DS{B}$ from $\DS{\Phi} \cup \DS{\Phi}^{-1}$ when its associated contig is joined to some other contig.  Of course, this pointer has a null value when $B$ is not referenced by a pointer in $\DS{\Phi} \cup \DS{\Phi}^{-1}$.  Finally, each vertex $v$ of a contig $\phi \in \Phi$ keeps a \emph{reverse} pointer to its incarnation in $\phi^{-1}$.

Recall that all the contigs of a round representation $\Phi$ are linear when $|\Phi| > 1$; this invariant must be satisfied by the data structure. Thus, we need some way to detect if an operation on a linear contig $\phi$ yields a circular contig $\psi$.  Actually, the only operation in which we are ignorant about the linearity of $\psi$ is when we compute the $\{v\}$-reception of $\Bip{B_l}{B_r}$.  As it is shown in~\cite{HellShamirSharanSJC2001}, the only possibility for $\psi$ to be circular is when $[B_l, B_r)$ contains the right end semiblock of $\phi$.  To detect this case, we need to know if two end semiblocks belong to the same contig.  As it was proved in~\cite{HellShamirSharanSJC2001}, the connectivity problem is not solvable in $O(1)$ time when both insertions and removal of semiblocks are allowed.  Thus, a \emph{connectivity} data structure is kept to solve this problem.  Its interface provides the following operations:
\begin{description}
  \item [\tt create()] returns an empty connectivity structure. Complexity: $O(1)$ time.
  \item [\tt add($\DS{B}$)] adds $\DS{B}$ to the connectivity structure. Complexity: $O(\epAdd)$ time.
  \item [\tt remove($\DS{B}$)] removes $\DS{B}$ from the connectivity structure. Complexity: $O(\epRemove$) time.
  \item [\tt opposite($\DS{B}$)] returns a pointer the other end semiblock of the contig that contains $B$ if $B$ is an end semiblock, or $\bot$ if $B$ is not an end semiblock.  Complexity: $O(\epTest)$ time.
\end{description}
There are three flavors of the structure according to the operations supported by the main algorithm.  In the \emph{incremental} structure $O(\epTest) = O(1)$ and $O(\epRemove) = O(n)$, in the \emph{decremental} structure $O(\epTest) = O(n)$ and $O(\epRemove) = O(1)$, and in the \emph{fully dynamic} structure $O(\epAdd) = O(\epRemove) = O(\log n)$~\cite{HellShamirSharanSJC2001,SoulignacA2015}.

Let $B_1$ and $B_2$ be semiblocks of a round representation $\Phi$ that belong to the contigs $\phi_1$ and $\phi_2$ of $\Phi$, respectively.  To traverse the range $[B_1, B_2]$ of $\{\phi_1, \phi_2\}$, we have to provide a semiblock pointer $\DS{B}_i$ to $B_i$ associated to $\phi_i$, for $i \in \{1,2\}$.  If, in turn, the semiblock pointer $\DS{B}_1$ is associated to $\phi_1^{-1}$, we obtain the range $[B_1, B_2]$ of $\{\phi_1^{-1}, \phi_2\}$.  Thus, to describe the effects of an algorithm, we must specify the contig to which a pointer is associated.  We say that $\DS{B}$ has \emph{type $\Phi$} when $\DS{B}$ is associated to some contig $\phi \in \Phi$.  Those pointers of type $\Phi$ are sometimes referred to as \emph{$\Phi$-pointers}.  Of course, every semiblock pointer has type either $\Phi$ of $\Phi^{-1}$.  Each semiblock pointer to $B$ of type $\Phi^{-1}$ is called a \emph{reverse} of $\DS{B}$.  We recall that there is no efficient way to know the type of semiblock pointer and, in general, the type is not important in the implementation.  The purpose of this terminology is to aid in the specification of the different operations.

To avoid dealing with the pointers of round representations, we implement several operations that define an interface similar to the one used for contigs.  As usual, we use capital Greek letters for round representations, capital Roman letters for semiblocks, and tildes for pointers.  
\begin{description}
  \item [\tt straight($\Phi$)] returns true if $\Phi$ is straight.  For the implementation we test if a pointer in $\DS{\Phi}$ references an end semiblock.  Complexity: $O(1)$ time.
  \item [\tt newContig($\Phi$)] adds to $\Phi$ a new contig whose only block is $B = \{v\}$, and returns $\DS{v}$.  Requires $\Phi$ to be straight.  For the implementation, we add a new physical contig to both $\DS{\Phi}$ and $\DS{\Phi}^{-1}$.  Complexity: $O(1)$ time.
  \item [\tt reversed($\DS{B}$)] returns a reverse of $\DS{B}$.  For the implementation, we call {\tt semiblock($\DS{w}$)} where $\DS{w}$ is the reverse pointer of any $v \in B$.  Complexity: $O(1)$ time.
  \item [\tt type($\DS{B}$)] returns the type of $\DS{B}$, i.e., a pointer to either $\Phi$ or $\Phi^{-1}$.  Requires $B$ to be an end semiblock.  Let $\phi$ be the contig referenced by $\DS{B}$.  For the implementation we access the representation pointer of the physical semiblock referenced by either $\DS{B}$ (if $B$ is a left end semiblock in $\phi$) or a reverse of $\DS{B}$ (if $B$ is a right end semiblock in $\phi^{-1}$).  Complexity: $O(1)$ time.
  %
  %
  \item [\tt reverse($\DS{B}$)] reverses the contig of $\Phi$ that contains $B$.  Requires $B$ to be an end semiblock.  Applying {\tt reversed} if required, assume $B$ is a left end semiblock of a contig $\phi$.  Moreover, suppose, w.l.o.g., that $\DS{B} \in \DS{\Phi}$, as the other case is analogous.  For the implementation, we use the contig and round pointers of $B$ to move $\DS{B}$ from $\DS{\Phi}$ to $\DS{\Phi}^{-1}$.  Then, we use the connectivity structure together with {\tt reversed} to obtain a $\Phi^{-1}$-pointer $\DS{W}$ to the right end semiblock of $\phi$.  Finally, we move $\DS{W}$ from $\DS{\Phi}^{-1}$ to $\DS{\Phi}$.  Complexity: $O(\epTest)$ time.
  \item [\tt receptive($\DS{B}_l$, $\DS{B}_r$)] takes two $\Phi$-pointers and returns true if $\Bip{B_l}{B_r}$ is receptive in $\Phi$.  Let $\phi_1$ and $\phi_2$ be the contigs referenced by $\DS{B}_l$ and $\DS{B}_r$, respectively.  For the implementation, observe that $\Bip{B_l}{B_r}$ is receptive if and only if $\Bip{B_l}{B_r}$ is receptive in $\{\phi_1, \phi_2\}$ and either the $\{v\}$-reception $\psi$ of $\Bip{B_l}{B_r}$ is linear or $\phi_1=\phi_2$ is the unique contig in $\Phi$.  Note that $\psi$ is linear if and only if $[B_l, B_r)$ has no right end semiblocks or $W$ and $\RR(W)$ lie in different contigs, where $W \in [B_l, B_r)$ a right end semiblock.  If $W \in [B_l, B_r)$ is an end semiblock, then we can check if $\phi_1$ is the unique contig in $\Phi$ using the representation pointer of $\RR(W)$.  To check if $W$ and $\RR(W)$ lie in different contigs, a query to the connectivity structure is required.  Complexity $O(\#[B_l, B_r] + \epTest)$ time.
  \item [\tt reception($\DS{B}_l$, $\DS{B}_r$)] takes two $\Phi$-pointers and updates $\Phi$ into the $\{v\}$-reception $\Psi$ of $\Bip{B_l}{B_r}$.  Requires $\Bip{B_l}{B_r}$ to be receptive in $\Phi$, and returns a pointer to $v$.  For the implementation, suppose $\phi_1$ and $\phi_2$ are the contigs referenced by $\DS{B}_l$ and $\DS{B}_r$, respectively.  The first step is to apply {\tt reception($\DS{B}_l$, $\DS{B}_r$)} to transform $\phi_1$ and $\phi_2$ into a contig $\psi$ that represents the $v$-reception of $\{\phi_1, \phi_2\}$.  Let $W_l = R^\psi(\{v\})$, and observe that $W_l$ is a left end block of $\Phi$ if and only $F_l(W_l) = \{v\}$.  In this case, we access the contig pointer of $W_l$ to remove it from $\DS{\Phi}$.  Similarly, if $\{v\}$ is a left end semiblock of $\Psi$, then we add a semiblock pointer to $\{v\}$ in $\Phi$ and we update the contig pointer of $\{v\}$.  Finally, we test if $\DS{\Phi} = \emptyset$.  This happens when $\phi_1 = \phi_2$ is linear and $\psi$ is circular, in which case $\{v\}$ belongs to $\MP(\psi)$.  Thus, again, we add a pointer to $\{v\}$ in $\DS{\Phi}$.  Once the update is completed, we apply the same procedure to {\tt reversed($\DS{B}_r$)} and {\tt reversed($\DS{B}_l$)} to update $\DS{\Phi}^{-1}$.  Finally, we update the connectivity structure.  Complexity: $O(\#[B_l, B_r] + \epAdd)$ time.  
  \item [\tt remove($\Psi$, $\DS{B}$)] transforms $\Psi$ into a round representation $\Phi$ of $G(\Psi) \setminus B$.  Requires $\DS{B}$ to be of type $\Psi' \in \{\Psi, \Psi^{-1}\}$.  For the implementation, we first call {\tt remove($\DS{B}$)} to transform the contig $\psi \in \Psi'$ referenced by $\DS{B}$ into two (possible equal) physical contigs $\phi_1$ and $\phi_2$, where $\phi_1$ contains $\LL(B)$.  The next step depends on whether the contig pointer of $B$ is null or not.  In the latter case, we have to replace $B$ in $\DS{\Psi}'$ with a pointer to a semiblock $W_l$ in $\MP(\phi_2)$.  Note that $W_l$ must be a left end semiblock if $\phi_2$ is linear; such a case occurs only when $W_l = R(B)$ is a left end semiblock of $\phi_2$.  In the former case, we check if $W_l = R^\psi(B)$ is a left end semiblock.  If affirmative, then there are two possibilities.  If $\Psi = \{\psi\}$ is not straight, then we replace the unique pointer in $\DS{\Psi}$ with a pointer to $W_l$.  If negative, then we insert $W_l$ to $\DS{\Psi}$.  After completion, we apply the same transformations to $\DS{\Psi}^{-1}$ using a reverse of $\DS{B}$.  Finally, we remove $B$ from the connectivity structure.  We remark that the obtained round representation $\Phi$ is not necessarily equal to $\Psi \setminus \{B\} = (\Psi \setminus \{\psi\}) \cup \{\phi_1, \phi_2\}$.  The reason is that we are not aware of the type of $\DS{B}$.  Thus, if $\DS{B}$ has type $\Psi^{-1}$, then when we decide to insert $R^\psi(B)$ to $\DS{\Psi}$, we are actually computing $\Psi \setminus \{\psi\} \cup \{\phi_1, \phi_2^{-1}\}$.  Complexity: $O(\#[F_l(B), F_r(B)] + \epRemove)$. 
  \item [{\tt separate($\DS{B}$, $W$)}, {\tt separate($W$, $\DS{B}$)}, and {\tt compact($\DS{B}$)}] have the same effects as their contig versions on the contig of $\Phi$ that contains $B$.  For the implementation, we apply the corresponding operations on $\phi$ and $\phi^{-1}$, where $\phi$ is the contig referenced by $\DS{B}$.  Also, we take care of the contig and round pointers when $B$ is an end semiblock.  The details are similar to those described for the previous operations.  Complexity: $O(|W|)$ time for {\tt separate} and $O(\min\{|B|, |R(B)|\}$ time for {\tt compact}.
  \item [\tt join($\DS{B}_1$, $\DS{B}_2$)] takes a $\Phi_1$-pointer $B_1$ and a $\Phi_2$-pointer $B_2$ and builds the round representation $\Psi = \{\psi\}$ such that $\psi$ satisfies the same specifications as the contig version of {\tt join}.  Requires $B_1$ and $B_2$ to be left co-end blocks and $\Phi_1 \neq \Phi_2$.  The differences between this version of {\tt join} and the one for contigs are the following.  First, $\Phi_i$ could be a disconnected co-contig for $i \in \{1,2\}$.  In this case, $\Phi_i$ has only two contigs, each of which represents a clique.  Second, the output $\Psi$ is implemented as a round representation.  To compute {\tt join}, we apply one, two, or three calls to the contig version of {\tt join}, according to whether $\Phi_1$ and $\Phi_2$ are disconnected or not.  Since $\psi$ is connected, it takes $O(1)$ time to restore all the pointers required by the data structure of $\Psi$.  Finally, the connectivity structure can be updated in $O(\min\{\epAdd,\epRemove\})$ time as discussed in~\cite{SoulignacA2015}.  Time complexity: $O(u + \min\{\epAdd,\epRemove\})$ time, where $u$ is the number of universal semiblocks in $\Psi$.  
  \item [\tt split($B_1$, $B_2$)] has the same effects as its contig version, but it returns the co-contigs implemented as round representations.  Requires $B_1$ and $B_2$ to be left end blocks of the same type.  Note that $\Psi$ has a unique contig $\{\psi\}$, thus generating the pointer of the output representations is trivial.  The connectivity structure can be updated in $O(\min\{\epAdd,\epRemove\})$ time as well~\cite{SoulignacA2015}.  Complexity: $O(u + \min\{\epAdd,\epRemove\})$ time, where $u$ is the number of universal semiblocks in $\Psi$.
\end{description}

\subsection{The witnesses}
\label{sec:data structure:witnesses}

From the point of view of the end user, the dynamic algorithm keeps a round block representation $\Gamma$ of the dynamic graph $G(\Gamma)$.  To work with $\Gamma$, users can iterate through the semiblock pointers associated to contigs in $\Gamma$, while they execute the operations described in Section~\ref{sec:data structure:contig}.  However, only those operations that do not modify the internal structure of contigs are available, e.g., {\tt vertices}, $L$, $F_l$, etc.  To update $G(\Gamma)$, one of the following operations is applied.
\begin{description}
  \item [\tt create()] returns an empty round block representation $\Gamma$. Complexity: $O(1)$ time.
  \item [\tt insert($\Gamma$, $N$)] transforms $\Gamma$ into a round block representation $\Psi$ of a graph $H$ such that $H \setminus \{v\} = G(\Gamma)$ for some vertex $v \not\in V(\Gamma)$ with $N(v) = N$.  Returns the new vertex $v$.  The operation fails if $H$ is not PCA and, in this case, a \emph{minimally forbidden} of $H$ is obtained (see below).  Complexity: $O(|N| + \epTest)$ time.
  \item [\tt remove($v$)] transforms $\Gamma$ into a round block representation $\Phi$ of $G(\Gamma) \setminus \{v\}$, where $\Gamma$ is the round block representation containing $v$.  Complexity: $O(d(v) + \epRemove)$ time.
\end{description}

Our goal is not only to implement the above operations that deal with PCA graphs, but also to provide a certifying algorithm for the recognition of PIG graphs.  With respect to the positive witness, the latter problem is solved by satisfying the \emph{straightness invariant} that guarantees that every contig of $\Gamma$ is linear when $G(\Gamma)$ is a PIG graph.  Regarding the negative certificate, we implement the following operation.
\begin{description}
  \item [\tt forbiddenPIG($\Gamma$)] returns a minimally forbidden witnessing that $G(\Gamma)$ is not PIG (i.e., a structure that represents a graph of Theorem~\ref{thm:forbiddens PIG}).  Complexity: $O(1)$ time. 
\end{description}

When {\tt insert($\Gamma$, $N(v)$)} is executed and $H$ is not PCA, for $H \setminus \{v\} = G(\Gamma)$, the end user obtains a negative witness.  We say that a pair $\Bip{\Phi}{\N}$ is a \emph{forbidden} of $H$ (w.r.t. $\Gamma$) when:
\begin{itemize}
  \item $\Phi$ is a round block representation of an induced subgraph of $H$,
  \item $v$ is fully adjacent to every block in $\N$ and not adjacent to every block outside $\N$, and
  \item $H' = H[V(\Phi) \cup \{v\}]$ is not PCA.
\end{itemize}
If every subgraph of $H'$ obtained by removing $B \in \B(\Phi)$ is PCA, then $\Bip{\Phi}{\N}$ is a \emph{minimally} forbidden of $H$.

In case of failure, the output of {\tt insert} is a minimally forbidden $\Bip{\Phi}{\N}$.  To be useful to the end user, $\Bip{\Phi}{\N}$ has to be as efficient as possible.  The least a user can expect is that $L^\Phi$, $R^\Phi$, $F_l^\Phi$, and $F_r^\Phi$ take constant time.  This allows the user to traverse the corresponding forbidden graph in $O(1)$ time per edge, and to take advantage of the PCA structure of $G(\Phi)$.  Therefore, $\Phi$ is implemented with a data structure that satisfies these time bounds.  As a consequence, finding a minimal $\B \subseteq \B(\Gamma)$ such that $H[\B \cup \{v\}]$ is not PCA is not enough.  We also have to find the near and far neighbors of $\Gamma|\B$, and decide which vertices of $B \in \B$ survive when $v$ is both adjacent and co-adjacent to $B$.  As it is expected due to the time bounds, $\Phi$ shares some of the internal structure of $\Gamma$ and, consequently, $\Phi$ must be discarded (or copied) before applying further operations on $\Gamma$.

\section{An incremental and certified algorithm}
\label{sec:incremental}

This section is devoted to the implementation of {\tt insert} (Section~\ref{sec:data structure:witnesses}), whose aim is to insert a vertex $v \in V(H)$ into a round block representation $\Gamma$ of $G = H \setminus \{v\}$ in $O(d(v)+\epAdd)$ time.  An algorithm for this problem was given in~\cite{SoulignacA2015}, and the method we present takes advantage of the tools developed in~\cite{SoulignacA2015}.  However, the algorithm in~\cite{SoulignacA2015} is unable to output a minimally forbidden (or any witness whatsoever) when the insertion fails.  The purpose of this section is to complete the algorithm by providing the negative witness.  To show that our algorithm is correct, we prove of the Reception Theorem, following the same path as Deng et al.~\cite{DengHellHuangSJC1996}.  That is, we show a minimally forbidden of $H$ when \textbf{no} round representation of $G$ is $v$-receptive.

Because of the $O(d(v)+\epAdd)$ time bound, we face two major inconveniences.  First, we cannot traverse all the blocks of $\Gamma$.  Thus, it is impossible to determine whether $B \to W$ (in $O(d(v)+\epAdd)$ time) when $B$ and $W$ are arbitrary blocks.  This means that we need to infer some of the adjacencies by making appropriate queries on $\Gamma$.  For this reason, in this section we assume that $\Gamma$, and every round representation obtained by transforming $\Gamma$, have been preprocessed as in the next observation, even if we are not explicit about this fact.  This allows us to answer basic adjacencies queries as in Observation~\ref{obs:basic queries}.

\begin{observation}[see e.g.~\cite{SoulignacA2015}]\label{obs:preprocessing}
  Let $H$ be a graph and $\Phi$ be a round representation of $H \setminus \{v\}$ for some $v \in V(H)$.  Given $N(v)$ as input, it is possible to preprocess the semiblocks of $\Phi$ in $O(d(v))$ time so that determining whether $v$ is (fully) adjacent to $B$ can be answered in $O(1)$ time for any $B \in \B(\Phi)$ when a semiblock pointer to $B$ is given.
\end{observation}

\begin{observation}\label{obs:basic queries}
  Let $H$ be a graph, $\Phi$ be a round representation of $H \setminus \{v\}$ for some $v \in V(H)$, and $W_l$, $W_r$ be (possibly equal) semiblocks of $\Phi$.  Given semiblock $\Phi$-pointers to $W_l, W_r$, the following problems can be solved in $O(\#[W_l, W_r])$ time:
  \begin{enumerate}[(a)]
    \item obtain a $\Phi$-pointer to the leftmost (resp.\ rightmost) semiblock of $(W_l, W_r)$ co-adjacent to $v$.\label{obs:basic queries:first co-adjacent}
    \item determine whether $W_l \to W$ and $W \to W_l$ (resp.\ $W \to W_r$ and $W_r \to W$), when a semiblock $\Phi$-pointer to $W \in [W_r, F_r(W_r)]$ (resp. $W \in [F_l(W_l), W_l]$) is given.\label{obs:basic queries:triangle adjacency}
  \end{enumerate}
\end{observation}

The second inconvenience that arises when we want to compute a minimally forbidden, is that doing so requires a heavy amount of case by case analysis.  The case by case analysis is somehow inherent to these kinds of proofs, as we need to proceed differently according to whether some edge exists or not (and the existence of such an edge may or may not imply the existence of other edges).  To alleviate this situation, we make use of \emph{adequate} forbidden.  We say a family $\B$ of semiblocks of $H$ is \emph{forbidden} when $H[\B]$ is not PCA; $\B$ is \emph{adequate} (with respect to a round representation $\Phi$) when all the adjacencies between the semiblocks in $\B$ can be computed in $O(d(v))$ time when $N(v)$ and $\Phi$ are given as input.  Clearly, if $\B$ is an adequate forbidden, then a minimally forbidden of $H$ can be obtained in $O(d(v))$ time when $\B$, $N(v)$, and $\Phi$ are given.  To prove that $\B$ is an adequate forbidden family we still have to prove that $H'$ is not PCA.  This task, however, can done by a computer (see Appendix~\ref{app:adequacy}).

We divide our exposition in two major sections, according to whether $G$ and $H$ are co-connected or not.  In the remaining of this section, we always use $v$ to denote the vertex being inserted.  So, for $B \in \B(\Phi)$ we write $+B$ and $-B$ as shortcuts for $B \cap N(v)$ and $B \setminus N(v)$, respectively, and we write $\pm B$ to mean a nonempty semiblock in $\{+B,-B\}$.  Also, we write $v$ as a shortcut for $\{v\}$ when a semiblock of $H$ is expected.

\subsection{\texorpdfstring{Both $H$ and $G$ are co-connected}{Both {\it H} and {\it G} are co-connected}}

Throughout this section we consider that both $H$ and $G = H \setminus \{v\}$ are co-connected.  The advantage of this case is that $\Psi \cup \Psi^{-1} = \Gamma \cup \Gamma^{-1}$ for every pair of block co-contigs $\Psi$ and $\Gamma$ representing $G$~\cite{HuangJCTSB1995}.  By~\receptionref{1}, we obtain that $H$ is PCA only if the blocks fully adjacent to $v$ are consecutive in some block co-contig representing $G$, say $\Gamma$.  In this case, $\Gamma$ (and $\Gamma^{-1}$) can be associated to at most two co-contigs that simultaneously satisfy \receptionref{1} and \receptionref{2}.  Therefore, it suffices to consider only these $O(1)$ co-contigs to prove the Reception Theorem.  We first show how to obtain a minimally forbidden when $N(v)$ is consecutive in none of the block co-contigs representing $G$.  But, before dealing with the consecutiveness of $N(v)$, we solve a rather restricted case in which the input representation $\Gamma$ contains some ``bad'' blocks.  The existence of such ``bad'' blocks is what makes it hard to test whether two blocks of $\Gamma$ are adjacent.  Without this hurdle, we can answer more powerful adjacencies queries.

We say that a semiblock $B$ of a co-contig $\Phi$ is \emph{good} when $v$ is not adjacent to $B$ or $v$ is fully adjacent to all the semiblocks in either $[F_l(B), B)$ or $(B, F_r(B)]$.  If $B$ is good and $v$ is either fully or not adjacent to $B$, then $B$ is \emph{perfect}, while $B$ is \emph{bad} when it is not good.  It is not hard to see that $\Phi$ satisfies \receptionref{1} only if all its semiblocks are perfect.  For such a co-contig $\Phi$ to exist, all the blocks of $\Gamma$ must be good.  Lemma~\ref{lemma:bad block} shows how to obtain a minimally forbidden when some block in $\Gamma$ is bad.  We consider two prior cases in Lemmas \ref{lemma:bad block long}~and~\ref{lemma:bad block short} whose common parts appear in the next lemma.

\begin{lemma}\label{lemma:subprocedures}
  Let $H$ be a graph with a vertex $v$, $\Gamma$ be a block co-contig representing $H \setminus \{v\}$, and $T_1, T_2, T_3$ be blocks of $\Gamma$ such that $T_1 \to T_2$ and $T_2 \to T_3$.  Given semiblock $\Gamma$-pointers to $T_1$, $T_2$, and $T_3$, a minimally forbidden of $H$ can be obtained in $O(d(v))$ time when either of the following conditions holds.
  
  \begin{enumerate}[(a)]
    \item $v$ is co-adjacent to $T_1$ and $T_3$ and adjacent to every block in $(T_1, T_3)$, and $T_1 \nto T_3$.\label{lemma:subprocedures:left to right}
    \item $v$ is co-adjacent to $T_3$ and to $W \in [T_1, T_2]$, and $T_3 \to T_1$ \label{lemma:subprocedures:two triangle}
    \item $v$ is co-adjacent to $T_1$ and $T_3$ and adjacent to every block in $(T_1, T_3)$, $T_3 \to U_r(T_2)$, $T_1 \to T_3$, $U_l(T_2) \to T_1$, and $U_r(T_2) \to U_l(T_2)$ is obtainable in $O(d(v))$ time. \label{lemma:subprocedures:center not dominating}
    \item $v$ is co-adjacent to $T_1$ and $T_3$ and adjacent to every block in $(T_1, T_3)$, $T_1 \to T_3$, $F_r(T_1) = F_r(T_3)$, $U_l(T_1) \to F_l(T_1)$, and $F_r(T_1) \to U_l(T_1)$ is obtainable in $O(d(v))$ time.\label{lemma:subprocedures:center and right dominated}
  \end{enumerate}
\end{lemma}

\begin{proof}
  We provide $O(d(v))$ time algorithms to find an adequate forbidden in each case.  
  
  \begin{enumerate}[(a)]
    \item First, we query if $T_3 \to T_1$ as in Observation~\ref{obs:basic queries}~(\ref{obs:basic queries:triangle adjacency}) with input $W_l = T_1$, $W_r = T_2$, $W = T_3$; if false, then \{$+T_2$, $v$, $-T_1$, $-T_3$\} induces a $K_{1,3}$ of $H$.  Suppose, then, that $T_3 \to T_1$.  Let $a \geq 1$ be the minimum such that either $U_l^{2a+1}(T_1) = U_l^{2a-1}(T_1)$ or $T_3 \nto U_l^{2a}(T_1)$.  It is not hard to observe that such a value $a$ always exists because the blocks not adjacent to $W$ form the range $(F_r(W), F_l(W))$ for every $W \in \B(\Gamma)$.  Moreover (see Figure~\ref{fig:left to right}~(a)):
    \begin{enumerate}[(i)]
      \item $U_l(T_1) \in (T_2, T_3)$ and $U_l^{2i+1}(T_1) \in [U_l^{2i-1}(T_1), T_3)$ for every $2 \leq i < a$,
      \item $U_l^{2i}(T_1) \in (U_l^{2i-2}(T_1), T_2]$ for every $1 \leq i \leq a$.
    \end{enumerate}
    Thus, by marking every block in $[U_r(T_3), U_l(T_3)]$, the sequence $U_l(T_1)$, \ldots, $U_l^{2a+1}(T_1)$ can be obtained in $O(|(T_1, T_3)|) = O(d(v))$ time.  
    
    \begin{figure}
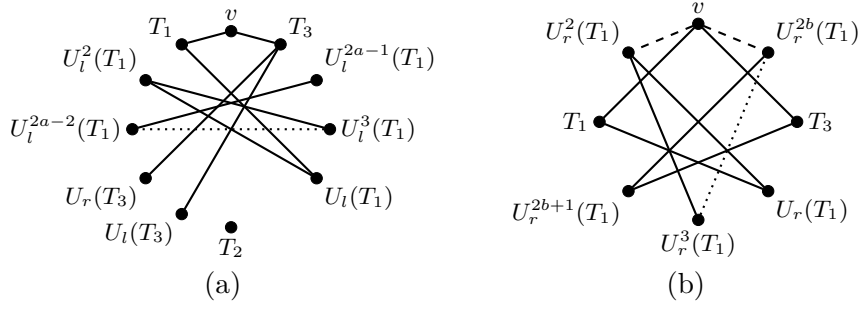

      \centering
      \begin{tabular}{c@{\hspace{1cm}}c}
        \input{srcFigLeftToRight} & \input{srcFigLeftToRight2} \\
         (a)  &  (b)
      \end{tabular}
      \caption{Adjacencies of $\overline{G(\Gamma)}$ in Lemma~\ref{lemma:subprocedures} (\ref{lemma:subprocedures:left to right}).  The blocks are drawn as they appear in the circular ordering $\B(\Gamma)$.  Dotted lines are used for continuity, while dashed lines means an edge could be present or absent.}\label{fig:left to right}
    \end{figure}

    If $T_3 \nto U_l^{2a}(T_1)$, then, by (i)~and~(ii), it follows that \{$v$, $-T_1$, $+U_l(T_1)$, \ldots, $+U_l^{2a}(T_1)$, $-T_3$\} induces a co-cycle in $H$, whose semiblocks are all adjacent to $T_2$.  Consequently, a minimally forbidden can be obtained in $O(|(T_1, T_3)|) = O(d(v))$ time.  Suppose, then, that $T_3 \to U_l^{2a}(T_1)$, thus $U_l^{2a+1}(T_1) = U_l^{2a-1}(T_1)$ and $U_l^{2a}(T_1)$ are the right co-end block of $\Gamma$.  
    
    By (i), the co-end blocks of $\Gamma$ belong to $(T_1, T_3)$.  Since $\Gamma$ is co-contig, we conclude that $T_1$ and $T_3$ belong to the same co-contig range, thus $v$ is fully adjacent to all the blocks in the co-contig range that contains $T_2$.  Let $b$ be the minimum such that $T_1$, $U_r(T_1)$, \ldots, $U_r^{2b+1}(T_1)$, $T_3$ induces a co-path (see Figure~\ref{fig:left to right}~(b)).  This co-path can be found in $O(d(v))$ time because $v$ is fully adjacent to $U_r^{2i+1}(T_1)$ for every $0 \leq i \leq b$.  Moreover, since $T_2 \to T_3$, it follows that $b \geq 1$.
    
    If $b > 1$ or $v$ is fully adjacent to $U_r^2(T_1)$, then \{$v$, $-T_1$, $\pm U_r^1(T_1)$, \ldots, $\pm U_r^{2b+1}(T_1)$, $-T_3$\} is an adequate forbidden (as it contains an induced $\overline{C_{2k}}$ ($k \geq 3$) or $\overline{H_2}$).  If $b = 1$ and $v$ is co-adjacent to $U_r^2(T_1)$, then, as $T_1 \to T_2$ and $T_2 \to T_3$, it follows that $T_2$, $U_r(T_1)$, $T_3$, $U_r^2(T_1)$, $T_1$, $U_r^3(T_1)$ appear in this order, thus $T_2$ is not adjacent to $U_r^2(T_1)$.  Hence, \{$v$, $+T_2$, $+U_r(T_1)$, $-T_3$, $-U_r^2(T_1)$, $-T_1$, $+U_r^3(T_1)$\} induces an $\overline{H_2}$ of $H$ (\ref{app:subprocedures:left to right}).
  \end{enumerate}
    
  Before dealing with~(\ref{lemma:subprocedures:two triangle}), we consider two subcases that are also solved in $O(d(v))$ time.
  \begin{enumerate}[(a)]
    \setcounter{enumi}{4}
    \item \label{lemma:subprocedures:zero triangle} $v$ is co-adjacent to $T_1$, $T_2$, and $T_3$, and $T_3 \to T_1$.  By inspection, it can be observed that $\B$ = \{$v$, $-T_1$, $-T_2$, $-T_3$, $\pm U_r(T_1)$, $\pm U_r(T_2)$, $\pm U_r(T_3)\}$ is a forbidden (\ref{app:subprocedures:zero triangle}).  Moreover, $\B$ is adequate because we can determine all the adjacencies in $O(1)$ time.  Indeed, as $U_r(T_i) \in (T_{i+1}, T_{i-1})$, it follows that $U_r(T_i) \to T_{i-1}$; since $U_r(T_i) \nto T_i$, then $U_r(T_i) \nto T_{i+1}$, while the adjacencies between $U_r(T_i)$ and $U_r(T_{i+1})$ are obtained in $O(1)$ time by Observation~\ref{obs:basic queries} (\ref{obs:basic queries:triangle adjacency}) (with input $W_l = F_r(T_{i+1})$, $W_r = U_r(T_{i+1})$, and $W = U_r(T_i)$). 
    
    \item \label{lemma:subprocedures:one triangle} $v$ is co-adjacent to $T_1$ and $T_3$ and $T_3 \to T_1$.  Let $W_l$ be the rightmost block in $[T_1, T_2]$ co-adjacent to $v$, and $W_r$ be the leftmost block in $[T_2, T_3]$ co-adjacent to $v$.  By Observation~\ref{obs:basic queries} (\ref{obs:basic queries:first co-adjacent}), $W_l$ and $W_r$ can be obtained in $O(d(v))$ time.  By Observation~\ref{obs:basic queries} (\ref{obs:basic queries:triangle adjacency}), the query of whether $W_l \to T_3$ can also be answered in $O(d(v))$ time.  If affirmative, then we obtain a minimally forbidden by invoking (\ref{lemma:subprocedures:zero triangle}) with input $T_1, W_l, T_3$.  Otherwise, we obtain that $W_l \neq T_2$ and, thus, $W_r \neq T_2$.  Again, by Observation~\ref{obs:basic queries} (\ref{obs:basic queries:triangle adjacency}), $O(d(v))$ time suffices to find out whether $W_l \to W_r$; if negative, then we obtain a minimally forbidden by invoking (\ref{lemma:subprocedures:left to right}) with input $W_l, T_2, W_r$.  When $W_l \to W_r$, we query whether $T_3 \to W_l$ using Observation~\ref{obs:basic queries} (\ref{obs:basic queries:triangle adjacency}); if negative, then \{$+T_2$, $v$, $-W_l$, $-T_3$\} induces a $K_{1,3}$, while if affirmative, then $W_r \neq T_3$ and we can obtain a minimally forbidden by invoking (\ref{lemma:subprocedures:zero triangle}) with input $W_l, W_r, T_3$.
  \end{enumerate}

  \begin{enumerate}[(a)]
    \setcounter{enumi}{1}
    \item Let $W$ be the rightmost block in $[T_1, T_2]$ that is co-adjacent to $v$.  If $W = T_1$, then we obtain a minimally forbidden by invoking (\ref{lemma:subprocedures:one triangle}).  When $W \neq T_1$, we query whether $W \to T_3$ and $T_3 \to W$ using Observation~\ref{obs:basic queries} (\ref{obs:basic queries:triangle adjacency}).  If $W \to T_3$, then we obtain a minimally forbidden by invoking (\ref{lemma:subprocedures:one triangle}) with input $T_3, T_1, W$.  If $T_3 \to W$, then we run (\ref{lemma:subprocedures:one triangle}) with input $W, T_2, T_3$.  Finally, if $W \nto T_3$ and $T_3 \nto W$, then \{$+T_2$, $v$, $-W$, $-T_3$\} induces a $K_{1,3}$.
    
    \item By Observation~\ref{obs:basic queries} (\ref{obs:basic queries:triangle adjacency}), we can query in $O(d(v))$ whether $U_r(T_2) \to T_1$ (resp.\ $T_3 \to U_l(T_2)$); if so, then we obtain a forbidden using (\ref{lemma:subprocedures:two triangle}) with input $T_1$, $T_3$, $U_r(T_2)$ (resp.\ $U_l(T_2)$). Otherwise, it follows that $U_l(T_2) \neq U_r(T_2)$ and, so, \{$v$, $-T_1$, $+T_2$, $-T_3$, $\pm U_l(T_2)$, $\pm U_r(T_2)$\} is a forbidden (\ref{app:subprocedures:center not dominating}).  Moreover, it is adequate as the only unknown edge is $U_r(T_2) \to U_l(T_2)$ and this edge can be queried in $O(d(v))$ time.
    
    \item As $\Gamma$ is a block representation, it cannot contain indistinguishable blocks.  Thus, $F_l(T_2) \nto T_3$ and $F_l(T_1) \nto T_2$ (hence $F_l(T_1) \neq F_l(T_2)$ and $F_l(T_2) \neq T_1$).  Observe that $F_r(T_1) \to F_l(T_1)$ if and only if $F_r(T_1) \to U_l(T_1)$ and $F_r^2(T_1) \neq U_l(T_1)$.  By hypothesis, we can query, in $O(d(v))$ time, whether $F_r(T_1) \to U_l(T_1)$; hence, we can also find out if $F_r(T_1) \to F_l(T_1)$ in $O(d(v))$ time.  Two cases then follow.
    \begin{description}
      \item [Case 1:] $F_r(T_1) \to F_l(T_1)$.  We first check if $v$ is co-adjacent to $F_r(T_1)$ or $F_l(T_1)$.  If so, then we find a forbidden by calling (\ref{lemma:subprocedures:two triangle}) with input $T_1 = F_l(T_1)$, $T_2 = T_1$, and $T_3 = F_r(T_1)$ (in the former case) or $T_1 = T_1$, $T_2 = F_r(T_1)$, $T_3 = F_l(T_1)$ (in the latter case).  Otherwise, $\B$ $=$ \{$v$, $-T_1$, $+T_2$, $-T_3$, $+U_l(T_1)$, $+F_l(T_1)$, $\pm F_l(T_2)$\} is a forbidden (\ref{app:subprocedures:center and right dominated:case 1}).
      \item [Case 2:] $F_r(T_1) \nto F_l(T_1)$.  In this case, $\B$ $=$ \{$\pm U_l(T_1)$, $\pm F_l(T_1)$, $\pm F_l(T_2)$, $-T_1$, $+T_2$, $-T_3$, $\pm F_r(T_1)$\} is a forbidden (\ref{app:subprocedures:center and right dominated:case 2}).  
    \end{description}
    
    Whichever the case, the forbidden $\B$ is adequate as it has at most two unknown edges, namely $F_r(T_1) \to U_l(T_1)$ and $U_l(T_1) \to F_l(T_2)$.  The former can be queried in $O(d(v))$ time by hypothesis.   To find out if the latter edge exists, we observe that $U_l(T_1) \to F_l(T_2)$ if and only if $F_r(U_l(T_1)) \to T_2$.  Since $U_l(T_1) \to F_l(T_1)$, then $F_r(U_l(T_1)) \in [F_l(T_1), T_1)$.  So, by Observation~\ref{obs:basic queries} (\ref{obs:basic queries:triangle adjacency}) ---with input $W_l = T_1$, $W_r = T_2$, $W = F_r(U_l(T_1))$---, $O(d(v))$ time suffices to determine if $F_r(U_l(T_1)) \to T_2$.  
  \end{enumerate}

\end{proof}

The next lemma describes how to find a forbidden when a long bad block $B$ is given.  Recall that $B$ is long when $F_r(B) \to F_l(B)$.  How $B$ was found, or why do we know that $B$ is long are irrelevant questions at this point.

\begin{lemma}\label{lemma:bad block long}
  Let $H$ be a graph with a vertex $v$, $\Gamma$ be a block co-contig representing $H \setminus \{v\}$, and $B \in \B(\Gamma)$ be a long bad block.  Given a semiblock $\Gamma$-pointer to $B$, a minimally forbidden of $H$ can be obtained in $O(d(v))$ time.
\end{lemma}

\begin{proof}
  If $v$ is co-adjacent to a block $W$ in $[F_l^2(B), F_r(B)]$, then a minimally forbidden can be obtained by invoking Lemma~\ref{lemma:subprocedures} (\ref{lemma:subprocedures:two triangle}) with input $T_1 = F_l(B)$, $T_2 = B$, $T_3 = W$.  Analogously, a minimally forbidden is obtained in $O(d(v))$ time when $v$ is co-adjacent to a block in $[F_l(B), F_r^2(B)]$.  Let $W_l$ be the rightmost block in $(F_r^2(B), B)$ that is co-adjacent to $v$, and $W_r$ be the leftmost block in $(B, F_l^2(B))$ that is co-adjacent to $v$.  By Observation~\ref{obs:basic queries}, $W_l$ and $W_r$ are found in $O(d(v))$ time, while we can also query whether $W_l \to W_r$ in $O(d(v))$ time.  If $W_l \nto W_r$, then we compute a minimally forbidden by invoking Lemma~\ref{lemma:subprocedures}~(\ref{lemma:subprocedures:left to right}).  Using Observation~\ref{obs:basic queries}~(\ref{obs:basic queries:first co-adjacent}), we search for a block $W \in (F_r(B), F_l(B))$ co-adjacent to $v$.  Moreover, when such a block exists, we query whether $W \to W_l$ and $W_r \to W$ with Observation~\ref{obs:basic queries} (\ref{obs:basic queries:triangle adjacency}).  When $W \to W_l$ and $W_r \to W$, we find a forbidden by calling Lemma~\ref{lemma:subprocedures} (\ref{lemma:subprocedures:two triangle}), while when $W \nto W_l$ (resp.\ $W_r \nto W$), the family \{$+F_l(B)$, $v$, $-W_l$, $-W$\} (resp.\ \{$+F_r(B)$, $v$, $-W_r$, $-W$\}) induces a $K_{1,3}$. Therefore, we may suppose from now on that $v$ is fully adjacent to every block in $[F_l^2(B), F_r^2(B)]$.

  Note that, by hypothesis, $F_r(B) \to U_l(B)$ and $U_r(B) \toeq U_l(B)$.  Hence, we can obtain a minimally forbidden in $O(d(v))$ by invoking Lemma~\ref{lemma:subprocedures} (\ref{lemma:subprocedures:center not dominating}) with input $T_1 = W_l$, $T_2 = B$, and $T_3 = W_r$ when $F_l(W_l) \neq F_l(B)$ and $F_r(B) \neq F_r(W_l)$.  Analogously, we obtain a minimally forbidden in $O(d(v))$ time when $F_r(W_l) = F_r(W_r)$, by calling Lemma~\ref{lemma:subprocedures} (\ref{lemma:subprocedures:center and right dominated}).  By exchanging the roles of $\Gamma$ and $\Gamma^{-1}$ if required, suppose $F_r(B) = F_r(W_r)$ and, hence, $F_r(W_l) \neq F_r(B)$.
  
  Recall that $F_r(B) \nto W_l$ by definition.  For the next step, we proceed according to whether $U_r(B) = F_l(W_l)$ or not.  If $U_r(B) \neq F_l(W_l)$, then \{$v$, $-W_l$, $+B$, $-W_r$, $+F_r(B)$, $+U_r(B)$, $+F_l(B)$\} induces an $\overline{H_4}$ (\ref{app:bad block long}).  When $U_r(B) = F_l(W_l)$, we determine in $O(1)$ time if $U_r(B)$ is a left co-end block by querying whether $U_r^3(B) = U_r(B)$.  If $U_r(B)$ is not a left co-end block, then $B' = U_r^2(B) \in (W_l, B)$ is a block such that $F_r(B') \neq F_r(W_r) = F_r(B)$ and $F_l(B') \neq F_l(W_l) = U_r(B)$.  In this case, we obtain a minimally forbidden by calling Lemma~\ref{lemma:subprocedures}~(\ref{lemma:subprocedures:center not dominating}) with input $T_1 = W_l$, $T_2 = B'$, and $T_3 = W_r$. (As $F_r(B') \to U_r(B)$ and $F_l(B') \in (U_r(B), F_l(B))$, we can find out whether $F_r(B') \to F_l(B')$ by traversing $(U_r(B), F_l(B))$.  Recall that $v$ is adjacent to every block in $(U_r(B), F_l(B))$, thus $O(d(v))$ time suffices to answer the query).  We can suppose, then, that $B'$ and $U_r(B)$ are left co-end blocks.   
  
  \begin{figure}
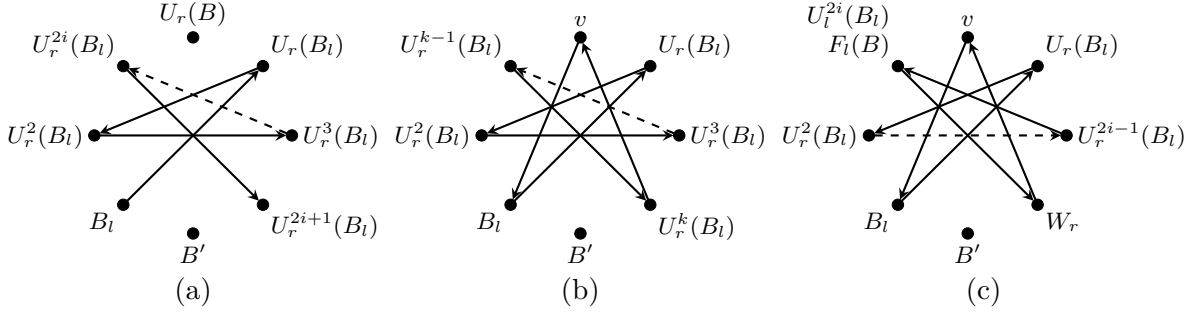

    \begin{tabular}{c@{\hspace{2mm}}c@{\hspace{2mm}}c}
      \input{srcFigLongBadBlockPath} & \input{srcFigLongBadBlockCase1} & \input{srcFigLongBadBlockCase2} \\
      (a) & (b) & (c)
    \end{tabular}
    \caption{Co-paths of $H$ in Lemma~\ref{lemma:bad block long}: (a) shows the position of the blocks in $\B(\Gamma)$, (b) depicts Case~1 for $k$ odd, and (c) shows Cases 2~and~3.}\label{fig:bad block long}
  \end{figure}

  Let $B_l = U_r(F_r(B))$ if $v$ is co-adjacent to $U_r(F_r(B))$, and $B_l = W_l$ otherwise.  Since $\Gamma$ is a co-contig, $U_r(B)$ and $B'$ are the only left co-end blocks.  Consequently $U_r^0(B_l)$, \ldots, $U_r^k(B_l)$ induce a co-path for every $k>0$ such that $U_r^k(B_l) \not\in \{B', U_r(B)\}$.  Moreover, the order of these blocks in $\B(\Gamma)$ is as in Figure~\ref{fig:bad block long}~(a).  Let $k$ be the minimum such that either:
  \begin{enumerate}[(i)]
    \item $U_l^{k}(B_l)$ is co-adjacent to $v$,
    \item $k$ is odd and $U_r^k(B_l) \nto F_l(B)$, or
    \item $k$ is even and $U_l^{k}(B_l) \nto W_r$,
  \end{enumerate}
  and consider the following cases.
  \begin{description}
    \item [Case 1:] (i) holds; see Figure~\ref{fig:bad block long} (b). Then, $\B$ $=$ \{$v$, $-B_l$, $+U_r(B_l)$, \ldots, $+U_r^{k-1}(B_l)$, $-U_r^{k}(B_l)$\} induces a co-cycle of length $k+2$.  Suppose $k$ is even and note that $U_r(B_l) \in [B', F_r(B)]$.  Hence $U_r^2(B_l) \in [F_r(B), U_r(F_r(B))] \subseteq [F_r(B), B_l]$.  Moreover, $U_r^2(B_l) \neq B_l$ because otherwise $B_l$ would be a left co-end block different to  $B'$ and $U_r(B)$.  Therefore, either $F_r(B) \to U_r^2(B_l)$ or $U_r^2(B_l) = U_r(F_r(B)) \neq B_l$, which implies that $v$ is fully adjacent to $U_r^2(B_l)$.  Consequently, $k > 2$ and $\B$ is a forbidden.  When $k$ is odd, $B' \to U_r^i(B_l)$ for every odd $1 \leq i \leq k$, while, because of the minimality of $k$, $U_r^i(B_l) \to B'$ for every even $1 \leq i < k$.   Therefore, $\B \cup \{B'\}$ is a forbidden.  
 
    \item [Case 2:] (ii) holds and (i) does not; see Figure~\ref{fig:bad block long} (c) with $k = i-1$.  Then, $\B$ $=$ \{$v$, $-B_l$, $+U_r(B_l)$, \ldots, $+U_r^{k}(B_l)$\} induces an odd co-path.  Moreover, by the minimality of $k$, we obtain that $F_l(B)$ is adjacent to every block in $\B \setminus \{U_r^{k}(B_l)\}$, while $B'$ and $W_r$ are adjacent to every block in $\B \setminus \{v\}$.  Hence, $\B \cup \{F_l(B), W_r, B'\}$ is a forbidden.
    
    \item [Case 3:] (iii) holds and (i) does not; see Figure~\ref{fig:bad block long} (c) with $k = i$.  This time, $\B$ $=$ \{$v$, $-B_l$, $-W_r$, $+U_r(B_l)$, \ldots, $+U_r^k(B_l)$\} induces an odd co-cycle.  As in case (i), $B'$ is adjacent to every block of $\B \setminus \{U_r^k(B_l)\}$, while, by the minimality of $k$, $U_r^k(B_l) \in [F_l(B), B]$, thus $U_r^k(B_l) \to B'$.  Consequently, $\B \cup \{B'\}$ is a forbidden.   
  \end{description}
  To compute $U_r(B_l), \ldots, U_r^k(B_l)$, we proceed as follows.  First, we mark with $1$ the blocks in $[U_r(B), F_l(B))$ and with $2$ the blocks in $[B, W_r]$.  Then we traverse $U_r^k(B_l)$ for each $k > 0$ until we find a block co-adjacent to $v$ or a marked block.  In the former case (i) holds for $k$, while in the latter case (ii) or (iii) holds for $k-1$ when $U_r^k(B)$ is marked with $1$ or $2$, respectively.  We conclude that $O(d(v))$ time is enough to find the forbidden.
\end{proof}

The case in which $B$ is a short bad block is solved next.  Note that, a priori, $\Gamma$ could contain long bad blocks; yet, we have no evidence about the existence or non-existence of long bad blocks.

\begin{lemma}\label{lemma:bad block short}
  Let $H$ be a graph with a vertex $v$, $\Gamma$ be a block co-contig representing $H \setminus \{v\}$, and $B \in \B(\Gamma)$ be a short bad block.  Given a semiblock $\Gamma$-pointer to $B$, a minimally forbidden of $H$ can be obtained in $O(d(v))$ time.
\end{lemma}

\begin{proof}
  Let $W_l$ be the rightmost block in $[F_l(B), B)$ co-adjacent to $v$, and $W_r$ be the leftmost block in $(B, F_r(B)]$ co-adjacent to $B$.  We can find $W_l$ and $W_r$ and test whether $W_l \to W_r$ with Observation~\ref{obs:basic queries}. If $W_l \nto W_r$, we find a minimally forbidden as in Lemma~\ref{lemma:subprocedures} (\ref{lemma:subprocedures:left to right}).  
  
  By hypothesis, $F_r(B) \nto F_l(B)$; hence, $F_r(B) \to U_l(B)$ if and only if $F_r^2(B) = U_l(B)$.  Similarly, $U_r(B) \to U_l(B)$ if and only if $U_r(B) \neq U_l(B)$ and either $F_r(B) \to U_l(B)$ or $F_l(U_l(B)) = U_r(B)$.  We conclude, therefore, that $O(1)$ time suffices to determine whether $F_r(B) \to U_l(B)$ and whether $U_r(B) \to U_l(B)$.  Therefore, by Lemma~\ref{lemma:subprocedures} (\ref{lemma:subprocedures:center not dominating}), we can suppose that $F_r(B) = F_r(W_r)$.  Otherwise, we obtain a minimally forbidden in $O(d(v))$ time.  Similarly, by Lemma~\ref{lemma:subprocedures}~(\ref{lemma:subprocedures:center and right dominated}) we find a minimally forbidden if $F_r(W_l) = F_r(B)$ and $F_l(W_l)$ is not an end block.  Then, two cases remain:
  \begin{description}
    \item[Case 1:] $F_r(W_l) \neq F_r(B)$.  Thus, $F_r(B) \neq W_r$ and, as $\Gamma$ has no pair of indistinguishable blocks, $F_l(B) \nto W_r$ and, so, $F_l(B) \neq W_l$.  Under this situation, \{$v$, $-W_l$, $+B$, $-W_r$, $\pm F_l(B)$, $\pm F_r(B)$\} is an adequate forbidden (\ref{app:bad block short:case 1}).
    \item[Case 2:] $F_r(W_l) = F_r(B)$ and $F_l(W_l)$ is an end block.  Note that $F_l(W_l) \neq F_l(B)$ and $F_l(B) \neq W_l$ because $\Gamma$ has no indistinguishable blocks.  If $F_r(W_l)$ is not a right end block, then let $X = U_r(B_l)$ and note that $W_r \nto X$.  Otherwise, since $W_l$ is not universal, it follows that $\Gamma$ has some other contig with a block $X$.  Whichever the case, \{$v$, $\pm F_l(W_l)$, $\pm F_l(B)$, $-W_l$, $+B$, $-W_r$, $\pm X$\} is an adequate forbidden (\ref{app:bad block short:case 2}).  We remark that $X$ can be obtained in $O(1)$ time by using the representation and contig pointers of $F_l(W_l)$. 
  \end{description}
\end{proof}

We are now ready to deal with the existence of bad blocks, regardless of their type.  The main idea is to find a bad block $B$ which can be used as input of either Lemma~\ref{lemma:bad block long} or Lemma~\ref{lemma:bad block short}.  In order to apply either lemma, we need to find out whether $B$ is short or long.  However, we do not know how to test, in $O(d(v))$ time, whether $F_r(B) \to F_l(B)$ when only $B$ is given.  The solution to this problem is to take advantage of the dynamic nature of the algorithm.  That is, the answer to $F_r(B) \to F_l(B)$ was found when the last vertex of $G$ was inserted and it is implicitly encoded in the semiblock paths.

\begin{lemma}\label{lemma:bad block}
  Let $H$ be a graph with a vertex $v$, and $\Gamma$ be block co-contig representing $H \setminus \{v\}$.  Given $N(v)$, it costs $O(d(v))$ time to test if $\Gamma$ has bad blocks.  Furthermore, if $\Gamma$ has bad blocks, then a minimally forbidden of $H$ can be obtained within the same amount of time.
\end{lemma}

\begin{proof}
  The algorithm has two main phases.  In the first phase, all the bad blocks of $\Gamma$ are marked; in the second phase a minimally forbidden is obtained.
  
  To find the bad blocks we first mark all the blocks of $\Gamma$ that are fully adjacent to $v$ in such a way that $B, W$ have the same mark if and only if $B$ and $W$ lie in the same contig, $v$ is fully adjacent to all the blocks in $[B, W]$ (or $[W, B]$), and $[B, W)$ (or $[W, B)$) has no right end blocks.  Then, a block $B$ adjacent to $v$ is bad if and only if $B$ is not an end block and either:
  \begin{itemize}
    \item $L(B)$ and $R(B)$ are unmarked, or
    \item $R(B)$ and $F_r(B)$ have different marks, and $L(B)$ and $F_l(B)$ have different marks.
  \end{itemize}
  It is not hard to both steps can be achieved in $O(d(v))$ time (see e.g.~\cite{SoulignacA2015}).  After these steps we can test if any block in $\Gamma$ is bad in $O(1)$ time.
  
  Let $T_1, \ldots, T_k$ be the blocks of $\MP(\gamma)$ for any $\gamma \in \Gamma$.  Recall that, by definition, $\Gamma$ has a long block if and only if $k = 3$ and $\MP(\gamma)$ is circular.  Hence, if $k \neq 3$ or $\MP(\gamma)$ is linear, then we obtain a minimally forbidden by invoking Lemma~\ref{lemma:bad block short} with and any bad block as input.  When $k = 3$ and $T_3 \to T_1$, we first check if $v$ is co-adjacent to $T_1$, $T_2$, and $T_3$; if true, then we obtain the minimally forbidden via Lemma~\ref{lemma:subprocedures} (\ref{lemma:subprocedures:two triangle}) with input $T_1$, $T_2$, $T_3$.  Suppose, then, that $v$ is fully adjacent to $T_2$.  If $T_2$ is a bad block, then we obtain the forbidden with Lemma~\ref{lemma:bad block long} using $T_2$ as input.  Otherwise, $v$ is fully adjacent to all the blocks in either $[T_2, F_r(T_2)]$ or $[F_l(T_2), T_2]$. Assume the former, as the proof for the latter is analogous. 
  
  Let $W_r$ and $W_l$ be the leftmost and rightmost blocks co-adjacent to $v$ in $(F_r(T_2), T_2)$, respectively, and observe that $W_l \nto W_r$ because $T_2 \nto W_r$.  We can both find $W_l$ and $W_r$ and query if $W_r \to W_l$ in $O(d(v))$ time by Observation~\ref{obs:basic queries}.  If any $B \in (W_l, W_r)$ is bad, then $W_l \to B$ and $B \to W_r$ because $v$ is fully adjacent to all the blocks in $(W_l, W_r)$.  When $W_r \nto W_l$, the family \{$v$, $+B$, $-W_l$, $-W_r$\} induces a $K_{1,3}$, while when $W_r \to W_l$ we find a minimally forbidden by calling Lemma~\ref{lemma:bad block long} with input $B$.  Whichever the case, $O(d(v))$ time is enough to obtain a minimally forbidden when some block in $(W_l, W_r)$ is bad.
  
  For the final case, suppose no block in $(W_l, W_r)$ is bad, thus the bad block $B$ belongs to $(W_r, W_l)$.  Note that either $F_r(B) \in (W_l, W_r)$ or $F_l(B) \in (W_l, W_r)$, we assume the former as the other case is analogous.  This means that $F_r(B)$ is good, thus $F_r^2(B)$ belongs to $(W_l, W_r)$ as well.  Hence, we can check if $F_r^2(B) \to B$ as in Observation~\ref{obs:basic queries}~(\ref{obs:basic queries:triangle adjacency}) (with input $W_l = F_r(B)$, $W_r = F_r^2(B)$, and $W = B$).  Then we find the minimally forbidden calling Lemma~\ref{lemma:bad block long} (if affirmative) or Lemma~\ref{lemma:bad block short} (if negative) with input $B$.  We conclude, therefore, that a minimally forbidden can be obtained in $O(d(v))$ time.
\end{proof}

Having dealt with bad blocks, we now consider the case in which $N(v)$ is not consecutive in $\Gamma$.  That is, we discuss how to find a forbidden when no co-contig representing $G$ satisfies~\receptionref{1}.  The core of the proof is given in the next lemma.

\begin{lemma}\label{lemma:range}
  Let $H$ be a graph with a vertex $v$, and $\Gamma$ be a block co-contig representing $H \setminus \{v\}$ with no bad blocks.  If $B_l \neq B_r$, and $X$ are blocks of $\Gamma$ such that:
  \begin{itemize}
    \item $X \not\in [B_l, B_r] \cup \{L(B_l), R(B_r)\}$, and $v$ is adjacent to every block in $[B_l, B_r] \cup \{X\}$,
    \item if $L(B_l) \neq \bot$, then $v$ is co-adjacent $L(B_l)$; otherwise $v$ is co-adjacent to $B_l$, and
    \item if $R(B_r) \neq \bot$, then $v$ is co-adjacent $R(B_r)$; otherwise $v$ is co-adjacent to $B_r$,
  \end{itemize}
  then a (minimally) forbidden of $H$ can be obtained in $O(d(v))$ time when semiblock $\Gamma$-pointers to $B_l$ and $B_r$ and a semiblock pointer $\DS{X}$ to $X$ are given.
\end{lemma}

\begin{proof}
  The first step of the algorithm is to decide if $X \to L(B_l)$.  Even though we are unaware of the type of $\DS{X}$, we can answer this query in $O(d(v))$ time by observing that, since $B_l$ is good and $v$ is co-adjacent to $B_r$ or $R(B_r)$, then either $L(B_l) = \bot$ or $F_r(B_l) \in [B_l, B_r]$.  In the former case $X \nto L(B_l)$.  In the latter case $X \to L(B_l)$ if and only if $F_r(X) \in [L(B_l), B_r]$, which happens if and only if $\DS{X} \in [L(B_l), B_r]$ or $\DS{X}^{-1} \in [L(B_l), B_r]$.  (Here $\DS{X}^{-1}$ is a reverse of $\DS{X}$.)  In a similar way, we can test if $R(B_r) \to X$ in $O(d(v))$ time.  If $X \nto L(B_l)$ and $R(B_r) \nto X$, then $R(B_r) \nto L(B_l)$ and:
  \begin{itemize}
    \item \{$v$, $+B_l$, $+B_r$, $+X$\} induces a $K_{1,3}$ if $B_l \nto B_r$, and
    \item \{$v$, $+B_l$, $+B_r$, $+X$, $-L(B_l)$, $-R(B_r)$\} induces an $\overline{S_3}$ if $B_l \to B_r$ (\ref{app:range:case 1}).  We remark that $B_l$ is not the left end block in this case, thus $L(B_l) \neq \bot$.  Otherwise, $B_r$ is not the right end block and, since $B_r$ is good and $v$ is co-adjacent to $B_l$, it follows that $v$ is fully adjacent to $(B_r, F_r(B_r)]$.  This is a contradiction because $v$ is co-adjacent to $R(B_r)$.  In a similar way $R(B_r) \neq \bot$.
  \end{itemize}
  From now on suppose $R(B_r) \to X$, as the proof when $L(B_l) \to X$ is analogous.  Hence, $X$ and $B_r$ lie in the same contig.  Moreover, by applying {\tt reversed} if required, we may assume that $\DS{X}$ is a $\Gamma$-pointer, thus we can invoke Observation~\ref{obs:basic queries} whenever it is required.
  
  Let $X_r = F_r(R(B_r))$, $W_l = B_l$ if $L(B_l) = \bot$, and $W_l = L(B_l)$ otherwise.  Since $B_r$ is good and $B_r \to R(B_r)$, we observe that $W_l \nto B_r$.  Similarly, since $R(B_r) \to X$ and $X$ is good, we obtain that $v$ is adjacent to $X_r$ and, since $X_r$ is good, it follows that $X_r \nto W_l$.   By checking if $F_l(X_r) \neq R(B_r)$, we can decide if $B_r \to X_r$; if negative, then $\{v, +B_r, -R(B_r), +X_r, -W_l\}$ induces a $C_{4}^*$.  Suppose, then, that $B_r \to X_r$ (hence $B_r\to X$), thus $F_l(B_r) \nto R(B_r)$ because $\Gamma$ has no indistinguishable blocks.  Next we query if $F_r(X) = X_r$.  If affirmative, then $U_l(X) \to R(B_r)$ because $\Gamma$ has no indistinguishable semiblocks.  So, $U_l(X) \in (F_l(B_l), B_r)$ is adjacent to $v$, which implies that $W_l \nto U_l(X)$.  Consequently, \{$v$, $+U_l(X)$, $-R(B_r)$, $+X$, $-W_l$\} induces a $C_4^*$.  Finally, if $F_r(X) \neq X_r$, then \{$v$, $-W_l$, $+F_l(B_r)$, $+B_r$, $-R(B_r)$, $+X$, $+F_r(X)$\} is a forbidden (\ref{app:range:case 2}) which is adequate by Observation~\ref{obs:basic queries}.
\end{proof}

We are now ready to find a minimally forbidden when $N(v)$ is not consecutive in $\Gamma$.  Before doing so, it is convenient to state precisely what we mean by consecutive.  We remark that the definition holds for any round representation and not only for co-contigs.  We say that $N(v)$ is \emph{consecutive} in a round representation $\Phi$ when there exist two (possibly equal) semiblocks $B_a$ and $B_b$ such that $N(v) \subseteq \bigcup [B_a, B_b] \cup +B_a \cup +B_b$.  In such a case, $\Bip{B_a}{B_b}$ \emph{witnesses} that $N(v)$ is consecutive in $\Phi$.  Clearly, if $\Bip{B_l}{B_r}$ satisfies~\receptionref{1}, then $N(v) = \bigcup [B_l, B_r]$ is consecutive in $\Phi$.  However, consecutiveness is a slight generalization of condition~\receptionref{1} that allows $v$ to be co-adjacent to $B_a$ and $B_b$.  The next result applies Lemma~\ref{lemma:range} to find a minimally forbidden when $N(v)$ is not consecutive in any co-contig representing $G(\Gamma)$.

\begin{lemma}\label{cor:consecutive}
  Let $H$ be a graph with a non-isolated vertex $v \in V(H)$, and $\Gamma$ be block co-contig representing $H \setminus \{v\}$ with no bad blocks.  Given $N(v)$, it takes $O(d(v) + \epTest)$ time to transform $\Gamma$ into a block co-contig $\Gamma'$ representing $G(\Gamma)$ in which $N(v)$ is consecutive.  The algorithm outputs $\Gamma'$-pointers to the blocks $\Bip{B_a}{B_b}$ witnessing that $N(v)$ is consecutive in $\Gamma'$, or a minimally forbidden of $H$ when such a representation $\Gamma'$ does not exist.
\end{lemma}

\begin{proof}
  As discussed in~\cite{HellShamirSharanSJC2001,SoulignacA2015}, $O(d(v))$ time suffices to find a set $\{\DS{B}_i, \DS{W}_i \mid 1 \leq i \leq k\}$ of semiblock pointers such that:
  \begin{enumerate}[(a)]
    \item $\DS{B}_i$ and $\DS{W}_i$ are associated to the same contig $\gamma_i$,
    \item $v$ is fully adjacent to every block in $(B_i, W_i)$, \label{cor:consecutive:fully}
    \item $[B_1, W_1], \ldots, [B_k, W_k]$ is a partition of the blocks adjacent to $v$, and \label{cor:consecutive:range}
    \item either $L_i(B_i) = \bot$ (resp.\ $R_i(W_i) = \bot$) or $v$ is co-adjacent to $L_i(B_i)$ (resp.\ $R_i(W_i)$).\label{cor:consecutive:maximal}
  \end{enumerate}
  We remark that the type of $\DS{B}_i$ is unknown, as we are unaware of whether $\gamma_i \in \Gamma$ or $\gamma_i \in \Gamma^{-1}$.  In (\ref{cor:consecutive:fully})~and~(\ref{cor:consecutive:range}) above, $[B_1, W_1]$ refers to the range in $\gamma_i$, that could be a range of $\Gamma^{-1}$.  For the sake of notation, we write $L_i$ and $R_i$ as shortcuts for $L^{\gamma_i}$ and $R^{\gamma_i}$, as in~(\ref{cor:consecutive:maximal}) above.  
  
  When $k=1$, $\Bip{B_1}{W_1}$ witnesses that $N(v)$ is consecutive in $\Gamma'$, where $\Gamma'$ is the type of $\DS{B}_1$.  Suppose, then, that $k \geq 2$.  If either none of $\{B_i, W_i\}$ ($1\leq i\leq k$) are end blocks or $B_i \neq W_i$ and $v$ is co-adjacent to the end blocks in $\{B_i, W_i\}$, then we obtain a minimally forbidden by invoking Lemma~\ref{lemma:range} with input $B_l = B_i$, $B_r = W_i$, and $X = B_{j}$ for $j \neq i$.  If $k \geq 3$, then \{$v$, $+E_1$, $+E_2$, $+E_3$\} induces a $K_{1,3}$, where $E_i$ is the end block of $\{B_i, W_i\}$.  Therefore, $k = 2$, $E_i \in \{B_i, W_i\}$ ($i \in \{1,2\}$) is an end block, and $v$ is fully adjacent to $E_i$ if $B_i \neq W_i$.  Exchanging the roles of $\gamma_i$ and $\gamma_i^{-1}$ if required, we may assume that $E_1 = W_1$ is a right end block and $E_2 = B_2$ is a left end block.  If $B_2$ has the same type $\Gamma'$ as $W_1$, then $\Bip{B_1}{W_2}$ witnesses that $N(v)$ is consecutive in $\Gamma'$; otherwise, $\Bip{B_1}{W_2}$ witnesses that $N(v)$ is consecutive in $(\Gamma' \setminus \{\gamma_2\}) \cup \{\gamma_2^{-1}\}$.  As discussed in Section~\ref{sec:data structure:round representations}, $O(1)$ time is enough to query the types of $W_1$ and $B_2$, while replacing $\gamma_2$ with $\gamma_2^{-1}$ costs $O(\epTest)$ time.
\end{proof}

By definition, $N(v) \neq \emptyset$ is consecutive in a round block representation $\Gamma$ when it has two (possibly equal) blocks $B_a$ and $B_b$ such that $N(v) = \bigcup(B_a, B_b) \cup +B_a \cup +B_b$.  We can separate $B_a$ and $B_b$ in pairs of consecutive indistinguishable semiblocks $\Bip{-B_a}{+B_a}$ and $\Bip{+B_b}{-B_b}$, respectively, to obtain a new round representation of $G$.  Of course, if either $+B_a$ (resp.\ $-B_b$) or $-B_a$ (resp. $+B_b$) is empty, then nothing is separated out of $B_a$ (resp.\ $B_b$).  Similarly, if $B_a = B_b$, then $+B_a$ is separated to either the left or the right of $-B_a$.  We refer to the round representation so obtained as being \emph{$v$-associated} to $\Gamma$.  Observe that $v$ is simultaneously adjacent and co-adjacent to at most two blocks of $\Gamma$, namely $B_a$ and $B_b$.  Thus, no matter which pair witnesses that $N(v)$ is consecutive in $\Gamma$, only $B_a$ and $B_b$ could be separated.  Moreover, both $B_a$ and $B_b$ get separated only if $v$ is co-adjacent to both $B_a$ and $B_b$, in which case either $B_a = B_b$ or $\Bip{B_a}{B_b}$ is the only pair witnessing that $N(v)$ is consecutive.  We conclude, therefore, that at most two round representations $v$-associated to $\Gamma$ exist.

Recall that, when $v$ is not isolated, the co-connected graph $H$ is PCA if and only if $G = H \setminus \{v\}$ admits a $v$-receptive round representation $\Phi$.  By the Reception Theorem (applied to the component that contains $v$), $\Phi$ has a pair of semiblocks $\Bip{B_l}{B_r}$ that satisfies \receptionref{1}--\receptionref{3}.  Recall that $\Bip{B_l}{B_r}$ satisfies \receptionref{1} if and only if $N(v) = \bigcup[B_l, B_r]$, while $\Bip{B_l}{B_r}$ satisfies \receptionref{2} when no pair of semiblocks in $\Phi \setminus \{B_l, B_r\}$ are indistinguishable.  It is not hard to see that $\Phi$ satisfies \receptionref{1} and \receptionref{2} if and only if $\Phi$ is a round representation $v$-associated to $\Gamma$, for some round block representation $\Gamma$ of $G$.  Indeed, by~\receptionref{2}, $\Phi$ has at most two pair of indistinguishable semiblocks, namely $\{L(B_l), B_l\}$ and $\{B_r, R(B_r)\}$.  By compacting $L(B_l) \cup B_l$ into $B_a$ and $B_r \cup R(B_r)$ into $B_b$, we obtain a round block representation $\Gamma$ that has $\Phi$ as its $v$-associated representation.  We record this fact for future reference.

\begin{observation}\label{obs:v-associated}
  Let $H$ be a co-connected graph with a vertex non-isolated vertex $v$.  A round representation $\Phi$ of $H \setminus \{v\}$ satisfies \receptionref{1}~and~\receptionref{2} if and only if $\Phi$ is $v$-associated to a round block representation of $H \setminus \{v\}$.
\end{observation}

By definition, a round representations is just a family of contigs with no order.  Consequently, $G$ admits only two round block representations in which $N(v)$ is consecutive, namely $\Gamma$ and $\Gamma^{-1}$.  Therefore, the only round representations of $G$ that satisfy \receptionref{1} and \receptionref{2} are those $v$-associated with $\Gamma$ and $\Gamma^{-1}$.  We show how to obtain a minimally forbidden when none of the representations $v$-associated to $\Gamma$ satisfies \receptionref{3}.  Before doing so, we find it convenient to recall that $\Bip{B_l}{B_r}$ satisfies \receptionref{3} when there exists $B_m \in [B_l, B_r]$ witnessing that $\Bip{B_l}{B_r}$ is receptive.  Also, recall that $B_m$ witnesses that $\Bip{B_l}{B_r}$ is receptive when (see Figure~\ref{fig:receptive}):
\begin{description}
  \item[\witnessref{1}]  $B_m$ is an end semiblock, $B_l \toeq B_m$, and $B_m \toeq B_r$, or
  \item[\witnessref{2}]  $B_l \toeq B_m$, $B_m \nto R(B_r)$, $L(B_l) \nto R(B_m)$, and $\RR(B_m) \toeq B_r$. 
\end{description}

\begin{lemma}\label{lemma:reception fails}
 Let $H$ be a graph with a non-universal vertex $v \in V(H)$, $\Gamma$ be a block co-contig representing $H \setminus \{v\}$ with no bad blocks, and $\Bip{B_a}{B_b}$ be a pair witnessing that $N(v) \neq \emptyset$ is consecutive in $\Gamma$. Given semiblock $\Gamma$-pointers to $B_a$ and $B_b$, it takes $O(d(v))$ time to transform $\Gamma$ into a co-contig $\Phi$ representing $H \setminus \{v\}$ that satisfies \receptionref{1}--\receptionref{3}.  The algorithm outputs the semiblock $\Phi$-pointers to a pair $\Bip{B_l}{B_r}$ that satisfies \receptionref{1}--\receptionref{3}, or a minimally forbidden of $H$ when such a representation $\Phi$ does not exist.
\end{lemma}

\begin{proof}
  If $B_a = B_b$, then $B_a$ is an end block because it is good.  Thus, one of the representations $\Phi$ $v$-associated to $\Gamma$ has $+B_a$ as an end semiblock and, so, $+B_a$ satisfies \witnessref{1}.  Then, by Observation~\ref{obs:v-associated}, $\Bip{+B_a}{+B_a}$ satisfies \receptionref{1}--\receptionref{3} in $\Phi$.  Suppose, from now on, that $B_a \neq B_b$.  Let $\Phi$ be the co-contig $v$-associated to $\Gamma$ such that $N(v) = \bigcup[B_l, B_r]$ for $B_l = +B_a$ and $B_r = +B_b$.  By Observation~\ref{obs:v-associated}, $\Phi$ satisfies both \receptionref{1}~and~\receptionref{2}.  With the operations discussed in Section~\ref{sec:data structure:round representations}, we can transform $\Gamma$ into $\Phi$ and test if $\Phi$ satisfies \receptionref{3} in $O(d(v))$ time.  If affirmative, then we return $\Bip{B_l}{B_r}$.  Suppose, then, that $\Phi$ does not satisfy~\receptionref{3}.  
  
  \begin{description}
    \newcounter{lemma:reception fails:claims}
    \refstepcounter{lemma:reception fails:claims}
    \item [Claim~\arabic{lemma:reception fails:claims}:] \label{lemma:reception fails:minimal range} $B_l \nto B_r$.  If $B_l \to B_r$, then neither $B_l$ nor $B_r$ is an end semiblock because $B_m \in \{B_l, B_r\}$ does not satisfy \witnessref{1}.  So, $L(B_l) \to B_l$ and $B_r \to R(B_r)$ and, since $B_l$ and $B_r$ are good, it follows that $B_l = F_l(B_r)$ and $B_r = F_r(B_l)$.  Consequently, $B_m = U_l(R(B_r)) \in [B_l, B_r)$, thus $B_l \toeq B_m$ and $R(B_m) \to B_r$.  Moreover, since $B_m \nto R(B_r)$ and $B_m$ does not satisfy \witnessref{2}, it follows that $L(B_l) \to R(B_m) = F_l(R(B_r))$, a contradiction because $R(B_m)$ is good.

    \refstepcounter{lemma:reception fails:claims}
    \item [Claim~\arabic{lemma:reception fails:claims}:] \label{lemma:reception fails:short range} $F_r(B_l) \nto B_r$.  Suppose $F_r(B_l) \to B_r$, thus $B_l \neq F_r(B_l)$ and $F_r(B_l) \neq B_r$.  If $F_r^2(B_l) = B_r$, then $F_r(B_m)$ witnesses that $\Bip{B_l}{B_r}$ is receptive by \witnessref{2}.  Otherwise, $F_r(B_l) \to R(B_r)$, thus $B_m = U_l(R(B_r)) \in [B_l, B_r)$.  Consequently, $B_l \toeq B_m$ and, since $B_m \nto R(B_r)$ and $R(B_m) \toeq B_r$ by definition, we obtain that $L(B_l) \to R(B_m)$.  This is a contradiction, because $R(B_m)$ is good and $R(B_m) \to R(B_r)$.
    
    \refstepcounter{lemma:reception fails:claims}
    \item [Claim~\arabic{lemma:reception fails:claims}:] \label{lemma:reception fails:mid range} $\UU_r(B_l) \nto B_r$.  Otherwise, $B_m = F_r(B_l)$ satisfies \witnessref{2}.  Just note that, by Claim~\ref{lemma:reception fails:short range}, $B_m \nto B_r$.
  \end{description}

  Let $B_m = \UU_r(B_l)$ and $B_r' = \UU_r(B_m)$.  By Claims~\ref{lemma:reception fails:minimal range}--\ref{lemma:reception fails:mid range}, we observe that $B_l$, $F_r(B_l)$, $B_m$, $F_r(B_m)$, $B_r'$, and $B_r$ appear in this order in a traversal of $[B_l, B_r]$ where, possibly, $B_l = F_r(B_l)$, $B_m = F_r(B_m)$, and $B_r' = B_r$.  In $O(d(v))$ time we can check whether $B_r' \to B_l$; if negative, then $\{v, B_l, B_m, B_r'\}$ induces a $K_{1,3}$.  Thus, we assume that $B_r' \to B_l$.  Hence, $B_r'$ is not an end block and, by Claim~\ref{lemma:reception fails:mid range}, $B_r', R(B_r), B_l$ are pairwise different and appear in this order in a traversal of the contig that contains $B_l$ and $B_r'$.
  
  \begin{description}
    \refstepcounter{lemma:reception fails:claims}
    \item [Claim~\arabic{lemma:reception fails:claims}:] \label{lemma:reception fails:opposite block} if $(B_r, B_l)$ has some semiblock $W$ that is indistinguishable to neither $B_l$ nor $B_r'$, then a minimally forbidden can be obtained in $O(d(v))$ time.  By~\receptionref{2}, there are $O(1)$ blocks that are indistinguishable to either $B_l$ or $B_r'$, thus $W$ can be obtained in $O(1)$ time.  Clearly, $v$ is not adjacent to $W$ by \receptionref{1}.  First we test if $B_r' \to F_r(B_l)$ by looking whether $F_r(B_r') = F_r(B_l)$.  If affirmative, then, since no pair of $B_r'$, $W$, and $B_l$ are indistinguishable, it follows that $W_a = U_l(W) \to B_r'$ and $W_b = U_l(B_l) \to W$.  Note that, consequently, $B_m$, $W_a$, $W_b$, and $B_r'$ are pairwise different.  Hence, \{$v$, $+B_l$, $+B_m$, $+W_a$, $+W_b$, $+B_r'$, $-W$\} induces an $\overline{H_4}$ (\ref{app:reception fails:opposite block:case 1}).  Suppose, then, that $B_r' \nto F_r(B_l)$, thus $B_l \neq F_r(B_l)$.  Then we check whether $W \to F_r(B_l)$, i.e., whether $F_r(W) = F_r(B_l)$.  If affirmative, then, as before, $W_b = U_l(B_l) \to W$ is different to $B_m$ and $B_r'$.  Then, \{$v$, $+B_l$, $+F_r(B_l)$, $+B_m$, $+W_b$, $+B_r'$, $-W$\} is a forbidden (\ref{app:reception fails:opposite block:case 2}).  When $W \nto F_r(B_l)$, we test if $W_a = F_l(B_r') \to W$ by looking if $F_l(W) = W_a$.  If false, then $B_r' \neq W_a$ and \{$v$, $+B_l$, $+F_r(B_l)$, $+B_m$, $+W_a$, $+B_r'$, $-W$\} is a forbidden (\ref{app:reception fails:opposite block:case 3}).  Otherwise, since $W$ and $B_r'$ are not indistinguishable, we obtain that $B_r' \nto F_r(W)$.  Consequently, $F_r(W) \in (B_l, F_r(B_l))$ and \{$v$, $+B_l$, $+F_r(W)$, $+F_r(B_l)$, $+B_m$, $+B_r'$, $-W$\} is a forbidden (\ref{app:reception fails:opposite block:case 4}).  Note that every computed forbidden is adequate by Observation~\ref{obs:basic queries}.
  \end{description}

  By Claim~\ref{lemma:reception fails:opposite block}, we may assume that every semiblock in $(B_r, B_l)$ is indistinguishable to either $B_r'$ or $B_l$.  By~\receptionref{2}, $(B_r, B_l)$ has at most two semiblocks, namely $R(B_r)$ and $L(B_l)$, that are indistinguishable to $B_l$ and $B_r$.  Then, two cases remain.

  \begin{description}
    \item [Case 1:] $\#(B_r, B_l) = 2$.  In this case, $\{B_l, L(B_l)\}$ and $\{B_r, R(B_r)\}$ are both pairs of indistinguishable semiblocks and, since $B_l$ y $B_r$ are not indistinguishable, then either $B_r \nto F_r(B_l)$ or $F_l(B_r) \nto B_l$.  Suppose the former, as the other case is analogous.  Then $B_l \neq F_r(B_l)$ and \{$v$, $-L(B_l)$, $+B_l$, $+F_r(B_l)$, $+B_m$, $+B_r$, $-R(B_r)$\} is an adequate forbidden (\ref{app:reception fails:case 1}).
    
    \item [Case 2:] $\#(B_r, B_l) = 1$.  Applying Claims \ref{lemma:reception fails:minimal range}--\ref{lemma:reception fails:opposite block} to $\Phi^{-1}$, we observe that $R(B_r)$ is indistinguishable to either $B_r$ or $B_l' = \UU_l^2(B_r)$.  As $(B_r, B_l) \subseteq (B_r', B_l')$, we observe that either $B_r = B_r'$ is indistinguishable to $R(B_r)$ or $B_l = B_l'$ is indistinguishable to $R(B_r)$.  Assume the former, as the other case is analogous.  Now, if $\UU_r(B_r) \nto F_r(B_m)$, then $\UU_r(B_r) \neq B_m$ and \{$v$, $+B_l$, $+\UU_r(B_r)$, $+B_m$, $+F_r(B_m)$, $+B_r$, $-R(B_r)$\} is an adequate forbidden (possibly $B_m = F_r(B_m)$; \ref{app:reception fails:case 2}).  Suppose, then, that $\UU_r(B_r) \to F_r(B_m) = \LL(B_r)$, and let $\Psi$ be the co-contig that is obtained from $\Phi$ by exchanging the order between $B_r$ and $R(B_r)$. (For the rest of the proof, whenever we write $f$ without a superscript we mean $f^\Phi$.)  Clearly, $\Psi$ represents $H \setminus \{v\}$, $N(v) = [B_r, F_r(B_m)]$ in $\Psi$, while $B_r$ and $R(B_r)$ are the only possible pair of indistinguishable semiblocks.  That is, $\Bip{B_r}{F_r(B_m)}$ satisfies \receptionref{1}--\receptionref{2} for $\Psi$.  Moreover, since $\UU_r(B_r) \to F_r(B_m)$, we obtain, by Claim~\ref{lemma:reception fails:mid range} applied to $\Psi$, that $\Bip{B_r}{F_r(B_m)}$ satisfies \receptionref{3}.
  \end{description}
\end{proof}

The main theorem of this section follows by combining Lemmas \ref{lemma:bad block}, \ref{cor:consecutive},~and~\ref{lemma:reception fails} while filling the missing cases.

\begin{theorem}\label{thm:both prime}
  Let $H$ be a graph with a non-universal vertex $v \in V(H)$, and $\Gamma$ be block co-contig representing $H \setminus \{v\}$.  Given $\Gamma$ and $N(v)$, it costs $O(d(v) + \epTest)$ time to determine if $H$ is PCA.  Furthermore, within the same amount of time, $\Gamma$ can be transformed into a block co-contig representing $H$, unless a minimally forbidden of $H$ is obtained.
\end{theorem}

\begin{proof}
  First suppose $v$ is isolated in $H$.  In this case there are three possibilities.  First, if $\Gamma$ is straight, then $H$ is PCA and $\Gamma \cup \{\psi\}$ is a block co-contig representing $H$ for any contig $\psi$ whose only vertex is $v$.  Second, if $\Gamma = \{\gamma\}$ is not straight and $|\MP(\gamma)| \geq 4$, then $\MP(\gamma) \cup \{v\}$ induces a cycle plus an isolated vertex.  Thus, $\Bip{\MP(\gamma)}{\emptyset}$ is a minimally forbidden of $H$.  Finally, if $\Gamma = \{\gamma\}$ is not straight and $\MP(\gamma) = B_1, B_2, B_3$, then we can obtain a minimally forbidden by invoking Lemma~\ref{lemma:subprocedures} (\ref{lemma:subprocedures:two triangle}) with input $B_1, B_2, B_3$.
    
  Now suppose $N(v) \neq \emptyset$.  To compute a co-contig representing $H$, we first apply Lemma~\ref{lemma:bad block} with input $N(v)$ to verify that $\Gamma$ has only good blocks.  If negative, then we obtain a minimally forbidden.  Otherwise, we apply Lemma~\ref{cor:consecutive} with input $N(v)$ to transform $\Gamma$ into a block co-contig $\Gamma'$ representing $H \setminus \{v\}$ in which $N(v)$ is consecutive.  This time we obtain either a minimally forbidden or a pair of blocks $\Bip{B_a}{B_b}$ witnessing that $N(v)$ is consecutive in $\Gamma'$.  In the latter case, we apply Lemma~\ref{lemma:reception fails} with input $\Bip{B_a}{B_b}$ to transform $\Gamma'$ into a co-contig $\Phi$ representing $H \setminus \{v\}$ that satisfies \receptionref{1}--\receptionref{3}.  By Lemma~\ref{lem:receptive representation}, we obtain either a minimally forbidden or a pair $\Bip{B_l}{B_r}$ that is $v$-receptive in $\{\phi,\phi_r\}$, where $\phi, \phi_r \in \Phi$ are the contigs that contain $B_l$ and $B_r$, respectively.  Finally, we check if $\Bip{B_l}{B_r}$ is receptive in $\Phi$ and we proceed as follows according to the answer.
  \begin{description}
    \item [Case 1:] $\Bip{B_l}{B_r}$ is receptive in $\Phi$.  We first transform $\Phi$ into the $\{v\}$-reception $\Psi$ of $\Bip{B_l}{B_r}$ in $\Phi$.  Then, we compact $\{v\}$ if it is indistinguishable to either $R^\Psi(\{v\})$ or $L^\Psi(\{v\})$.  As a result, $\Psi$ is a round block representation of $H$.
    \item [Case 2:] $\Bip{B_l}{B_r}$ is not receptive in $\Phi$.  As discussed in Section~\ref{sec:data structure:round representations}, the only case in which $\Bip{B_l}{B_r}$ is receptive in $\{\phi, \phi_r\}$ and not receptive in $\Phi$ is when $\phi = \phi_r$, $(B_l, B_r]$ contains the left end semiblock $W_l$ of $\phi$, and $\Phi\setminus \{\phi\}$ has some contig $\rho$.  Moreover, the $v$-reception $\{\psi\}$ of $\{\phi\}$ is not straight in this case.  Thus, $H[V(\psi) \cup \{x\}]$ is not PCA for every $x \in V(\rho)$.  Then, we can obtain a minimally forbidden of $H$ by trying to insert a vertex $x \in V(\rho)$ into $\{\psi\}$.  Since $x$ has no neighbors in $V(\psi)$, this insertion trial costs $O(1)$ time, while we can find $x \in V(\rho)$ in $O(1)$ time by using the representation pointer of $W_l$.
  \end{description}
  As discussed in Section~\ref{sec:data structure:round representations}, $O(d(v) + \epTest)$ time suffices to test if $\Bip{B_l}{B_r}$ is receptive in $\Phi$ and to compute the $\{v\}$-reception of $\Bip{B_l}{B_r}$ in $\Phi$ and $\{\phi\}$.  By Lemmas~\ref{lemma:bad block}, \ref{cor:consecutive},~and~\ref{lemma:reception fails}, we conclude that the whole algorithm costs $O(d(v) + \epTest)$ time.
\end{proof}

\subsection{\texorpdfstring{$H$ and $G$ need not be co-connected}{{\it H} and {\it G} need not be co-connected}}

In this section we deal with the general case in which $H$ and $G = H \setminus \{v\}$ need not be co-connected.  In other words, $G = G[N] + G[V_1]+ \ldots + G[V_k]$ where:
\begin{itemize}
  \item $V_u$ contains the universal vertices of $G$ in $V(G) \setminus N(v)$, for some $1 \leq u \leq k$,
  \item For $i \neq u$, $G[V_i]$ is a co-component with $V_i \setminus N(v) \neq \emptyset$, and
  \item $N = V(G) \setminus (V_1 \cup \ldots \cup V_k)$ contains only vertices in $N(v)$.
\end{itemize}
We are taking a loose definition of $G[\bullet]$ and $+$ here, as it could happen that $V_u = \emptyset$, $N = \emptyset$, or $k = 1$; the missing details are obvious though.  The algorithm in~\cite{SoulignacA2015} builds a round block representation of $H$ in two phases.  The first phase finds a block co-contig $\psi_v$ of $H[W_k \cup\{v\}]$, where $W_j = \bigcup_{i=1}^j V_i$ for every $0 \leq j \leq k$.  The second phase joins $\psi_v$ and a round block representation $\Gamma_N$ of $H \setminus V$ into a round block representation $\Psi$ of $H$.  Our certifying algorithm mimics these two phases; the internal details are different, though.

The purpose of the first phase is to find a co-contig $\psi_v$ of $H[W_k \cup \{v\}]$.  To fulfill its goal, the algorithm in~\cite{SoulignacA2015} computes all the round block representations of $H[W_k]$ to see if $N(v)$ is consecutive in one of them.  For those in which $N(v)$ is consecutive, it checks if some of its $v$-associated representations is $v$-receptive.  The algorithm is correct by the Reception Theorem and Observation~\ref{obs:v-associated}, but it could require exponential time.  A key observation in~\cite{SoulignacA2015} is that $H$ is not PCA when $k > 3$, thus only $O(1)$ round representations need to be examined, hence the algorithm is efficient.  The problem with this ``brute force'' strategy is that it makes it difficult to find a negative certificate when $H$ is not PCA.  An alternative approach is to note that, as $H[W_k \cup \{v\}]$ is co-connected, at most two of the generated representations, $\Phi$ and $\Phi^{-1}$, are $v$-receptive.  The idea is to characterize how does $\Phi$ look like so that a minimally forbidden can be obtained when $H[W_k]$ has no $v$-receptive representations.  

Instead of dealing with $H[W_k] = H[V_1] + \ldots + H[V_k]$ as a whole, we use an iterative approach.  Before the algorithm is executed, we have a round block representation $\Gamma_0 = \Gamma$ of $G$ and we build a new block co-contig $\Psi_v$ of $H[\{v\}]$.  After $i$ iterations, we have transformed $\Gamma$ into a round block representation $\Gamma_{i}$ of $G \setminus W_i$ and $\psi_v$ into a block co-contig of $H[W_i \cup \{v\}]$.  To cope with iteration $i+1$, we use Steps~1--3 below.  In brief terms, this procedure works a follows:
\begin{description}
  \item [Step 1] splits from $\Gamma_i$ a block co-contig $\gamma_{i+1}$ having blocks co-adjacent to $v$. Let $V_{i+1} = V(\gamma_{i+1})$.
  \item [Step 2] updates $\gamma_{i+1}$ into a block co-contig $\psi_{i+1}$ of $H[V_{i+1} \cup \{v\}]$.
  \item [Step 3] joins $\psi_{i+1}$ and $\psi_v$ to obtain a block co-contig of $H[W_{i+1} \cup \{v\}]$.
\end{description}
Once the iterative process is completed, we have round block representations $\Gamma_k$ of $G \setminus W_k$ and $\psi_v$ of $H[W_k \cup \{v\}]$.  We use Phase~2 below to combine these representations into a representation of $H$.  Of course, any of these steps can fail, and a minimally forbidden is provided if so.

\subsubsection{Step 1: split \texorpdfstring{$\gamma_{i+1}$}{\textbackslash gamma\_\{i+1\}} out of \texorpdfstring{$\Gamma_i$}{\textbackslash Gamma\_i}}

To split $\gamma_{i+1}$ out of $\Gamma_i$, we traverse $\B(\Gamma_{i})$ until the first block $B$ co-adjacent to $v$ is found.  If no such block exists, then Phase~1 concludes and Phase~2 begins.  Otherwise, we invoke Lemma~\ref{lemma:is cobipartite} below to obtain the family $E$ of co-end blocks of $\gamma_{i+1}$, where $\gamma_{i+1}$ is the co-contig of $\Gamma_i$ that contains $B$.  If Lemma~\ref{lemma:is cobipartite} outputs a minimally forbidden, then the algorithm halts; otherwise, we check if $B$ is a universal block.  If affirmative, then we separate $B$ into $+B$ and $-B$, and we update $\gamma_{i+1}$ to be the co-contig containing $-B$.  The separation is done in $O(|{+B}|)$ time, as discussed in Section~\ref{sec:data structure:round representations}.  Finally, we split $\gamma_{i+1}$ out of $\Gamma_{i}$.  Note that the case $E = \emptyset$ is trivial, as $\gamma_{i+1} = \Gamma_i$ and $\Gamma_{i+1} = \emptyset$, while the split when $E \neq \emptyset$ costs $O(1)$ time as discussed in Section~\ref{sec:data structure:contig}.  Therefore, Step~1 costs $O(d(v))$ time.

\begin{lemma}\label{lemma:is cobipartite}
  Let $H$ be a graph with a vertex $v$, $\phi$ be a co-contig of a round representation $\Phi$ of $H \setminus \{v\}$, and $B \in \B(\phi)$ be co-adjacent to $v$.  Given a $\Phi$-pointer to $B$, it takes $O(d(v))$ time to determine if $G(\Phi)$ is co-bipartite when $H$ is PCA.  The algorithm outputs either a minimally forbidden of $H$ or a set containing $\Phi$-pointers to all the co-end semiblocks of $\phi$.
\end{lemma}

\begin{proof}
  The algorithm outputs $\emptyset$ when $\Phi$ is not robust, and $\{\DS{B}\}$ when $B$ is universal.  In the remaining case, the algorithm tries to locate the left co-end semiblocks of $\phi$.  For this, it computes the minimum $i \geq 0$ such that:
  \begin{enumerate}
    \item $\UU_r^i(B)$ is a left co-end semiblock,\label{lemma:is cobipartite:co-end}
    \item $\UU_r^i(B) = \UU_r^j(B)$ for some $j < i$, or\label{lemma:is cobipartite:cycle}
    \item $i \geq 5$ and $v$ is co-adjacent to $\UU_r^{i-5}(B)$, $\UU_r^{i-4}(B)$, $\UU_r^{i-2}(B)$, and $\UU_r^{i-1}(B)$.\label{lemma:is cobipartite:forbidden}
  \end{enumerate}
  Observe that $\UU_r^i(B) \in \B(\phi)$ because $B$ is not universal.  Therefore: if \ref{lemma:is cobipartite:co-end}.\ holds, then $\UU_r^i(B)$ and $\UU_r^{i+1}(B)$ are the left co-end semiblock of $\phi$; if \ref{lemma:is cobipartite:cycle}.\ holds, then $G(\Phi)$ is not co-bipartite because $\UU_r^j(B), \ldots, \UU_r^i(B)$ induces a co-cycle of odd length; and if \ref{lemma:is cobipartite:forbidden}.\ holds (and \ref{lemma:is cobipartite:cycle}.\ does not), then $H$ is not PCA because \{$v$, $-\UU_r^{i-5}(B)$, $-\UU_r^{i-4}(B)$, $-\UU_r^{i-2}(B)$, $-\UU_r^{i-1}(B)$\} induces a $C_4^*$.  Clearly, $i$ can be obtained in $O(d(v))$ time.  Indeed, each semiblock is traversed $O(1)$ times by \ref{lemma:is cobipartite:co-end}.\ and \ref{lemma:is cobipartite:cycle}., while at most $6d(v)$ blocks co-adjacent to $v$ are visited by \ref{lemma:is cobipartite:forbidden}. (See~\cite{SoulignacA2015} for a better bound.)  When \ref{lemma:is cobipartite:co-end}.\ holds, the algorithm computes the right co-end semiblocks of $\phi$ by replacing $\UU_r$ with $\UU_l$ in \ref{lemma:is cobipartite:co-end}--\ref{lemma:is cobipartite:forbidden}.   
\end{proof}

\subsubsection{Step 2: update of \texorpdfstring{$\gamma_{i+1}$}{\textbackslash gamma\_\{i+1\}} into \texorpdfstring{$\psi_{i+1}$}{\textbackslash psi\_\{i+1\}}}

There are two possibilities for Step~2, according to whether $\gamma_{i+1}$ has a unique (universal) block or not.  In the former case, $\{v\}$ is a block of $H[V_{i+1} \cup \{v\}]$ co-adjacent to the clique $V_{i+1}$, thus computing the block co-contig $\psi_{i+1}$ in $O(1)$ time is trivial.  In the latter case, both $H[V_{i+1}]$ and $H[V_{i+1} \cup \{v_i\}]$ are co-connected.  Thus, we invoke Theorem~\ref{thm:both prime}, with input $\gamma_{i+1}$ and $V_{i+1}\cap N(v)$, to transform $\gamma_{i+1}$ into a round block representation $\Psi_{i+1}$ of $H[V_{i+1} \cup \{v\}]$.  By Lemma~\ref{lemma:restricted neighborhood}, $O(d(v))$ time suffices to compute $V_{i+1} \cap N(v)$, thus Step~2 requires $O(d(v) + \epTest)$ time.

\begin{lemma}\label{lemma:restricted neighborhood}
  Let $H$ be a graph with a vertex $v$, $\phi$ be a co-contig of a round representation $\Phi$ of $H \setminus \{v\}$, and $B_l$ be a left co-end block of $\phi$.  Given $N(v)$ and a semiblock pointer to $B_l$, it takes $O(d(v))$ time to compute $V(\phi) \cap N(v)$ when $H$ is PCA.  When $H$ is not PCA, the algorithm outputs either $V(\phi) \cap N(v)$ or a minimally forbidden of $H$.
\end{lemma}

\begin{proof}
  For each $B \in \B(\Phi)$ adjacent to $v$, we find a pointer $E(B)$ to a left co-end semiblock of the co-contig $\phi_B$ that contains $B$; initially, $E(B) = \bot$.  To compute $E$, we traverse each $w \in N(v)$ to process the semiblock $B$ that contains $w$.  If $B$ is universal, then we set $E(B) = B$ and pass to the next vertex.  Otherwise, we look for the minimum $i \geq 0$ such that:
  \begin{enumerate}
    \item $E(\UU_r^i(B)) \neq \bot$,\label{lemma:restricted neighborhood:pointer}
    \item $\UU_r^i(B)$ is a left co-end semiblock, or\label{lemma:restricted neighborhood:co-end}
    \item $i \geq 4$ and $v$ is co-adjacent to $\UU_r^{i-4}(B)$, $\UU_r^{i-3}(B)$, $\UU_r^{i-1}(B)$, and $\UU_r^i(B)$.\label{lemma:restricted neighborhood:forbidden}
  \end{enumerate}
  Since $B$ is not universal, it follows that $B$ and $\UU_r^j(B)$ belong to the same co-component for every $0 \leq j \leq i$.  Hence, $E(\UU_r^i(B))$ is a left co-end semiblock of $\phi_B$ if~\ref{lemma:restricted neighborhood:pointer}., while $\UU_r^i(B)$ is a left co-end semiblock of $\phi_B$ if~\ref{lemma:restricted neighborhood:co-end}.  Therefore: if \ref{lemma:restricted neighborhood:pointer}., then we set $E(\UU_r^j(B)) = E(\UU_r^i(B))$ for every $0 \leq j \leq i$; if~\ref{lemma:restricted neighborhood:co-end}., then we set $E(\UU_r^j(B)) = \UU_r^i(B)$ for every $0 \leq j \leq i$; and if~\ref{lemma:restricted neighborhood:forbidden}., then we output that $H$ is not PCA because \{$v$, $-\UU_r^{i-4}(B)$, $-\UU_r^{i-3}(B)$, $-\UU_r^{i-1}(B)$, $-\UU_r^{i}(B)$\} induces a $C_4^*$.  The computation of $E$ ends after all the vertices in $N(v)$ have been considered.  After $E$ is computed, the algorithm outputs $V(\phi) \cap N(v) = \{w \in N(v) \mid E(B(w)) \in \{B_l, U_r(B_l)\}\}$, where $B(w)$ is the semiblock that contains $w$.  Clearly, by \ref{lemma:restricted neighborhood:pointer}.~and~\ref{lemma:restricted neighborhood:co-end}., the algorithm traverses each semiblock $B$ adjacent to $v$ only $O(|B \cap N(v)|)$ times, while, by~\ref{lemma:restricted neighborhood:forbidden}., it traverses at most $5d(v)$ blocks co-adjacent to $v$. 
\end{proof}

\subsubsection{Step 3: join of \texorpdfstring{$\psi_{i+1}$}{\textbackslash psi\_\{i+1\}} and \texorpdfstring{$\psi_v$}{\textbackslash psi\_v}}

Step~3 has to join $\psi_v$ and $\psi_{i+1}$ into a block co-contig representing $H[W_{i+1} \cup \{v\}]$.  This is trivial when $i = 1$ as we replace $\psi_v$ with $\psi_1$.  When $i > 1$, at most one between $H[W_i]$ and $H[V_{i+1}]$ is a clique.  Thus, we can combine $\psi_v$ and $\psi_{i+1}$ in $O(1)$ time with the following lemma.

\begin{lemma}\label{lemma:insertion join}
  Let $H$ be a co-connected graph with a vertex $v$ such that $H \setminus \{v\}$ is a PCA graph isomorphic to $H[V_1] + H[V_2]$ for some $\emptyset \subset V_1 \subset V(H)$ and $V_2 = V(H) \setminus (V_1 \cup \{v\})$, and let $B_i$ be the block that contains $v$ in a block co-contig $\psi_i$ representing $H[V_i \cup \{v\}]$, for $i \in \{1,2\}$.  Suppose $V_1$ is not a block of $\psi_1$.  Then, $H$ is a PCA graph if and only if either:
  \begin{enumerate}[(i)]
    \item $B_1$ and $B_2$ are co-end blocks of $\psi_1$ and $\psi_2$,\label{lemma:insertion join:co-end}
    \item $V_2$ is a block of $\psi_2$, $\psi_1$ is robust, and $F_r(\RR(B_1)) = U_l(\LL(B_1)) \neq \bot$, or\label{lemma:insertion join:clique}
    \item $V_2$ is a block of $\psi_2$ and $V_1$ has exactly three non-adjacent blocks: \{$v$, $W_2$, $W_3$\}.\label{lemma:insertion join:O(1)}
  \end{enumerate}
  Consequently, $O(1)$ time suffices to determine if $H$ is PCA, when $\psi_i$-pointers to $B_i$ are given.  The algorithm either transforms $\psi_1$ and $\psi_2$ into a block co-contig representing $H$ or outputs a minimally forbidden of $H$.
\end{lemma}

\begin{proof}
  First we prove that $H$ is a PCA graph when some of (\ref{lemma:insertion join:co-end})--(\ref{lemma:insertion join:O(1)}) holds.
  \begin{description}
    \item [(\ref{lemma:insertion join:co-end}) holds.]  The proof is implicit in~\cite{SoulignacA2015}.  By reversing $\psi_1$ and $\psi_2$ if required, suppose $B_1$ is a right co-end block and $B_2$ is a left co-end block.  As discussed in Section~\ref{sec:data structure:round representations}, we can join $\psi_1$ and $\psi_2$ into a block co-contig $\rho$ representing $G(\rho) = G(\psi_1) + G(\psi_2) = H[V_1 \cup \{v\}] + H[V_2 \cup \{v\}]$ in which $B_1$ and $B_2$ are consecutive.  Clearly, $B_1$ witnesses that $\Bip{F_l(B_1)}{F_r(B_2)}$ is receptive in $\rho$, thus the $\{w\}$-reception of $\Bip{F_l(B_1)}{F_r(B_2)}$ is a block co-contig representing a graph $H'$ with three vertices $v_1 \in B_1$, $v_2 \in B_2$ and $w$ such that: $N(w) = N_H(v)$ and $H' \setminus \{v_1, v_2, w\} = H \setminus \{v\}$.  Consequently, $H = H' \setminus \{v_1, v_2\}$ is PCA.
    
    \item [(\ref{lemma:insertion join:clique}) holds.]  Let $\rho$ be the co-contig that is obtained from $\psi_1$ by a separation of $B_1$ into $\Bip{B_1 \setminus \{v\}}{\{v\}}$, and note that $B_m = F_r^\rho(\RR^\rho(\{v\}))$ witnesses that $\Bip{\RR^\rho(\{v\})}{\LL^\rho(\{v\})}$ is receptive in $\rho$.  Consequently, the $V_2$-reception of $\Bip{\RR^\rho(\{v\})}{\LL^\rho(\{v\})}$ is a block co-contig that represents $H$ because $v$ is the unique vertex not adjacent to $V_2$.
    
    \item [(\ref{lemma:insertion join:O(1)}) holds.] Trivial.
  \end{description}
  It is not hard to obtain the block co-contig $\rho$ in $O(1)$ time using the operations described in Section~\ref{sec:data structure:round representations} with some low-level manipulation of the contigs (i.e., avoiding \texttt{reception}); see~\cite{SoulignacA2015}.
    
  Now suppose none of (\ref{lemma:insertion join:co-end})--(\ref{lemma:insertion join:O(1)}) holds.  To prove that $H$ is not PCA we show an $O(1)$ time algorithm that computes a minimally forbidden of $H$.  If $B_2$ is not a co-end block, then $V_2$ is not a block, thus we may replace $V_1$ and $V_2$ without affecting the hypothesis of the lemma.  Hence, as (\ref{lemma:insertion join:co-end}) is false, we suppose $B_1$ is not a co-end block.  Through the proof we work only with $\psi_1$ and two blocks of $\psi_2$, called $X$ and $Y$, which are not adjacent to $B_2$ and $X$, respectively.  Also, $Y \neq \{v\}$ unless $V_2 = X$.  Note that $X$ and $Y$ are obtainable in $O(1)$ time.
  
  The first step of the algorithm is to verify if $\psi_1$ is robust.  By Theorem~\ref{thm:forbiddens PCA}, $H[V_1]$ is co-bipartite because $H \setminus \{v\}$ is not co-connected.  Then, $\psi_1 \setminus \{B_1\}$ is robust, thus either $\psi_1$ is robust or $B_1 = \{v\}$ is isolated in $\psi_1$ and $\psi_1 \setminus \{B_1\}$ has exactly two non-adjacent blocks $W_2$, $W_3$.  Therefore, we can decide if $\psi_1$ is robust in $O(1)$ time, obtaining pointers to $W_2$ and $W_3$ if negative.  Moreover, $Y \neq \{v\}$ because (\ref{lemma:insertion join:O(1)}) is false.  Consequently, $\{B_1, X, W_2, W_3, Y \setminus\{v\}\}$ contains either a $K_{1,3}$ or a $C_4^*$.  Such a minimally forbidden can be obtained in $O(1)$ time.  From now on we assume $\psi_1$ is robust, hence $\LL$, $\RR$, $\UU_l$, and $\UU_r$ are well defined for $\psi_1$.  Moreover, as $B_1$ is not a co-end block, we obtain that $\UU_r(B_1) \nto \LL(B_1)$ and $\RR(B_1) \nto \UU_l(B_1)$.  
  
  The second step is to check if $\UU_l(B_1) \toeq \LL(B_1)$ and if $\RR(B_1) \toeq \UU_r(B_1)$.  If $\UU_l(B_1) \neq \LL(B_1)$ and $\UU_l(B_1) \nto \LL(B_1)$, then $\LL(B_1) \nto F_r(B_1)$ because $B_1$ and $\LL(B_1)$ are not indistinguishable.   So, $F_r(B_1) \to \UU_l(B_1)$ because, otherwise, $\UU_l(B_1)$, $\LL(B_1)$, $F_r(B_1)$ are pairwise non-adjacent blocks, contradicting the fact that $H[V_1]$ is co-bipartite.  Similarly, $\LL(B_1) \to \RR(B_1)$ because $\UU_l(B_1)$, $\LL(B_1)$, $\RR(B_1)$ cannot be pairwise non-adjacent.  Hence, $B_1$, $\RR(B_1)$, $F_r(B_1)$, $U_r(B_1)$ are pairwise different and appear in this order in a traversal of $[B_1, \UU_r(B_1)]$.  The minimally forbidden we generate depends on whether $\RR(B_1) \to \UU_r(B_1)$ or not.  In the affirmative case, \{$B_1$, $\LL(B_1)$, $\RR(B_1)$, $F_r(B_1)$, $\UU_r(B_1)$, $\UU_l(B_1)$, $X$\} induces an $\overline{H_4}$ (\ref{app:insertion join:case 1}).  In the negative case we observe that, as before, $F_l(B_1) \nto \RR(B_1)$ and $F_r(B_1) \to \UU_l(B_1)$.  This implies $F_l(B_1) \to F_r(B_1)$ because \{$F_l(B_1)$, $\LL(B_1)$, $\RR(B_1)$, $F_r(B_1)$, $\UU_r(B_1)$\} does not induce a $C_5$, thus \{$B_1$, $\LL(B_1)$, $\RR(B_1)$, $F_r(B_1)$, $F_l(B_1)$, $\UU_l(B_1)$, $X$\} induces an $\overline{H_5}$ (\ref{app:insertion join:case 2}).  From now on, we assume $\UU_l(B_1) \toeq \LL(B_1)$ and, similarly, $\RR(B_1) \toeq \UU_r(B_1)$.  Hence, $\UU_l(B_1) \neq \UU_r(B_1)$.

  Note that $F_r(\RR(B_1)), F_l(\LL(B_1))$ are either equal or appear in this order in a traversal of $[B_1, \LL(B_1)]$; otherwise, any block inside $(F_l(\LL(B_1)), F_r(\RR(B_1)))$ would be indistinguishable to $F_r(\RR(B_1))$.  For the third step, the algorithm tests if $(F_r(\RR(B_1)), F_l(\LL(B_1)))$ has some block $W$.  If affirmative, then $\LL(B_1) \to \RR(B_1)$ since, otherwise, $W$, $\LL(B_1)$, $\RR(B_1)$ are pairwise non-adjacent.  Consequently, $\LL(B_1)$, $\RR(B_1)$, $\UU_r(B_1)$, and $\UU_l(B_1)$ are all different.  This implies that $\UU_r(B_1) \to \UU_l(B_1)$ because no subset of $\{W, \UU_l(B_1), \LL(B_1), \RR(B_1), \UU_r(B_1)\}$ induces an $\overline{C_3}$ or $C_5$.  Consequently \{$B_1$, $X$, $W$, $\UU_l(B_1)$, $\UU_r(B_1)$, $\LL(B_1)$, $\RR(B_1)$\} induces an $\overline{H_2}$ (\ref{app:insertion join:case 3}).  
  
  Finally, note that $Y \neq \{v\}$ when either $F_r(\RR(B_1)) = F_l(\LL(B_1))$ or $F_r(\RR(B_1)) = \UU_l(\LL(B_1))$.  Indeed, in the former case $Y \neq \{v\}$ because $F_r(\RR(B_1))$ and $X$ are not twins, while in the latter case $Y \neq \{v\}$ because (\ref{lemma:insertion join:clique}) is false.  Consequently, $\B = \{B_1, Y \setminus \{v\}$, $X$, $\UU_l(B_1)$, $\LL(B_1)$, $\RR(B_1)$, $\UU_r(B_1)$\} is a forbidden (\ref{app:insertion join:case 4}) whose edges can be obtained in $O(1)$ time.  We remark that not all the blocks in $\B$ are pairwise different.  
\end{proof}

\subsubsection{Phase 2: join of \texorpdfstring{$\Psi_v$}{\textbackslash Psi\_v} and \texorpdfstring{$\Gamma$}{\textbackslash Gamma}}

After the first phase is completed, we have a round block representation $\Gamma$ of $G \setminus W_k$ and a block co-contig $\psi_v$ representing $H[W_k \cup \{v\}]$ for $W_k = \bigcup_{i=1}^k V_i$.  The goal of the second phase is to find a round block representation of $H$. This is trivial when $W_k = V(G)$, as $\psi_v$ is the desired representation.  For the other case, we invoke Theorem~\ref{thm:H not co-connected} using $v$ and $w \in V(\Gamma)$ as input.

\begin{theorem}\label{thm:H not co-connected}
  Let $H$ be a graph such that $H = H[V_1] + H[V_2]$ for $\emptyset \subset V_1, V_2 \subset V(H)$, and $\Phi_i$ be a round block representation of $H[V_i]$, for $i \in \{1,2\}$.  Then, $H$ is PCA if and only if $H[V_1]$ and $H[V_2]$ are PCA and co-bipartite.  Furthermore, if semiblock $\Phi_i$-pointers to $B_i\in \B(\Phi_i)$ are given, then $O(|N(B_i)|)$ time suffices to determine if $H$ is PCA.  The algorithm either transforms $\Phi_1$ and $\Phi_2$ into a round block representation of $H$ or outputs a minimally forbidden of $H$.
\end{theorem}

\begin{proof}
  The fact that $H$ is PCA if and only if $H[V_1]$ and $H[V_2]$ are co-bipartite PCA graphs follows from Theorem~\ref{thm:forbiddens PCA}.  
  
  The algorithm to detect if $H$ is PCA is as follows.  Let $B_i$ be any block of $\Phi_i$ ($\{i,j\} = \{1,2\}$), and $q$ be the minimum such that either $U_r^{2q}(B_i) \nto B_i$ or $U_r^{2q}(B_i) = U_r^{2q-2}(B_i)$.  Note that, since $U_r^{2p}(B_i) \toeq B_i$, the blocks $B_i, U_r^{2p+1}(B_i), U_r^{2p+2}(B_i), U_r^{2p}(B_i)$ appear in this order in $\B(\Phi_i)$ for every $0 \leq p < q$ (see Figure~\ref{fig:H not co-connected}~(a)).  Consequently, the value $q$ is well defined, and the blocks of $\Phi_i$ appear as in Figure~\ref{fig:H not co-connected}~(b).  Therefore, if $p$ is the maximum such that $U_r^{2q}(B_i) \nto U_r^{2p}(B_i)$, then either $p = q-1$ or $\B = \{U_r^{2p}(B_i), \ldots, U_r^{2q}(B_i)\}$ induces an odd co-cycle.  In the former case, $U_r^{2q}(B_i) = U_r^{2q-2}(B_i)$ is a co-end block, while, in the latter case, $\B \cup \{B_j\}$ is a minimally forbidden of $H$ for every $B_j \in \B(\Phi_j)$.  Replacing $i$ by $j$, we can find a minimally forbidden when $\Phi_j$ has no co-end blocks.  When both $\Phi_1$ and $\Phi_2$ have co-end blocks, we can {\tt join} $\Phi_1$ and $\Phi_2$ into a round block representation of $H$ as in Section~\ref{sec:data structure:round representations}.
  
      \begin{figure}
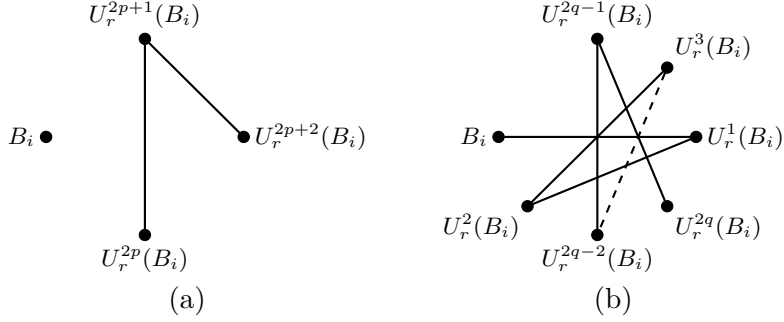

      \centering
      \begin{tabular}{c@{\hspace{1cm}}c}
        \input{srcFigHnotCoconnectedOrderI} & \input{srcFigHnotCoconnectedOrder} \\
         (a)  &  (b)
      \end{tabular}
      \caption{Adjacencies of $\overline{H[V_i]}$ in Theorem~\ref{thm:H not co-connected}.  The blocks are drawn as they appear in the circular ordering $\B(\Phi_i)$.  Note that $U_r^{2q}(B_i) = U_r^{2q-2}(B_i)$ when $H[V_i]$ is co-bipartite.}\label{fig:H not co-connected}
    \end{figure}

  To compute the sequence $B_i, \ldots, U_r^{2q}(B_i)$ we proceed as follows.  First, we mark all the blocks in $[B_i, F_r(B_i)]$. Then, $U_r^{2i}(B_i) \to B_i$ if and only if $F_r(U_r^{2i}(B_i))$ is marked; thus, $q$ is the minimum value such that $F_r(U_r^{2q}(B))$ is not marked.  Then, to obtain the value $p$, first note that $p = 0$ if $q = 1$.  When $q > 1$, we traverse $[F_l(B_i), B_i]$ while looking for $F_r(U_r^{2q}(B_i))$; then $U_r^{2p+2}(B_i)$ is the last block of the traversed sequence.  Since $B_i$ and $B_j$ are adjacent to every block in $[F_l(B_i), F_r(B_i)]$, the cost of this algorithm is $O(\min\{|N(B_i)|, |N(B_j)|\})$.
\end{proof}

\section{The certifying recognition algorithms}
\label{sec:recognition algorithms}

By Theorem~\ref{thm:forbiddens PCA}, at most three iterations of Phase~1 in Section~\ref{sec:incremental} can be completed without finding a minimally forbidden.  Hence, since each iteration of Phase~1 costs $O(d(v) + \epTest)$ time, and Phase~2 costs $O(d(v))$ time, we obtain the main result of the previous section: there is an $O(d(v) + \epTest)$ algorithm that transforms a round block representation $\Gamma$ of $H \setminus \{v\}$ into a round block representation $\Psi$ of $H$, unless a minimally forbidden is obtained.  Note that the algorithm ignores the straightness invariant of $\Gamma$, and it does not ensure the straightness invariant for $\Psi$.  The straightness invariant, instead, is required for the recognition of PIG graphs.  Fortunately, we can restore the straightness invariant in $O(1)$ time with Corollary~\ref{cor:straightness invariant} below.  Before describing this corollary, we define what a locally straight representation is.

Recall that a semiblock $B$ of a round representation $\Phi$ is long when $F_r(B) \to F_l(B)$.  When no block of $\Phi$ is long, $\Phi$ is said to be a \emph{locally straight representation}.  A graph $G$ is a \emph{proper Helly circular-arc (PHCA)} graph if it is isomorphic to $G(\Phi)$ for some locally straight representation $\Phi$.  As it is shown in~\cite{LinSoulignacSzwarcfiterDAM2013}, $G$ is a PHCA graph if and only if it admits a PCA model in which no two nor three arcs cover the circle.  The following results imply Corollary~\ref{cor:straightness invariant} below.

\begin{theorem}[\cite{LinSoulignacSzwarcfiterDAM2013}]\label{thm:forbiddens PHCA}
  A PCA graph is a PHCA graph if and only if it contains no $W_4$ or $S_3$ as induced subgraphs, where $W_4$ is the graph obtained after inserting a universal vertex in $C_4$.
\end{theorem}

\begin{theorem}[\cite{LinSoulignacSzwarcfiterDAM2013}]\label{thm:characterization PHCA}
  If $B$ is the universal block of a contig $\phi$, then either 1.~$\phi$ is linear, 2.~$F_r(L(B)) = B$ or 3.~$G(\phi)$ is not PHCA.  If 2., then $F_r(B)$ witnesses that $\Bip{R(B)}{L(B)}$ is receptive in $\phi \setminus \{B\}$, and its $B$-reception is a linear contig representing $G(\phi)$.
\end{theorem}

\begin{lemma}[\cite{LinSoulignacSzwarcfiterDAM2013}]\label{lemma:certification PHCA}
  If a round representation $\Phi$ has three non-universal blocks $B_1$, $B_2$, $B_3$ such that $B_1 \to B_2$, $B_2 \to B_3$, and $B_3 \to B_1$, then $G(\Phi)$ is not PHCA.  
\end{lemma}

\begin{corollary}\label{cor:straightness invariant}
  Given a round block representation $\Psi$, it takes $O(1)$ time to transform $\Psi$ into a round block representation $\Psi'$ of $G(\Psi)$ that satisfies the straightness invariant.  Moreover, $\Psi'$ is locally straight when $G(\Psi)$ is PHCA.
\end{corollary}

\begin{proof}
  By Theorems~\ref{thm:forbiddens PIG}, \ref{thm:forbiddens PHCA},~\ref{thm:characterization PHCA}, and Lemma~\ref{lemma:certification PHCA}, the algorithm has nothing to do in the following situations because either 1.~$\Psi$ is straight, 2.~$\Psi$ is locally straight and $G(\Psi)$ has an induced cycle, or 3.~$G(\Psi)$ is not PHCA:
  \begin{itemize}
    \item $|\Psi| > 1$,
    \item $\Psi = \{\psi\}$ and $|\MP(\psi)| > 3$,
    \item $\Psi = \{\psi\}$, $|\MP(\psi)| = 3$ and no block of $\MP(\psi)$ is universal, or
    \item $\Psi = \{\psi\}$, $|\MP(\psi)| = 3$, $B \in \MP(\psi)$ is universal, and $F_r(L(B)) \neq B$.
  \end{itemize}
  Finally, if $\Psi = \{\psi\}$, $|\MP(\psi)| = 3$, $B \in \MP(\psi)$ is universal, and $F_r(L(B)) = B$, the algorithm moves $B$ to the position that follows $F_r(B)$ in a traversal of $\B(\Psi)$.  The block representation $\Psi'$ so obtained is straight by Theorem~\ref{thm:characterization PHCA}. Clearly, $O(1)$ time is enough to test the above conditions and to apply the required move using {\tt split} and {\tt join} (see Section~\ref{sec:data structure:round representations}).
\end{proof}

The main theorems of this article then follow.

\begin{theorem}
  Let $H$ be a graph with a vertex $v$, and $\Gamma$ be a round block representation of $H \setminus \{v\}$.  Given $\Gamma$ and $N(v)$, it takes $O(d(v) + \epTest)$ time to determine if $H$ is a PCA graph.  The algorithm transforms $\Gamma$ into a round block representation of $H$ that satisfies the straightness invariant, unless a minimally forbidden of $H$ is obtained.
\end{theorem}

\begin{theorem}
  When a vertex $v$ of a round block representation $\Psi$ is given, $O(d(v) + \epRemove)$ time is enough to transform $\Psi$ into a round block representation of $G(\Psi) \setminus \{v\}$ that satisfies the straightness invariant.
\end{theorem}

\begin{proof}
  Let $B$ be the block that contains $v$.  If $|B| > 1$, then we remove $v$ out of $B$; otherwise we call {\tt remove($B$)} to transform $\Psi$ into a round block representation $\Phi$ of $H \setminus \{v\}$.  Afterwards, we apply Corollary~\ref{cor:straightness invariant} on $\Phi$ to restore the straightness invariant.  
\end{proof}

\begin{theorem}
  Given a round block representation $\Gamma$ of a graph $H$ that satisfies the straightness invariant, it takes $O(1)$ time to determine if $H$ is PHCA.  If $H$ is not PHCA, then the algorithm outputs $\Gamma|\B$ for a family of blocks $\B$ such that $H[\B]$ is isomorphic to either $W_4$ or $S_3$.
\end{theorem}

\begin{proof}
  The algorithm answers yes when $|\Gamma| > 1$ or $\Gamma = \{\gamma\}$ and $|\MP(\gamma)| > 3$.  Conversely, if $\Gamma = \{\gamma\}$ and $\MP(\gamma) = B_1, B_2, B_3$, then, by the straightness invariant, $H$ is not PHCA.  Moreover, $H[\B]$ is not PHCA for $\B$ $=$ \{$B_1$, $B_2$, $B_3$, $U_r(B_1)$, $U_r(B_2)$, $U_r(B_3)$\}~\cite{LinSoulignacSzwarcfiterDAM2013}.  As discussed in Lemma~\ref{lemma:subprocedures}, we can compute all the adjacencies of $H[\B]$ in $O(1)$ time.  
\end{proof}

\begin{theorem}
  Given a round block representation $\Gamma$ of a graph $H$ that satisfies the straightness invariant, it takes $O(1)$ time to determine if $H$ is PIG.  If $H$ is not PIG, then the algorithm outputs $\Gamma|\B$ for a family of blocks $\B$ such that $H[\B]$ is either a $S_3$ or a $C_k$ ($k \geq 4$).  
\end{theorem}

\begin{proof}
  By the straightness invariant, all we need to do to test if $H$ is PIG is to call {\tt straight} (Section~\ref{sec:data structure:round representations}).  If $H$ is not straight, then we test if $H$ is PHCA.  If negative, then we extract an induced $S_3$ or $C_4$ from the output $\Gamma|\B$.  Otherwise, $\Gamma = \{\gamma\}$ and $\MP(\gamma)$ induces a $C_k$ ($k \geq 4$).
\end{proof}

\subsection{The authentication problems}

As discussed in Section~\ref{sec:introduction}, we can conceive three types of checkers, namely static, dynamic, and monitors.  The static checker, which has the simplest implementation, authenticates the witnesses against the static graph $G$.  The dynamic checker, instead, test the witness obtained after one operation is applied on a round block representation $\Phi$.  Although it is more efficient than applying the static checker for each update, the dynamic checker requires some extra effort as different tests are performed for the different updates.  Finally, the monitor is a new layer between the end user and the dynamic algorithm that checks the correct behavior of $\Phi$ and the witnesses it generates.  To do its work in the most efficient way, the monitor requires privileged access to some operations that are restricted to the final user~\cite{BrightSullivan1995}.  Thus, the implementation of the monitor is not as simple as for the checkers, as it requires some knowledge about the internal representation of $\Phi$.  In this section we briefly discuss the static and dynamic checkers, and we skim through a possible design of a monitor.

\subsubsection{The static checker}

The static checker authenticates that a witness $W$ is correct for a graph $G$.  Of course, the correctness depends on the recognition problem we are dealing with and on whether $W$ is positive or negative.  Since we consider three problems, i.e., the recognition of PIG, PHCA, and PCA graphs, the static checker has to solve six problems.

When $G$ is claimed to be PCA, the witness is a round block representation $\Phi$ of $G$.  To authenticate $\Phi$, the checker tests if:
\begin{itemize}
  \item each $\phi \in \Phi$ is a block contig,
  \item $G(\Phi)$ is isomorphic to $G$,
  \item every block of $G$ corresponds to a block of $\Phi$, and
  \item if required, $\Phi$ is (locally) straight.
\end{itemize}
On the other hand, if $G$ is declared as not being a member of the class, then the negative witness is a minimally forbidden $\Bip{\Phi}{\N}$, where $\N = \emptyset$ when the problem is the recognition of PIG or PHCA graphs.  To authenticate that $\Bip{\Phi}{\N}$ is correct, the checker builds the graph $F$ represented by $\Bip{\Phi}{\N}$, and then it tests that $F$ is isomorphic to an induced subgraph of $G$.

It is not hard to see that these problems can be solved in $O(n+m)$ time.  Moreover, the implementation of the checker is simple as desired.

\subsubsection{The dynamic checker}

The static checker is optimal for the authentication of static graphs.  However, its time complexity is excessive when compared to the time required by an update of the round block representation $\Gamma$.  Thus, the static checker is not well suited for the dynamic algorithm.  The dynamic checker, instead, tries to authenticate the witness against $\Gamma$.  Of course, the authentication depends on the applied update and the kind of witness obtained.

Suppose we want to authenticate a successful insertion.  In this case, the input is $N(v)$ and a round block representation $\Gamma$.  The output is the vertex $v$, and the witness is the round block representation $\Psi$ of $G$ that satisfies the straightness invariant.  Let $B$ be the block of $\Psi$ that contains $v$.  To authenticate $\Psi$, we check that:
\begin{enumerate}[(i)]
  \item $\Psi$ is a round block representation that satisfies the straightness invariant,
  \item $N(v) = \bigcup[F_l(B), F_r(B)]$,
  \item $G(\Psi) \setminus \{v\}$ is isomorphic to $G(\Gamma)$, and 
  \item the vertices of $G(\Gamma)$ appear in the same blocks in $\Psi \setminus \{\{v\}\}$ and $\Gamma$.
\end{enumerate}
There is a large asymmetry between the {\tt insert} operation and its authentication.  Whereas the former deals mostly with $N(v)$, the latter tests all the blocks of $\Psi$.  There are three reasons why the checker must look at the complete structure.  First, because the dynamic algorithm could modify $O(n)$ far neighbors even when $d(v) = O(1)$ (e.g., when the universal block $B \in \Gamma$ is separated into $-B$ and $+B$).  Second, and most important, because we cannot assure that a buggy implementation of {\tt insert} updates only the blocks it is supposed to.  Third, because the specification of {\tt insert} requires $G(\Gamma)$ to be isomorphic to $G(\Psi) \setminus \{v\}$, and no more guaranties are provided.

It is not hard to see that (i)~and~(ii) can be implemented in $O(n)$ time.  For (iii)~and~(iv), the checker works as follows.  Let $\Phi = \Psi \setminus \{\{v\}\}$.  First, the checker looks for all the co-end blocks in both $\Phi$ and $\Gamma$.  Note that $\Phi$ needs not be a block representation.  However, it is not hard to consider the twin semiblocks of $\Phi$ as being part of the same block; we omit the details.  Also, observe that $\Phi$ is not actually computed; instead $v$ is ignored in $\Psi$.  In the second step, the algorithm checks that the co-end blocks of $\Phi$ and $\Gamma$ coincide.  If not, the checker reports that the implementation is buggy.  When all the co-end blocks coincide, the checker traverses each co-contig $\phi$ of $\Phi$ to test that the remaining blocks appear in the same order as in $\Gamma$ (or its reversal).  If negative, then the checker outputs that the implementation is buggy; otherwise, both (iii)~and~(iv) hold.  The correctness of this algorithm follows from the fact that co-connected PCA graphs admit exactly two round block representations, one the reverse of the other~\cite{HuangJCTSB1995}.  Note that the dynamic checker for the insertion runs in $O(n)$ time.  Its implementation, however, is not as simple as the one for the static checker.

The authentication required for {\tt remove} is similar and can be implemented in $O(n)$ time as well.  Analogously, the authentication that $\Phi$ is either straight or locally straight, required for {\tt forbiddenPIG} and {\tt forbiddenPHCA}, takes $O(n)$ time.  Finally, to verify a negative witness $\Bip{\Phi}{\N}$, the checker tests that $\Phi$ is indeed a representation induced from $\Gamma$, and that $\N$ contains the neighbors of $v$ in $\Phi$.  Both of these tests can be easily implemented in $O(n)$ time.

\subsubsection{The monitor}

Although the dynamic checker is much faster than the static one, it is still too expensive when compared to the update operations.  Unfortunately, the dynamic checker is optimal when no details about the implementation can be exploited.  When we have access to the implementation of the data structure, we can \emph{monitor} each operation to ensure its correctness~\cite{BrightSullivan1995,McConnellMehlhornNaeherSchweitzerCSR2011}.  Recall that the dynamic algorithm deals with five data structures, namely contigs, semiblock paths, round representations, connectivity structures, and witnesses.  The idea is to implement these data types in a way that we can trust all of them.  

To make the above statement more precise, consider the {\tt separate} operation of contigs.  Recall that {\tt separate($\DS{B}$, $W$)} transforms the contig $\gamma$ referenced by $\DS{B}$ into the contig $\phi$ that represents $G(\phi)$ by splitting $B$ into two indistinguishable semiblocks $B \setminus W$ and $W$ in such a way that $R^\phi(W) = R^\gamma(B)$, $L^\phi(W) = B \setminus W$, and $L^\phi(B \setminus W) = L^\gamma(B)$.  To verify that {\tt separate} is correct, a checker must guarantee, among other things, that $\phi$ represents $G(\gamma)$.  There are at least two inconveniences that the checker must confront.  First, a buggy implementation could fail to update $F_r$ for some neighbor of $B$.  Second, there could be $O(n)$ semiblocks that have $B$ as its right far neighbor in $\Phi$, and all of them should reference $W$ in $\Gamma$.  Thus, if the data structure is unknown, then the checker must traverse $O(n)$ semiblocks to authenticate $\gamma$.  However, the implementation spends $O(1)$ time to simultaneously update all the right far neighbors.  In fact, the algorithm consists of swapping two \emph{self} pointers~\cite{HellShamirSharanSJC2001,SoulignacA2015}.  If we were given access to the self pointers, then we could test that the swap is correct.  A second and more pragmatic approach is to consider that such a swap is correct by definition.  The reason is that proving the correctness of an implementation of swap is as simple, if not simpler, than authenticating the output of swap.  Moreover, if we cannot trust the implementation of swap, then we cannot trust the implementation of the monitor either.  

In a similar way as described for the update of far pointers, we may assume that a contig $\phi$ provides other basic operations, which are accessible only to the monitor, that are correct by definition.  However, some operations are harder to implement and should be monitored.  We differentiate three types of errors that impact on the design of $\phi$ and its monitor.
\begin{description}
  \item [(Improper) access errors] arise when a portion of the data structure that should not be accessed is modified.  For instance, only the semiblocks in $[B_l, B_r]$ need to be updated in {\tt reception($B_l$, $B_r$)}.  We consider those modifications to semiblocks outside $[B_l, B_r]$ as access errors.   There are at least two basic methods for dealing with access errors.  The simplest one is to ignore the error; this strategy is appropriate if we can assure that the error will be caught when the modified portion of the structure is accessed.  The alternative method is to use some kind of supervised memory that tracks all the updates of the data structure.  Then, the monitor can refuse those operations that access a restricted portion of the memory.  The first approach is used in~\cite{McConnellMehlhornNaeherSchweitzerCSR2011} for ordered dictionaries.  When the monitor asks the dictionary to {\tt insert} a pair $(k, i)$, the dictionary could (erroneously) erase an item $(k', i')$.  Such misbehavior is not detected by the monitor until it tries to access $(k',i')$.  Thus, the monitor is not certifying the whole data structure for {\tt insert}, but only that the insertion takes place where it should.  We remark that missing such an error is not critical for dictionaries, because the elements that it holds are independent of each other.  For contigs, perhaps it is better to take actions immediately using the supervised memory solution.
  
  \item [Memory errors] occur when an uncontrolled memory location is accessed.  To deal with uncontrolled memory locations, we can follow the same technique as in~\cite{McConnellMehlhornNaeherSchweitzerCSR2011}.  That is, each semiblock $B \in \Phi$ keeps the position of a semiblock pointer $\DS{B}$ in an array $T$ of ``trusted'' memory.  This array is controlled by the monitor to ensure that each access to $B$ is correct.  To authenticate the access to $B$, the monitor access its position of $T$ and uses $\DS{B}$ to control that $B$ was under the control of the data structure.
  
  \item [Logical errors] happen when an operations does not behave as it is supposed to, but accessing only the portions of the data structure to which they have access.  Suppose, for instance, that the monitor is asked to perform {\tt reception($B_l$, $B_r$)} on $\phi$.  The monitor forwards this operation to the data structure and it obtains the semiblock $B$ that contains $v$ on $\psi$.  When $B$ is not an end semiblock, the monitor outputs that the implementation is buggy if some of the following check fails.
  \begin{itemize}
    \item $F_l(B) = B_l$ and $F_r(B) = B_r$,
    \item $F_r^\psi(W) = B$ for every $W \in [B_l, B)$ such that $F_r^\phi(W) = L^\psi(B)$,
    \item $F_r^\psi(W) = F_r^\phi(W)$ for every $W \in [B_l, B)$ such that $F_r^\phi(W) \neq L^\psi(B)$,
    \item $F_l^\psi(W) = \{v\}$ for every $W \in (B, B_r]$ such that $F_l^\phi(W) = R^\psi(\{v\})$, and
    \item $F_l^\psi(W) = F_l^\phi(W)$ for every $W \in (B, B_r]$ such that $F_l^\phi(W) \neq R^\psi(\{v\})$.
  \end{itemize}
  The case in which $B$ is an end semiblock is handled similarly.
\end{description}

Using the above techniques, we can authenticate all the operations on contigs.  Then, the remaining data types should be monitored as well.  Suppose we need to check that {\tt reception($B_l$, $B_r$)} works as specified for a round representation $\Phi$.  A priori, the only operation of contigs that should be invoked is the trusted {\tt reception} with inputs $B_l$ and $B_r$.  Thus, any other update on the contigs should be regarded as an access error.  Following the supervised memory solution, we may ask the monitor of contigs to track the updates that it performs.  Then, the monitor of $\Phi$ can observe that the only update on its contigs was the {\tt reception} of $B_l$ and $B_r$.  Since this operation is under supervision, we many assume it is correct, thus we only need to check the logical errors.  In this case, that the obtained contig is not circular when $|\Psi| > 1$ for the obtained round representation.

\section{Conclusions}
\label{sec:conclusions}

We designed a new dynamic algorithm for the recognition of PCA, PHCA, and PIG graphs that allows vertex updates.  The algorithm keeps a round block representation $\Phi$ of the input graph $G$ that can be regarded as being a positive witness.  When the insertion of $v$ into $G = H \setminus \{v\}$ fails, the algorithms exhibits a minimally forbidden induced subgraph $F$ of $H$.  To work as fast as possible, the algorithm keeps a partial view of $F \setminus \{v\}$ that contains all but $O(d(v))$ vertices of $F$.  The problem of finding a negative certificate when edges updates are allowed is left open.

The certifying algorithm is optimal when applied for the recognition of static graphs, as it runs in $O(d(v))$ time per inserted vertex.  The algorithm is almost optimal when both insertions and removal are allowed, as it requires $O(d(v) + \log n)$ time per operation and the lower bound in the cell probe model of computation with word-size $b$ is $\Omega(d(v) + \log n/(\log\log n + \log b))$~\cite{HellShamirSharanSJC2001}.  

Regarding the authentication problem, we considered three possibilities, each one giving rise to a different kind of checker.  Static checkers test the result of the algorithm for a static graph $G$.  Its input is $G$ together with either a round representation $\Phi$ or a graph $F$, and the goal is to verify that $\Phi$ is a round block representation of $G$ or that $F$ is a minimally forbidden induced subgraph of $G$.  Dynamic checkers, instead, test that an operation on a round block representation $\Phi$ is successful.  Its input, then, is $\Phi$ plus the input of the operation and either a round representation $\Psi$ or a minimally forbidden $F$.  The goal in this case is to verify that $\Psi$ is a round block representation of the graph $H$ that should be obtained from $G(\Phi)$ when the operation is applied or to test that $F$ represents a minimally forbidden induced subgraph of $H$.  By definition, the problems associated to the static and dynamic checkers are static and require $\Omega(n+m)$ and $\Omega(n)$ time in the worst case as either $G$ or $\Phi$ have to be traversed, respectively.  Monitors, instead, are dynamic algorithms (i.e., data structures) that sit between the user and the round block representation $\Phi$ of the dynamic graph $G$.  When a monitor has access to some privileged (query) operations on $\Phi$, the time required for the authentication can be reduced.  In this article we skim through the process of designing a monitor for the algorithm which, we believe, could be used to authenticate each operation in $O(t)$ time, where $t$ is the time required by the operation itself.  There is no proof of this fact, as the monitor is incomplete; yet, we discussed some issues that can arise when such a monitor is developed.

\small

\begin{thebibliography}{32}
\providecommand{\natexlab}[1]{#1}
\providecommand{\url}[1]{\texttt{#1}}
\expandafter\ifx\csname urlstyle\endcsname\relax
  \providecommand{\doi}[1]{doi: #1}\else
  \providecommand{\doi}{doi: \begingroup \urlstyle{rm}\Url}\fi

\bibitem[Bang-Jensen and Gutin(2009)]{Bang-JensenGutin2009}
J.~Bang-Jensen and G.~Gutin.
\newblock \emph{Digraphs}.
\newblock Springer Monographs in Mathematics. Springer-Verlag London, Ltd.,
  London, second edition, 2009.
\newblock \doi{10.1007/978-1-84800-998-1}.

\bibitem[Bright and Sullivan(1995)]{BrightSullivan1995}
J.~D. Bright and G.~F. Sullivan.
\newblock On-line error monitoring for several data structures.
\newblock In \emph{Digest of Papers: FTCS-25, The Twenty-Fifth International
  Symposium on Fault-Tolerant Computing, Pasadena, California, USA, June 27-30,
  1995}, pp. 392--401. {IEEE} Computer Society, 1995.
\newblock \doi{10.1109/FTCS.1995.466960}.

\bibitem[Corneil(2004)]{CorneilDAM2004}
D.~G. Corneil.
\newblock A simple 3\mbox{-}{s}weep {LBFS} algorithm for the recognition of
  unit interval graphs.
\newblock \emph{Discrete Appl. Math.}, 138\penalty0 (3):\penalty0 371--379,
  2004.
\newblock \doi{10.1016/j.dam.2003.07.001}.

\bibitem[Corneil et~al.(1995)Corneil, Kim, Natarajan, Olariu, and
  Sprague]{CorneilKimNatarajanOlariuSpragueIPL1995}
D.~G. Corneil, H.~Kim, S.~Natarajan, S.~Olariu, and A.~P. Sprague.
\newblock Simple linear time recognition of unit interval graphs.
\newblock \emph{Inform. Process. Lett.}, 55\penalty0 (2):\penalty0 99--104,
  1995.
\newblock \doi{10.1016/0020-0190(95)00046-F}.

\bibitem[Crespelle(2010)]{Crespelle2010}
C.~Crespelle.
\newblock Fully dynamic representations of interval graphs.
\newblock In \emph{Graph-theoretic concepts in computer science}, vol. 5911 of
  \emph{Lecture Notes in Comput. Sci.}, pp. 77--87. Springer, Berlin, 2010.
\newblock \doi{10.1007/978-3-642-11409-0_7}.

\bibitem[Crespelle and Paul(2006)]{CrespellePaulDAM2006}
C.~Crespelle and C.~Paul.
\newblock Fully dynamic recognition algorithm and certificate for directed
  cographs.
\newblock \emph{Discrete Appl. Math.}, 154\penalty0 (12):\penalty0 1722--1741,
  2006.
\newblock \doi{10.1016/j.dam.2006.03.005}.

\bibitem[Crespelle and Paul(2010)]{CrespellePaulA2010}
C.~Crespelle and C.~Paul.
\newblock Fully dynamic algorithm for recognition and modular decomposition of
  permutation graphs.
\newblock \emph{Algorithmica}, 58\penalty0 (2):\penalty0 405--432, 2010.
\newblock \doi{10.1007/s00453-008-9273-0}.

\bibitem[Deng et~al.(1996)Deng, Hell, and Huang]{DengHellHuangSJC1996}
X.~Deng, P.~Hell, and J.~Huang.
\newblock Linear-time representation algorithms for proper circular-arc graphs
  and proper interval graphs.
\newblock \emph{SIAM J. Comput.}, 25\penalty0 (2):\penalty0 390--403, 1996.
\newblock \doi{10.1137/S0097539792269095}.

\bibitem[Gioan and Paul(2012)]{GioanPaulDAM2012}
E.~Gioan and C.~Paul.
\newblock Split decomposition and graph-labelled trees: characterizations and
  fully dynamic algorithms for totally decomposable graphs.
\newblock \emph{Discrete Appl. Math.}, 160\penalty0 (6):\penalty0 708--733,
  2012.
\newblock \doi{10.1016/j.dam.2011.05.007}.

\bibitem[Golumbic(2004)]{Golumbic2004}
M.~C. Golumbic.
\newblock \emph{Algorithmic graph theory and perfect graphs}, vol.~57 of
  \emph{Annals of Discrete Mathematics}.
\newblock Elsevier Science B.V., Amsterdam, second edition, 2004.

\bibitem[Heggernes and Mancini(2009)]{HeggernesManciniDAM2009}
P.~Heggernes and F.~Mancini.
\newblock Dynamically maintaining split graphs.
\newblock \emph{Discrete Appl. Math.}, 157\penalty0 (9):\penalty0 2057--2069,
  2009.
\newblock \doi{10.1016/j.dam.2008.06.028}.

\bibitem[Hell and Huang(2004/05)]{HellHuangSJDM2004/05}
P.~Hell and J.~Huang.
\newblock Certifying {L}ex{BFS} recognition algorithms for proper interval
  graphs and proper interval bigraphs.
\newblock \emph{SIAM J. Discrete Math.}, 18\penalty0 (3):\penalty0 554--570
  (electronic), 2004/05.
\newblock \doi{10.1137/S0895480103430259}.

\bibitem[Hell et~al.(2001)Hell, Shamir, and Sharan]{HellShamirSharanSJC2001}
P.~Hell, R.~Shamir, and R.~Sharan.
\newblock A fully dynamic algorithm for recognizing and representing proper
  interval graphs.
\newblock \emph{SIAM J. Comput.}, 31\penalty0 (1):\penalty0 289--305
  (electronic), 2001.
\newblock \doi{10.1137/S0097539700372216}.

\bibitem[Herrera~de Figueiredo et~al.(1995)Herrera~de Figueiredo, Meidanis, and
  Picinin~de Mello]{HerreraMeidanisPicininIPL1995}
C.~M. Herrera~de Figueiredo, J.~Meidanis, and C.~Picinin~de Mello.
\newblock A linear-time algorithm for proper interval graph recognition.
\newblock \emph{Inform. Process. Lett.}, 56\penalty0 (3):\penalty0 179--184,
  1995.
\newblock \doi{10.1016/0020-0190(95)00133-W}.

\bibitem[Hsu and Tsai(1991)]{HsuTsaiIPL1991}
W.~L. Hsu and K.-H. Tsai.
\newblock Linear time algorithms on circular-arc graphs.
\newblock \emph{Inform. Process. Lett.}, 40\penalty0 (3):\penalty0 123--129,
  1991.

\bibitem[Huang(1995)]{HuangJCTSB1995}
J.~Huang.
\newblock On the structure of local tournaments.
\newblock \emph{J. Combin. Theory Ser. B}, 63\penalty0 (2):\penalty0 200--221,
  1995.
\newblock \doi{10.1006/jctb.1995.1016}.

\bibitem[Ibarra(2008)]{IbarraATA2008}
L.~Ibarra.
\newblock Fully dynamic algorithms for chordal graphs and split graphs.
\newblock \emph{ACM Trans. Algorithms}, 4\penalty0 (4):\penalty0 Art. 40, 20,
  2008.
\newblock \doi{10.1145/1383369.1383371}.

\bibitem[Ibarra(2009)]{Ibarra2009}
L.~Ibarra.
\newblock A fully dynamic graph algorithm for recognizing proper interval
  graphs.
\newblock In \emph{W{ALCOM}---{A}lgorithms and computation}, vol. 5431 of
  \emph{Lecture Notes in Comput. Sci.}, pp. 190--201, Berlin, 2009. Springer.
\newblock \doi{10.1007/978-3-642-00202-1_17}.

\bibitem[Ibarra(2010)]{IbarraA2010}
L.~Ibarra.
\newblock A fully dynamic graph algorithm for recognizing interval graphs.
\newblock \emph{Algorithmica}, 58\penalty0 (3):\penalty0 637--678, 2010.
\newblock \doi{10.1007/s00453-009-9291-6}.

\bibitem[Kaplan and Nussbaum(2009)]{KaplanNussbaumDAM2009}
H.~Kaplan and Y.~Nussbaum.
\newblock Certifying algorithms for recognizing proper circular-arc graphs and
  unit circular-arc graphs.
\newblock \emph{Discrete Appl. Math.}, 157\penalty0 (15):\penalty0 3216--3230,
  2009.
\newblock \doi{10.1016/j.dam.2009.07.002}.

\bibitem[Kratsch et~al.(2006)Kratsch, McConnell, Mehlhorn, and
  Spinrad]{KratschMcConnellMehlhornSpinradSJC2006}
D.~Kratsch, R.~M. McConnell, K.~Mehlhorn, and J.~P. Spinrad.
\newblock Certifying algorithms for recognizing interval graphs and permutation
  graphs.
\newblock \emph{SIAM J. Comput.}, 36\penalty0 (2):\penalty0 326--353
  (electronic), 2006.

\bibitem[Lekkerkerker and Boland(1962/1963)]{LekkerkerkerBolandFM1962/1963}
C.~G. Lekkerkerker and J.~C. Boland.
\newblock Representation of a finite graph by a set of intervals on the real
  line.
\newblock \emph{Fund. Math.}, 51:\penalty0 45--64, 1962/1963.

\bibitem[Lin and Szwarcfiter(2009)]{LinSzwarcfiterDM2009}
M.~C. Lin and J.~L. Szwarcfiter.
\newblock Characterizations and recognition of circular-arc graphs and
  subclasses: a survey.
\newblock \emph{Discrete Math.}, 309\penalty0 (18):\penalty0 5618--5635, 2009.
\newblock \doi{10.1016/j.disc.2008.04.003}.

\bibitem[Lin et~al.(2013)Lin, Soulignac, and
  Szwarcfiter]{LinSoulignacSzwarcfiterDAM2013}
M.~C. Lin, F.~J. Soulignac, and J.~L. Szwarcfiter.
\newblock Normal {H}elly circular-arc graphs and its subclasses.
\newblock \emph{Discrete Appl. Math.}, 161\penalty0 (7-8):\penalty0 1037--1059,
  2013.
\newblock \doi{10.1016/j.dam.2012.11.005}.

\bibitem[McConnell et~al.(2011)McConnell, Mehlhorn, N{\"{a}}her, and
  Schweitzer]{McConnellMehlhornNaeherSchweitzerCSR2011}
R.~M. McConnell, K.~Mehlhorn, S.~N{\"{a}}her, and P.~Schweitzer.
\newblock Certifying algorithms.
\newblock \emph{Comput. Sci. Rev.}, 5\penalty0 (2):\penalty0 119--161, 2011.
\newblock \doi{10.1016/j.cosrev.2010.09.009}.

\bibitem[Meister(2005)]{MeisterDAM2005}
D.~Meister.
\newblock Recognition and computation of minimal triangulations for {AT}-free
  claw-free and co-comparability graphs.
\newblock \emph{Discrete Appl. Math.}, 146\penalty0 (3):\penalty0 193--218,
  2005.
\newblock \doi{10.1016/j.dam.2004.10.001}.

\bibitem[Nikolopoulos et~al.(2012)Nikolopoulos, Palios, and
  Papadopoulos]{NikolopoulosPaliosPapadopoulosTCS2012}
S.~D. Nikolopoulos, L.~Palios, and C.~Papadopoulos.
\newblock A fully dynamic algorithm for the recognition of {$P_4$}-sparse
  graphs.
\newblock \emph{Theoret. Comput. Sci.}, 439:\penalty0 41--57, 2012.
\newblock \doi{10.1016/j.tcs.2012.03.020}.

\bibitem[Pirlot and Vincke(1997)]{PirlotVincke1997}
M.~Pirlot and P.~Vincke.
\newblock \emph{Semiorders}, vol.~36 of \emph{Theory and Decision Library.
  Series B: Mathematical and Statistical Methods}.
\newblock Kluwer Academic Publishers Group, Dordrecht, 1997.
\newblock \doi{10.1007/978-94-015-8883-6}.

\bibitem[Shamir and Sharan(2004)]{ShamirSharanDAM2004}
R.~Shamir and R.~Sharan.
\newblock A fully dynamic algorithm for modular decomposition and recognition
  of cographs.
\newblock \emph{Discrete Appl. Math.}, 136\penalty0 (2-3):\penalty0 329--340,
  2004.
\newblock \doi{10.1016/S0166-218X(03)00448-7}.

\bibitem[Soulignac(2015)]{SoulignacA2015}
F.~J. Soulignac.
\newblock Fully dynamic recognition of proper circular-arc graphs.
\newblock \emph{Algorithmica}, 71\penalty0 (4):\penalty0 904--968, 2015.
\newblock \doi{10.1007/s00453-013-9835-7}.

\bibitem[Tedder and Corneil(2007)]{TedderCorneil2007}
M.~Tedder and D.~Corneil.
\newblock An optimal, edges-only fully dynamic algorithm for
  distance-hereditary graphs.
\newblock In \emph{S{TACS} 2007}, vol. 4393 of \emph{Lecture Notes in Comput.
  Sci.}, pp. 344--355. Springer, Berlin, 2007.
\newblock \doi{10.1007/978-3-540-70918-3_30}.

\bibitem[Tucker(1974)]{TuckerDM1974}
A.~Tucker.
\newblock Structure theorems for some circular-arc graphs.
\newblock \emph{Discrete Math.}, 7:\penalty0 167--195, 1974.

\end{thebibliography}

\appendix
\newgeometry{left=1.5cm,right=1.5cm,top=2cm,bottom=2.5cm}
\small
\twocolumn
\section{Adequacy proofs}
\label{app:adequacy}

In this appendix we include the proofs that a family of semiblocks $\B$ is an adequate forbidden.  These proofs were generated by a computer program, which explains why all the sections have the same structure.  Each section references the lemma in which it is required.  Then, a summary of the \emph{current knowledge} of $H[\B \cup \{v\}]$ is depicted. By current knowledge we mean that the edges that we actually know that belong or could belong to $H[\B]$ are depicted.  Together with this graph, we describe four fields:
\begin{itemize}
  \item $\B$ contains all the semiblocks in $\B$ in the order in which they appear in the round representation.
  \item $F_r$ shows the values of $F_r$ but only for those semiblocks whose values of $F_r$ cannot be deduced otherwise.  Also, this is depicted taking into account only those edges that belong to $H[\B]$ using our current knowledge.
  \item $N(v)$ shows the semiblocks to which $v$ is adjacent, according to our current knowledge.
  \item Rest includes all the adjacencies that could be added into $H[\B]$.
\end{itemize}
With this information, we can enumerate all the possible subgraphs that $\B$ could induce in $H$.  Of course, there is only one such possibility for $H[\B]$ when Rest is empty; in such a case, this summary is not depicted (see e.g., Section~\ref{app:subprocedures:left to right}).

After the summary, there is one subsection dealing with each possibility for $H[\B]$, except for those that are duplicated.  Each subsection includes the case stating ``If $\langle$case$\rangle$, then $\langle$forbidden subgraph$\rangle$'', where the forbidden subgraph is highlighted in blue.  There are two kinds of duplicated possibilities: those in which $H[\B]$ is isomorphic to a case already examined, and those that are included in some other case.  The former are ignored, while the latter are described in a section entitled ``Implied cases''. 

\raggedright
\input{lemmaLeftToRight}
\input{lemmaZeroTriangle}
\input{lemmaCenterNotDominating}
\input{lemmaCenterAndRightDominatedCase1}
\subsection{Lemma~\ref{lemma:subprocedures}~(\ref{lemma:subprocedures:center and right dominated})}\label{app:subprocedures:center and right dominated:case 2}
\begin{description}
\item[]$\mathcal{B}$: $U_l(T_1)$, $F_l(T_1)$, $F_l(T_2)$, $T_1$, $T_2$, $T_3$, $F_r(T_1)$
\item[]$F_r$: \mbox{$U_l(T_1)$ $\to$ $F_l(T_1)$}, \mbox{$F_l(T_2)$ $\to$ $T_2$}, \mbox{$F_l(T_1)$ $\to$ $T_1$}, \mbox{$T_1$ $\to$ $F_r(T_1)$}
\item[]$N(v)$: $T_2$
\item[]Rest: \mbox{$F_r(T_1)$ $\to$ $U_l(T_1)$}, \mbox{$U_l(T_1)$ $\to$ $F_l(T_2)$}, \mbox{($v$, $U_l(T_1)$)}, \mbox{($v$, $F_r(T_1)$)}, \mbox{($v$, $F_l(T_1)$)}, \mbox{($v$, $F_l(T_2)$)}
\end{description}\opt{TIKZ}{
\tikzsetnextfilename{app-subprocedures-center-and-right-dominated-case-2}
	  \begin{center}
	  \begin{tikzpicture}[vertex/.style={shape=circle, draw=black, fill=black, line width=0pt, minimum size=1.5mm, inner sep=0mm, outer sep=0mm}, %
			      highlight vertex/.style={vertex,draw=blue, fill=blue},
			      essential edge/.style={thick,->,>=stealth}, 
			      implied edge/.style={thin,->,>=stealth, draw=gray},
			      highlight edge/.style={essential edge, ultra thick, draw=blue},
			      highlight implied edge/.style={implied edge, draw=blue},
			      every label/.style={font=\footnotesize, inner sep=0pt, outer sep=0pt},
			      label distance=2pt,
	  ]
	\node [vertex] (T_3) at (1.732,-0.499) [label=below right:$T_3$] {};
\node [vertex] (T_2) at (0.999,-0.866) [label=below right:$T_2$] {};
\node [vertex] (U_lT_1) at (-2.0,0.0) [label= left:$U_l(T_1)$] {};
\node [vertex] (F_lT_1) at (-1.732,-0.499) [label=below left:$F_l(T_1)$] {};
\node [vertex] (v) at (0.0,1.0) [label=above:$v$] {};
\node [vertex] (F_lT_2) at (-1.0,-0.866) [label=below left:$F_l(T_2)$] {};
\node [vertex] (F_rT_1) at (2.0,0.0) [label= right:$F_r(T_1)$] {};
\node [vertex] (T_1) at (0.0,-1.0) [label=below:$T_1$] {};
\draw [implied edge,bend right=20] (T_3) to (F_rT_1);
\draw [implied edge,bend right=20] (T_2) to (F_rT_1);
\draw [implied edge,bend right=20] (T_2) to (T_3);
\draw [essential edge,bend right=20] (U_lT_1) to (F_lT_1);
\draw [essential edge,bend right=20] (F_lT_1) to (T_1);
\draw [implied edge,bend right=20] (F_lT_1) to (F_lT_2);
\draw [essential edge] (v) to (T_2);
\draw [essential edge,bend right=20] (F_lT_2) to (T_2);
\draw [implied edge,bend right=20] (F_lT_2) to (T_1);
\draw [essential edge,bend right=20] (T_1) to (F_rT_1);
\draw [implied edge,bend right=20] (T_1) to (T_3);
\draw [implied edge,bend right=20] (T_1) to (T_2);
\draw [essential edge,dashed] (v) to (F_lT_2);
\draw [essential edge,dashed] (v) to (F_rT_1);
\draw [essential edge,bend right=20,dashed] (F_rT_1) to (U_lT_1);
\draw [essential edge,dashed] (v) to (U_lT_1);
\draw [essential edge,dashed] (v) to (F_lT_1);
\draw [essential edge,bend right=20,dashed] (U_lT_1) to (F_lT_2);
\end{tikzpicture}
\end{center}}
\opt{NOTIKZ}{\begin{center}\includegraphics{\imgs 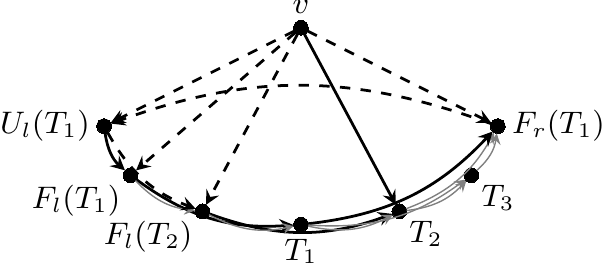}\end{center}}
\subsubsection{}\label{app:subprocedures:center and right dominated:case 2:1}If $\emptyset$, then $K_{1,3}$\opt{TIKZ}{
\tikzsetnextfilename{app-subprocedures-center-and-right-dominated-case-2-1}
	  \begin{center}
	  \begin{tikzpicture}[vertex/.style={shape=circle, draw=black, fill=black, line width=0pt, minimum size=1.5mm, inner sep=0mm, outer sep=0mm}, %
			      highlight vertex/.style={vertex,draw=blue, fill=blue},
			      essential edge/.style={thick,->,>=stealth}, 
			      implied edge/.style={thin,->,>=stealth, draw=gray},
			      highlight edge/.style={essential edge, ultra thick, draw=blue},
			      highlight implied edge/.style={implied edge, draw=blue},
			      every label/.style={font=\footnotesize, inner sep=0pt, outer sep=0pt},
			      label distance=2pt,
	  ]
	\node [highlight vertex] (T_3) at (1.732,-0.499) [label=below right:$T_3$] {};
\node [highlight vertex] (T_2) at (0.999,-0.866) [label=below right:$T_2$] {};
\node [vertex] (U_lT_1) at (-2.0,0.0) [label= left:$U_l(T_1)$] {};
\node [vertex] (F_lT_1) at (-1.732,-0.499) [label=below left:$F_l(T_1)$] {};
\node [highlight vertex] (v) at (0.0,1.0) [label=above:$v$] {};
\node [highlight vertex] (F_lT_2) at (-1.0,-0.866) [label=below left:$F_l(T_2)$] {};
\node [vertex] (F_rT_1) at (2.0,0.0) [label= right:$F_r(T_1)$] {};
\node [vertex] (T_1) at (0.0,-1.0) [label=below:$T_1$] {};
\draw [implied edge,bend right=20] (T_3) to (F_rT_1);
\draw [implied edge,bend right=20] (T_2) to (F_rT_1);
\draw [highlight implied edge,bend right=20] (T_2) to (T_3);
\draw [essential edge,bend right=20] (U_lT_1) to (F_lT_1);
\draw [essential edge,bend right=20] (F_lT_1) to (T_1);
\draw [implied edge,bend right=20] (F_lT_1) to (F_lT_2);
\draw [highlight edge] (v) to (T_2);
\draw [highlight edge,bend right=20] (F_lT_2) to (T_2);
\draw [implied edge,bend right=20] (F_lT_2) to (T_1);
\draw [essential edge,bend right=20] (T_1) to (F_rT_1);
\draw [implied edge,bend right=20] (T_1) to (T_3);
\draw [implied edge,bend right=20] (T_1) to (T_2);
\end{tikzpicture}
\end{center}}
\opt{NOTIKZ}{\begin{center}\includegraphics{\imgs 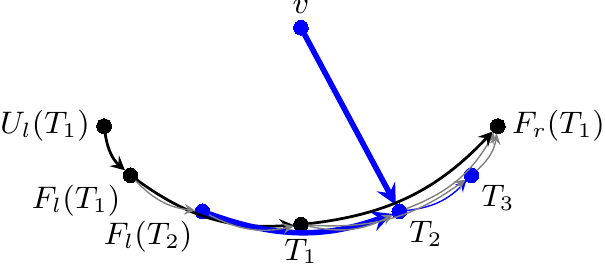}\end{center}}
\subsubsection{}\label{app:subprocedures:center and right dominated:case 2:2}If \mbox{($v$, $F_l(T_2)$)}, then $H_1$\opt{TIKZ}{
\tikzsetnextfilename{app-subprocedures-center-and-right-dominated-case-2-2}
	  \begin{center}
	  \begin{tikzpicture}[vertex/.style={shape=circle, draw=black, fill=black, line width=0pt, minimum size=1.5mm, inner sep=0mm, outer sep=0mm}, %
			      highlight vertex/.style={vertex,draw=blue, fill=blue},
			      essential edge/.style={thick,->,>=stealth}, 
			      implied edge/.style={thin,->,>=stealth, draw=gray},
			      highlight edge/.style={essential edge, ultra thick, draw=blue},
			      highlight implied edge/.style={implied edge, draw=blue},
			      every label/.style={font=\footnotesize, inner sep=0pt, outer sep=0pt},
			      label distance=2pt,
	  ]
	\node [highlight vertex] (T_3) at (1.732,-0.499) [label=below right:$T_3$] {};
\node [vertex] (T_2) at (0.999,-0.866) [label=below right:$T_2$] {};
\node [highlight vertex] (U_lT_1) at (-2.0,0.0) [label= left:$U_l(T_1)$] {};
\node [highlight vertex] (F_lT_1) at (-1.732,-0.499) [label=below left:$F_l(T_1)$] {};
\node [highlight vertex] (v) at (0.0,1.0) [label=above:$v$] {};
\node [highlight vertex] (F_lT_2) at (-1.0,-0.866) [label=below left:$F_l(T_2)$] {};
\node [vertex] (F_rT_1) at (2.0,0.0) [label= right:$F_r(T_1)$] {};
\node [highlight vertex] (T_1) at (0.0,-1.0) [label=below:$T_1$] {};
\draw [implied edge,bend right=20] (T_3) to (F_rT_1);
\draw [implied edge,bend right=20] (T_2) to (F_rT_1);
\draw [implied edge,bend right=20] (T_2) to (T_3);
\draw [highlight edge,bend right=20] (U_lT_1) to (F_lT_1);
\draw [highlight edge,bend right=20] (F_lT_1) to (T_1);
\draw [highlight implied edge,bend right=20] (F_lT_1) to (F_lT_2);
\draw [essential edge] (v) to (T_2);
\draw [highlight edge] (v) to (F_lT_2);
\draw [essential edge,bend right=20] (F_lT_2) to (T_2);
\draw [highlight implied edge,bend right=20] (F_lT_2) to (T_1);
\draw [essential edge,bend right=20] (T_1) to (F_rT_1);
\draw [highlight implied edge,bend right=20] (T_1) to (T_3);
\draw [implied edge,bend right=20] (T_1) to (T_2);
\end{tikzpicture}
\end{center}}
\opt{NOTIKZ}{\begin{center}\includegraphics{\imgs 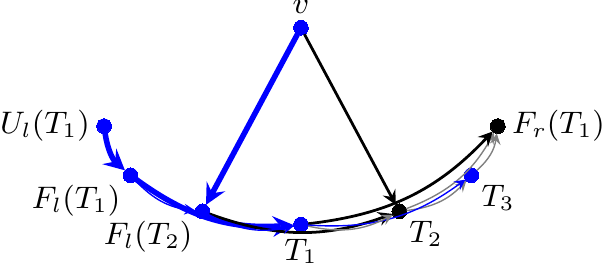}\end{center}}
\subsubsection{}\label{app:subprocedures:center and right dominated:case 2:3}If \mbox{($v$, $F_l(T_2)$)}, \mbox{($v$, $U_l(T_1)$)}, then $C_4^*$\opt{TIKZ}{
\tikzsetnextfilename{app-subprocedures-center-and-right-dominated-case-2-3}
	  \begin{center}
	  \begin{tikzpicture}[vertex/.style={shape=circle, draw=black, fill=black, line width=0pt, minimum size=1.5mm, inner sep=0mm, outer sep=0mm}, %
			      highlight vertex/.style={vertex,draw=blue, fill=blue},
			      essential edge/.style={thick,->,>=stealth}, 
			      implied edge/.style={thin,->,>=stealth, draw=gray},
			      highlight edge/.style={essential edge, ultra thick, draw=blue},
			      highlight implied edge/.style={implied edge, draw=blue},
			      every label/.style={font=\footnotesize, inner sep=0pt, outer sep=0pt},
			      label distance=2pt,
	  ]
	\node [highlight vertex] (T_3) at (1.732,-0.499) [label=below right:$T_3$] {};
\node [vertex] (T_2) at (0.999,-0.866) [label=below right:$T_2$] {};
\node [highlight vertex] (U_lT_1) at (-2.0,0.0) [label= left:$U_l(T_1)$] {};
\node [highlight vertex] (F_lT_1) at (-1.732,-0.499) [label=below left:$F_l(T_1)$] {};
\node [highlight vertex] (v) at (0.0,1.0) [label=above:$v$] {};
\node [highlight vertex] (F_lT_2) at (-1.0,-0.866) [label=below left:$F_l(T_2)$] {};
\node [vertex] (F_rT_1) at (2.0,0.0) [label= right:$F_r(T_1)$] {};
\node [vertex] (T_1) at (0.0,-1.0) [label=below:$T_1$] {};
\draw [implied edge,bend right=20] (T_3) to (F_rT_1);
\draw [implied edge,bend right=20] (T_2) to (F_rT_1);
\draw [implied edge,bend right=20] (T_2) to (T_3);
\draw [highlight edge,bend right=20] (U_lT_1) to (F_lT_1);
\draw [essential edge,bend right=20] (F_lT_1) to (T_1);
\draw [highlight implied edge,bend right=20] (F_lT_1) to (F_lT_2);
\draw [essential edge] (v) to (T_2);
\draw [highlight edge] (v) to (F_lT_2);
\draw [highlight edge] (v) to (U_lT_1);
\draw [essential edge,bend right=20] (F_lT_2) to (T_2);
\draw [implied edge,bend right=20] (F_lT_2) to (T_1);
\draw [essential edge,bend right=20] (T_1) to (F_rT_1);
\draw [implied edge,bend right=20] (T_1) to (T_3);
\draw [implied edge,bend right=20] (T_1) to (T_2);
\end{tikzpicture}
\end{center}}
\opt{NOTIKZ}{\begin{center}\includegraphics{\imgs 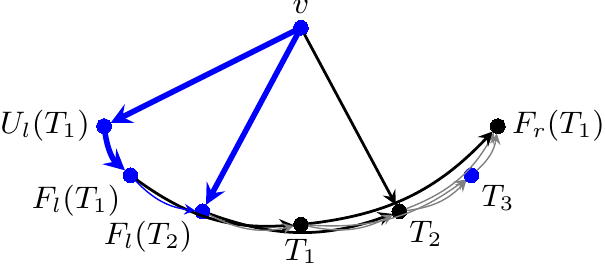}\end{center}}
\subsubsection{}\label{app:subprocedures:center and right dominated:case 2:4}If \mbox{($v$, $F_l(T_2)$)}, \mbox{($v$, $F_l(T_1)$)}, then $K_{1,3}$\opt{TIKZ}{
\tikzsetnextfilename{app-subprocedures-center-and-right-dominated-case-2-4}
	  \begin{center}
	  \begin{tikzpicture}[vertex/.style={shape=circle, draw=black, fill=black, line width=0pt, minimum size=1.5mm, inner sep=0mm, outer sep=0mm}, %
			      highlight vertex/.style={vertex,draw=blue, fill=blue},
			      essential edge/.style={thick,->,>=stealth}, 
			      implied edge/.style={thin,->,>=stealth, draw=gray},
			      highlight edge/.style={essential edge, ultra thick, draw=blue},
			      highlight implied edge/.style={implied edge, draw=blue},
			      every label/.style={font=\footnotesize, inner sep=0pt, outer sep=0pt},
			      label distance=2pt,
	  ]
	\node [vertex] (T_3) at (1.732,-0.499) [label=below right:$T_3$] {};
\node [vertex] (T_2) at (0.999,-0.866) [label=below right:$T_2$] {};
\node [highlight vertex] (U_lT_1) at (-2.0,0.0) [label= left:$U_l(T_1)$] {};
\node [highlight vertex] (F_lT_1) at (-1.732,-0.499) [label=below left:$F_l(T_1)$] {};
\node [highlight vertex] (v) at (0.0,1.0) [label=above:$v$] {};
\node [vertex] (F_lT_2) at (-1.0,-0.866) [label=below left:$F_l(T_2)$] {};
\node [vertex] (F_rT_1) at (2.0,0.0) [label= right:$F_r(T_1)$] {};
\node [highlight vertex] (T_1) at (0.0,-1.0) [label=below:$T_1$] {};
\draw [implied edge,bend right=20] (T_3) to (F_rT_1);
\draw [implied edge,bend right=20] (T_2) to (F_rT_1);
\draw [implied edge,bend right=20] (T_2) to (T_3);
\draw [highlight edge,bend right=20] (U_lT_1) to (F_lT_1);
\draw [highlight edge,bend right=20] (F_lT_1) to (T_1);
\draw [implied edge,bend right=20] (F_lT_1) to (F_lT_2);
\draw [essential edge] (v) to (T_2);
\draw [essential edge] (v) to (F_lT_2);
\draw [highlight edge] (v) to (F_lT_1);
\draw [essential edge,bend right=20] (F_lT_2) to (T_2);
\draw [implied edge,bend right=20] (F_lT_2) to (T_1);
\draw [essential edge,bend right=20] (T_1) to (F_rT_1);
\draw [implied edge,bend right=20] (T_1) to (T_3);
\draw [implied edge,bend right=20] (T_1) to (T_2);
\end{tikzpicture}
\end{center}}
\opt{NOTIKZ}{\begin{center}\includegraphics{\imgs 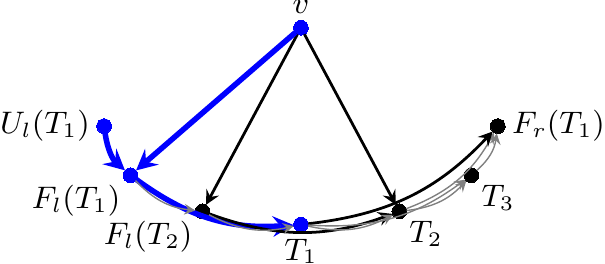}\end{center}}
\subsubsection{}\label{app:subprocedures:center and right dominated:case 2:5}If \mbox{$U_l(T_1)$ $\to$ $F_l(T_2)$}, \mbox{($v$, $F_l(T_2)$)}, then $K_{1,3}$\opt{TIKZ}{
\tikzsetnextfilename{app-subprocedures-center-and-right-dominated-case-2-5}
	  \begin{center}
	  \begin{tikzpicture}[vertex/.style={shape=circle, draw=black, fill=black, line width=0pt, minimum size=1.5mm, inner sep=0mm, outer sep=0mm}, %
			      highlight vertex/.style={vertex,draw=blue, fill=blue},
			      essential edge/.style={thick,->,>=stealth}, 
			      implied edge/.style={thin,->,>=stealth, draw=gray},
			      highlight edge/.style={essential edge, ultra thick, draw=blue},
			      highlight implied edge/.style={implied edge, draw=blue},
			      every label/.style={font=\footnotesize, inner sep=0pt, outer sep=0pt},
			      label distance=2pt,
	  ]
	\node [vertex] (T_3) at (1.732,-0.499) [label=below right:$T_3$] {};
\node [vertex] (T_2) at (0.999,-0.866) [label=below right:$T_2$] {};
\node [highlight vertex] (U_lT_1) at (-2.0,0.0) [label= left:$U_l(T_1)$] {};
\node [vertex] (F_lT_1) at (-1.732,-0.499) [label=below left:$F_l(T_1)$] {};
\node [highlight vertex] (v) at (0.0,1.0) [label=above:$v$] {};
\node [highlight vertex] (F_lT_2) at (-1.0,-0.866) [label=below left:$F_l(T_2)$] {};
\node [vertex] (F_rT_1) at (2.0,0.0) [label= right:$F_r(T_1)$] {};
\node [highlight vertex] (T_1) at (0.0,-1.0) [label=below:$T_1$] {};
\draw [implied edge,bend right=20] (T_3) to (F_rT_1);
\draw [implied edge,bend right=20] (T_2) to (F_rT_1);
\draw [implied edge,bend right=20] (T_2) to (T_3);
\draw [highlight edge,bend right=20] (U_lT_1) to (F_lT_2);
\draw [implied edge,bend right=20] (U_lT_1) to (F_lT_1);
\draw [essential edge,bend right=20] (F_lT_1) to (T_1);
\draw [implied edge,bend right=20] (F_lT_1) to (F_lT_2);
\draw [essential edge] (v) to (T_2);
\draw [highlight edge] (v) to (F_lT_2);
\draw [essential edge,bend right=20] (F_lT_2) to (T_2);
\draw [highlight implied edge,bend right=20] (F_lT_2) to (T_1);
\draw [essential edge,bend right=20] (T_1) to (F_rT_1);
\draw [implied edge,bend right=20] (T_1) to (T_3);
\draw [implied edge,bend right=20] (T_1) to (T_2);
\end{tikzpicture}
\end{center}}
\opt{NOTIKZ}{\begin{center}\includegraphics{\imgs 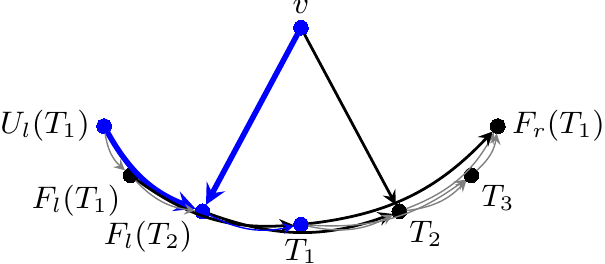}\end{center}}
\subsubsection{}\label{app:subprocedures:center and right dominated:case 2:6}If \mbox{($v$, $F_l(T_2)$)}, \mbox{($v$, $U_l(T_1)$)}, \mbox{($v$, $F_l(T_1)$)}, then $H_3$\opt{TIKZ}{
\tikzsetnextfilename{app-subprocedures-center-and-right-dominated-case-2-6}
	  \begin{center}
	  \begin{tikzpicture}[vertex/.style={shape=circle, draw=black, fill=black, line width=0pt, minimum size=1.5mm, inner sep=0mm, outer sep=0mm}, %
			      highlight vertex/.style={vertex,draw=blue, fill=blue},
			      essential edge/.style={thick,->,>=stealth}, 
			      implied edge/.style={thin,->,>=stealth, draw=gray},
			      highlight edge/.style={essential edge, ultra thick, draw=blue},
			      highlight implied edge/.style={implied edge, draw=blue},
			      every label/.style={font=\footnotesize, inner sep=0pt, outer sep=0pt},
			      label distance=2pt,
	  ]
	\node [highlight vertex] (T_3) at (1.732,-0.499) [label=below right:$T_3$] {};
\node [highlight vertex] (T_2) at (0.999,-0.866) [label=below right:$T_2$] {};
\node [highlight vertex] (U_lT_1) at (-2.0,0.0) [label= left:$U_l(T_1)$] {};
\node [highlight vertex] (F_lT_1) at (-1.732,-0.499) [label=below left:$F_l(T_1)$] {};
\node [highlight vertex] (v) at (0.0,1.0) [label=above:$v$] {};
\node [highlight vertex] (F_lT_2) at (-1.0,-0.866) [label=below left:$F_l(T_2)$] {};
\node [vertex] (F_rT_1) at (2.0,0.0) [label= right:$F_r(T_1)$] {};
\node [highlight vertex] (T_1) at (0.0,-1.0) [label=below:$T_1$] {};
\draw [implied edge,bend right=20] (T_3) to (F_rT_1);
\draw [implied edge,bend right=20] (T_2) to (F_rT_1);
\draw [highlight implied edge,bend right=20] (T_2) to (T_3);
\draw [highlight edge,bend right=20] (U_lT_1) to (F_lT_1);
\draw [highlight edge,bend right=20] (F_lT_1) to (T_1);
\draw [highlight implied edge,bend right=20] (F_lT_1) to (F_lT_2);
\draw [highlight edge] (v) to (T_2);
\draw [highlight edge] (v) to (F_lT_2);
\draw [highlight edge] (v) to (F_lT_1);
\draw [highlight edge] (v) to (U_lT_1);
\draw [highlight edge,bend right=20] (F_lT_2) to (T_2);
\draw [highlight implied edge,bend right=20] (F_lT_2) to (T_1);
\draw [essential edge,bend right=20] (T_1) to (F_rT_1);
\draw [highlight implied edge,bend right=20] (T_1) to (T_3);
\draw [highlight implied edge,bend right=20] (T_1) to (T_2);
\end{tikzpicture}
\end{center}}
\opt{NOTIKZ}{\begin{center}\includegraphics{\imgs 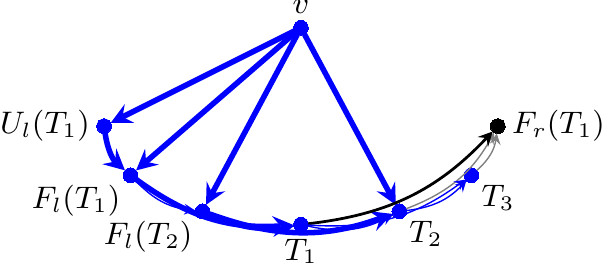}\end{center}}
\subsubsection{}\label{app:subprocedures:center and right dominated:case 2:7}If \mbox{$U_l(T_1)$ $\to$ $F_l(T_2)$}, \mbox{($v$, $F_l(T_2)$)}, \mbox{($v$, $U_l(T_1)$)}, then $W_5$\opt{TIKZ}{
\tikzsetnextfilename{app-subprocedures-center-and-right-dominated-case-2-7}
	  \begin{center}
	  \begin{tikzpicture}[vertex/.style={shape=circle, draw=black, fill=black, line width=0pt, minimum size=1.5mm, inner sep=0mm, outer sep=0mm}, %
			      highlight vertex/.style={vertex,draw=blue, fill=blue},
			      essential edge/.style={thick,->,>=stealth}, 
			      implied edge/.style={thin,->,>=stealth, draw=gray},
			      highlight edge/.style={essential edge, ultra thick, draw=blue},
			      highlight implied edge/.style={implied edge, draw=blue},
			      every label/.style={font=\footnotesize, inner sep=0pt, outer sep=0pt},
			      label distance=2pt,
	  ]
	\node [vertex] (T_3) at (1.732,-0.499) [label=below right:$T_3$] {};
\node [highlight vertex] (T_2) at (0.999,-0.866) [label=below right:$T_2$] {};
\node [highlight vertex] (U_lT_1) at (-2.0,0.0) [label= left:$U_l(T_1)$] {};
\node [highlight vertex] (F_lT_1) at (-1.732,-0.499) [label=below left:$F_l(T_1)$] {};
\node [highlight vertex] (v) at (0.0,1.0) [label=above:$v$] {};
\node [highlight vertex] (F_lT_2) at (-1.0,-0.866) [label=below left:$F_l(T_2)$] {};
\node [vertex] (F_rT_1) at (2.0,0.0) [label= right:$F_r(T_1)$] {};
\node [highlight vertex] (T_1) at (0.0,-1.0) [label=below:$T_1$] {};
\draw [implied edge,bend right=20] (T_3) to (F_rT_1);
\draw [implied edge,bend right=20] (T_2) to (F_rT_1);
\draw [implied edge,bend right=20] (T_2) to (T_3);
\draw [highlight edge,bend right=20] (U_lT_1) to (F_lT_2);
\draw [highlight implied edge,bend right=20] (U_lT_1) to (F_lT_1);
\draw [highlight edge,bend right=20] (F_lT_1) to (T_1);
\draw [highlight implied edge,bend right=20] (F_lT_1) to (F_lT_2);
\draw [highlight edge] (v) to (T_2);
\draw [highlight edge] (v) to (F_lT_2);
\draw [highlight edge] (v) to (U_lT_1);
\draw [highlight edge,bend right=20] (F_lT_2) to (T_2);
\draw [highlight implied edge,bend right=20] (F_lT_2) to (T_1);
\draw [essential edge,bend right=20] (T_1) to (F_rT_1);
\draw [implied edge,bend right=20] (T_1) to (T_3);
\draw [highlight implied edge,bend right=20] (T_1) to (T_2);
\end{tikzpicture}
\end{center}}
\opt{NOTIKZ}{\begin{center}\includegraphics{\imgs 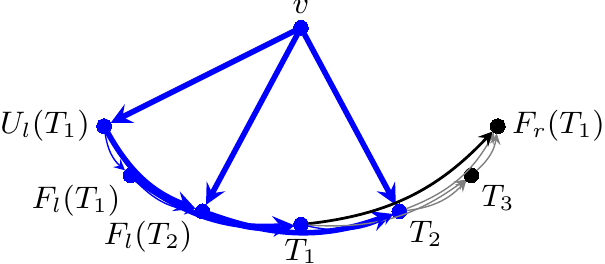}\end{center}}
\subsubsection{}\label{app:subprocedures:center and right dominated:case 2:8}If \mbox{$U_l(T_1)$ $\to$ $F_l(T_2)$}, \mbox{($v$, $F_l(T_2)$)}, \mbox{($v$, $U_l(T_1)$)}, \mbox{($v$, $F_l(T_1)$)}, then $H_2$\opt{TIKZ}{
\tikzsetnextfilename{app-subprocedures-center-and-right-dominated-case-2-8}
	  \begin{center}
	  \begin{tikzpicture}[vertex/.style={shape=circle, draw=black, fill=black, line width=0pt, minimum size=1.5mm, inner sep=0mm, outer sep=0mm}, %
			      highlight vertex/.style={vertex,draw=blue, fill=blue},
			      essential edge/.style={thick,->,>=stealth}, 
			      implied edge/.style={thin,->,>=stealth, draw=gray},
			      highlight edge/.style={essential edge, ultra thick, draw=blue},
			      highlight implied edge/.style={implied edge, draw=blue},
			      every label/.style={font=\footnotesize, inner sep=0pt, outer sep=0pt},
			      label distance=2pt,
	  ]
	\node [highlight vertex] (T_3) at (1.732,-0.499) [label=below right:$T_3$] {};
\node [highlight vertex] (T_2) at (0.999,-0.866) [label=below right:$T_2$] {};
\node [highlight vertex] (U_lT_1) at (-2.0,0.0) [label= left:$U_l(T_1)$] {};
\node [highlight vertex] (F_lT_1) at (-1.732,-0.499) [label=below left:$F_l(T_1)$] {};
\node [highlight vertex] (v) at (0.0,1.0) [label=above:$v$] {};
\node [highlight vertex] (F_lT_2) at (-1.0,-0.866) [label=below left:$F_l(T_2)$] {};
\node [vertex] (F_rT_1) at (2.0,0.0) [label= right:$F_r(T_1)$] {};
\node [highlight vertex] (T_1) at (0.0,-1.0) [label=below:$T_1$] {};
\draw [implied edge,bend right=20] (T_3) to (F_rT_1);
\draw [implied edge,bend right=20] (T_2) to (F_rT_1);
\draw [highlight implied edge,bend right=20] (T_2) to (T_3);
\draw [highlight edge,bend right=20] (U_lT_1) to (F_lT_2);
\draw [highlight implied edge,bend right=20] (U_lT_1) to (F_lT_1);
\draw [highlight implied edge,bend right=20] (F_lT_1) to (F_lT_2);
\draw [highlight edge,bend right=20] (F_lT_1) to (T_1);
\draw [highlight edge] (v) to (T_2);
\draw [highlight edge] (v) to (F_lT_2);
\draw [highlight edge] (v) to (F_lT_1);
\draw [highlight edge] (v) to (U_lT_1);
\draw [highlight edge,bend right=20] (F_lT_2) to (T_2);
\draw [highlight implied edge,bend right=20] (F_lT_2) to (T_1);
\draw [essential edge,bend right=20] (T_1) to (F_rT_1);
\draw [highlight implied edge,bend right=20] (T_1) to (T_3);
\draw [highlight implied edge,bend right=20] (T_1) to (T_2);
\end{tikzpicture}
\end{center}}
\opt{NOTIKZ}{\begin{center}\includegraphics{\imgs 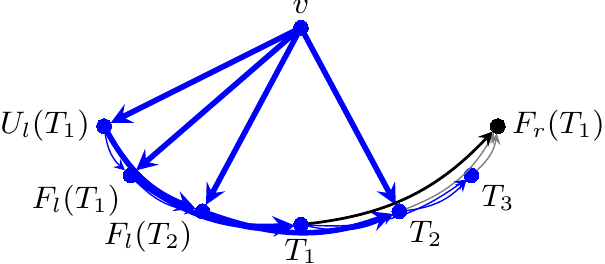}\end{center}}\subsubsection{Implied cases}\par By~\ref{app:subprocedures:center and right dominated:case 2:1}, at least one of \{\mbox{($v$, $F_l(T_2)$)}\} must belong to the graph when all of \{$\emptyset$\} are present.  Thus, the following cases are solved: \begin{enumerate} \setcounter{enumi}{8}\item \{\mbox{($v$, $F_r(T_1)$)}\}, \item \{\mbox{$F_r(T_1)$ $\to$ $U_l(T_1)$}\}, \item \{\mbox{($v$, $U_l(T_1)$)}\}, \item \{\mbox{($v$, $F_l(T_1)$)}\}, \item \{\mbox{$U_l(T_1)$ $\to$ $F_l(T_2)$}\}, \item \{\mbox{$F_r(T_1)$ $\to$ $U_l(T_1)$}, \mbox{($v$, $F_r(T_1)$)}\}, \item \{\mbox{($v$, $F_r(T_1)$)}, \mbox{($v$, $U_l(T_1)$)}\}, \item \{\mbox{($v$, $F_r(T_1)$)}, \mbox{($v$, $F_l(T_1)$)}\}, \item \{\mbox{$U_l(T_1)$ $\to$ $F_l(T_2)$}, \mbox{($v$, $F_r(T_1)$)}\}, \item \{\mbox{$F_r(T_1)$ $\to$ $U_l(T_1)$}, \mbox{($v$, $U_l(T_1)$)}\}, \item \{\mbox{$F_r(T_1)$ $\to$ $U_l(T_1)$}, \mbox{($v$, $F_l(T_1)$)}\}, \item \{\mbox{$F_r(T_1)$ $\to$ $U_l(T_1)$}, \mbox{$U_l(T_1)$ $\to$ $F_l(T_2)$}\}, \item \{\mbox{($v$, $U_l(T_1)$)}, \mbox{($v$, $F_l(T_1)$)}\}, \item \{\mbox{$U_l(T_1)$ $\to$ $F_l(T_2)$}, \mbox{($v$, $U_l(T_1)$)}\}, \item \{\mbox{$U_l(T_1)$ $\to$ $F_l(T_2)$}, \mbox{($v$, $F_l(T_1)$)}\}, \item \{\mbox{$F_r(T_1)$ $\to$ $U_l(T_1)$}, \mbox{($v$, $F_r(T_1)$)}, \mbox{($v$, $U_l(T_1)$)}\}, \item \{\mbox{$F_r(T_1)$ $\to$ $U_l(T_1)$}, \mbox{($v$, $F_r(T_1)$)}, \mbox{($v$, $F_l(T_1)$)}\}, \item \{\mbox{$F_r(T_1)$ $\to$ $U_l(T_1)$}, \mbox{$U_l(T_1)$ $\to$ $F_l(T_2)$}, \mbox{($v$, $F_r(T_1)$)}\}, \item \{\mbox{($v$, $F_r(T_1)$)}, \mbox{($v$, $U_l(T_1)$)}, \mbox{($v$, $F_l(T_1)$)}\}, \item \{\mbox{$U_l(T_1)$ $\to$ $F_l(T_2)$}, \mbox{($v$, $F_r(T_1)$)}, \mbox{($v$, $U_l(T_1)$)}\}, \item \{\mbox{$U_l(T_1)$ $\to$ $F_l(T_2)$}, \mbox{($v$, $F_r(T_1)$)}, \mbox{($v$, $F_l(T_1)$)}\}, \item \{\mbox{$F_r(T_1)$ $\to$ $U_l(T_1)$}, \mbox{($v$, $U_l(T_1)$)}, \mbox{($v$, $F_l(T_1)$)}\}, \item \{\mbox{$F_r(T_1)$ $\to$ $U_l(T_1)$}, \mbox{$U_l(T_1)$ $\to$ $F_l(T_2)$}, \mbox{($v$, $U_l(T_1)$)}\}, \item \{\mbox{$F_r(T_1)$ $\to$ $U_l(T_1)$}, \mbox{$U_l(T_1)$ $\to$ $F_l(T_2)$}, \mbox{($v$, $F_l(T_1)$)}\}, \item \{\mbox{$U_l(T_1)$ $\to$ $F_l(T_2)$}, \mbox{($v$, $U_l(T_1)$)}, \mbox{($v$, $F_l(T_1)$)}\}, \item \{\mbox{$F_r(T_1)$ $\to$ $U_l(T_1)$}, \mbox{($v$, $U_l(T_1)$)}, \mbox{($v$, $F_r(T_1)$)}, \mbox{($v$, $F_l(T_1)$)}\}, \item \{\mbox{$U_l(T_1)$ $\to$ $F_l(T_2)$}, \mbox{$F_r(T_1)$ $\to$ $U_l(T_1)$}, \mbox{($v$, $F_r(T_1)$)}, \mbox{($v$, $F_l(T_1)$)}\}, \item \{\mbox{$U_l(T_1)$ $\to$ $F_l(T_2)$}, \mbox{($v$, $U_l(T_1)$)}, \mbox{($v$, $F_r(T_1)$)}, \mbox{($v$, $F_l(T_1)$)}\}, \item \{\mbox{$U_l(T_1)$ $\to$ $F_l(T_2)$}, \mbox{$F_r(T_1)$ $\to$ $U_l(T_1)$}, \mbox{($v$, $U_l(T_1)$)}, \mbox{($v$, $F_l(T_1)$)}\}, \item \{\mbox{$U_l(T_1)$ $\to$ $F_l(T_2)$}, \mbox{$F_r(T_1)$ $\to$ $U_l(T_1)$}, \mbox{($v$, $U_l(T_1)$)}, \mbox{($v$, $F_r(T_1)$)}, \mbox{($v$, $F_l(T_1)$)}\}\end{enumerate}\par By~\ref{app:subprocedures:center and right dominated:case 2:2}, at least one of \{\mbox{$U_l(T_1)$ $\to$ $F_l(T_2)$}, \mbox{($v$, $U_l(T_1)$)}, \mbox{($v$, $F_l(T_1)$)}\} must belong to the graph when all of \{\mbox{($v$, $F_l(T_2)$)}\} are present.  Thus, the following cases are solved: \begin{enumerate} \setcounter{enumi}{38}\item \{\mbox{($v$, $F_l(T_2)$)}, \mbox{($v$, $F_r(T_1)$)}\}, \item \{\mbox{$F_r(T_1)$ $\to$ $U_l(T_1)$}, \mbox{($v$, $F_l(T_2)$)}\}, \item \{\mbox{$F_r(T_1)$ $\to$ $U_l(T_1)$}, \mbox{($v$, $F_l(T_2)$)}, \mbox{($v$, $F_r(T_1)$)}\}\end{enumerate}\par By~\ref{app:subprocedures:center and right dominated:case 2:3}, at least one of \{\mbox{$U_l(T_1)$ $\to$ $F_l(T_2)$}, \mbox{($v$, $F_l(T_1)$)}\} must belong to the graph when all of \{\mbox{($v$, $F_l(T_2)$)}, \mbox{($v$, $U_l(T_1)$)}\} are present.  Thus, the following cases are solved: \begin{enumerate} \setcounter{enumi}{41}\item \{\mbox{($v$, $F_l(T_2)$)}, \mbox{($v$, $F_r(T_1)$)}, \mbox{($v$, $U_l(T_1)$)}\}, \item \{\mbox{$F_r(T_1)$ $\to$ $U_l(T_1)$}, \mbox{($v$, $F_l(T_2)$)}, \mbox{($v$, $U_l(T_1)$)}\}, \item \{\mbox{$F_r(T_1)$ $\to$ $U_l(T_1)$}, \mbox{($v$, $F_l(T_2)$)}, \mbox{($v$, $U_l(T_1)$)}, \mbox{($v$, $F_r(T_1)$)}\}\end{enumerate}\par By~\ref{app:subprocedures:center and right dominated:case 2:4}, at least one of \{\mbox{($v$, $U_l(T_1)$)}\} must belong to the graph when all of \{\mbox{($v$, $F_l(T_2)$)}, \mbox{($v$, $F_l(T_1)$)}\} are present.  Thus, the following cases are solved: \begin{enumerate} \setcounter{enumi}{44}\item \{\mbox{($v$, $F_l(T_2)$)}, \mbox{($v$, $F_r(T_1)$)}, \mbox{($v$, $F_l(T_1)$)}\}, \item \{\mbox{$F_r(T_1)$ $\to$ $U_l(T_1)$}, \mbox{($v$, $F_l(T_2)$)}, \mbox{($v$, $F_l(T_1)$)}\}, \item \{\mbox{$U_l(T_1)$ $\to$ $F_l(T_2)$}, \mbox{($v$, $F_l(T_2)$)}, \mbox{($v$, $F_l(T_1)$)}\}, \item \{\mbox{$F_r(T_1)$ $\to$ $U_l(T_1)$}, \mbox{($v$, $F_l(T_2)$)}, \mbox{($v$, $F_r(T_1)$)}, \mbox{($v$, $F_l(T_1)$)}\}, \item \{\mbox{$U_l(T_1)$ $\to$ $F_l(T_2)$}, \mbox{($v$, $F_l(T_2)$)}, \mbox{($v$, $F_r(T_1)$)}, \mbox{($v$, $F_l(T_1)$)}\}, \item \{\mbox{$U_l(T_1)$ $\to$ $F_l(T_2)$}, \mbox{$F_r(T_1)$ $\to$ $U_l(T_1)$}, \mbox{($v$, $F_l(T_2)$)}, \mbox{($v$, $F_l(T_1)$)}\}, \item \{\mbox{$U_l(T_1)$ $\to$ $F_l(T_2)$}, \mbox{$F_r(T_1)$ $\to$ $U_l(T_1)$}, \mbox{($v$, $F_l(T_2)$)}, \mbox{($v$, $F_r(T_1)$)}, \mbox{($v$, $F_l(T_1)$)}\}\end{enumerate}\par By~\ref{app:subprocedures:center and right dominated:case 2:5}, at least one of \{\mbox{($v$, $U_l(T_1)$)}\} must belong to the graph when all of \{\mbox{$U_l(T_1)$ $\to$ $F_l(T_2)$}, \mbox{($v$, $F_l(T_2)$)}\} are present.  Thus, the following cases are solved: \begin{enumerate} \setcounter{enumi}{51}\item \{\mbox{$U_l(T_1)$ $\to$ $F_l(T_2)$}, \mbox{($v$, $F_l(T_2)$)}, \mbox{($v$, $F_r(T_1)$)}\}, \item \{\mbox{$F_r(T_1)$ $\to$ $U_l(T_1)$}, \mbox{$U_l(T_1)$ $\to$ $F_l(T_2)$}, \mbox{($v$, $F_l(T_2)$)}\}, \item \{\mbox{$U_l(T_1)$ $\to$ $F_l(T_2)$}, \mbox{$F_r(T_1)$ $\to$ $U_l(T_1)$}, \mbox{($v$, $F_l(T_2)$)}, \mbox{($v$, $F_r(T_1)$)}\}\end{enumerate}\par By~\ref{app:subprocedures:center and right dominated:case 2:6}, at least one of \{\mbox{$U_l(T_1)$ $\to$ $F_l(T_2)$}\} must belong to the graph when all of \{\mbox{($v$, $F_l(T_2)$)}, \mbox{($v$, $U_l(T_1)$)}, \mbox{($v$, $F_l(T_1)$)}\} are present.  Thus, the following cases are solved: \begin{enumerate} \setcounter{enumi}{54}\item \{\mbox{($v$, $F_l(T_2)$)}, \mbox{($v$, $U_l(T_1)$)}, \mbox{($v$, $F_r(T_1)$)}, \mbox{($v$, $F_l(T_1)$)}\}, \item \{\mbox{$F_r(T_1)$ $\to$ $U_l(T_1)$}, \mbox{($v$, $F_l(T_2)$)}, \mbox{($v$, $U_l(T_1)$)}, \mbox{($v$, $F_l(T_1)$)}\}, \item \{\mbox{$F_r(T_1)$ $\to$ $U_l(T_1)$}, \mbox{($v$, $F_l(T_2)$)}, \mbox{($v$, $U_l(T_1)$)}, \mbox{($v$, $F_r(T_1)$)}, \mbox{($v$, $F_l(T_1)$)}\}\end{enumerate}\par By~\ref{app:subprocedures:center and right dominated:case 2:7}, at least one of \{\mbox{($v$, $F_l(T_1)$)}\} must belong to the graph when all of \{\mbox{$U_l(T_1)$ $\to$ $F_l(T_2)$}, \mbox{($v$, $F_l(T_2)$)}, \mbox{($v$, $U_l(T_1)$)}\} are present.  Thus, the following cases are solved: \begin{enumerate} \setcounter{enumi}{57}\item \{\mbox{$U_l(T_1)$ $\to$ $F_l(T_2)$}, \mbox{($v$, $F_l(T_2)$)}, \mbox{($v$, $U_l(T_1)$)}, \mbox{($v$, $F_r(T_1)$)}\}, \item \{\mbox{$U_l(T_1)$ $\to$ $F_l(T_2)$}, \mbox{$F_r(T_1)$ $\to$ $U_l(T_1)$}, \mbox{($v$, $F_l(T_2)$)}, \mbox{($v$, $U_l(T_1)$)}\}\end{enumerate}\par By~\ref{app:subprocedures:center and right dominated:case 2:8}, at least one of \{$\emptyset$\} must belong to the graph when all of \{\mbox{$U_l(T_1)$ $\to$ $F_l(T_2)$}, \mbox{($v$, $F_l(T_2)$)}, \mbox{($v$, $U_l(T_1)$)}, \mbox{($v$, $F_l(T_1)$)}\} are present.  Thus, the following cases are solved: \begin{enumerate} \setcounter{enumi}{59}\item \{\mbox{$U_l(T_1)$ $\to$ $F_l(T_2)$}, \mbox{($v$, $F_l(T_2)$)}, \mbox{($v$, $U_l(T_1)$)}, \mbox{($v$, $F_r(T_1)$)}, \mbox{($v$, $F_l(T_1)$)}\}, \item \{\mbox{$U_l(T_1)$ $\to$ $F_l(T_2)$}, \mbox{$F_r(T_1)$ $\to$ $U_l(T_1)$}, \mbox{($v$, $F_l(T_2)$)}, \mbox{($v$, $U_l(T_1)$)}, \mbox{($v$, $F_l(T_1)$)}\}, \item \{\mbox{$U_l(T_1)$ $\to$ $F_l(T_2)$}, \mbox{$F_r(T_1)$ $\to$ $U_l(T_1)$}, \mbox{($v$, $F_l(T_2)$)}, \mbox{($v$, $U_l(T_1)$)}, \mbox{($v$, $F_r(T_1)$)}, \mbox{($v$, $F_l(T_1)$)}\}\end{enumerate}
\input{lemmaBadBlockLong}
\input{lemmaBadBlockShortCase1}
\input{lemmaBadBlockShortCase2}
\input{lemmaRangeCase1}
\input{lemmaRangeCase2}
\input{lemmaReceptionFailsOppositeBlockCase1}
\input{lemmaReceptionFailsOppositeBlockCase2}
\input{lemmaReceptionFailsOppositeBlockCase3}
\input{lemmaReceptionFailsOppositeBlockCase4}
\input{lemmaReceptionFailsCase1}
\input{lemmaReceptionFailsCase2}
\input{lemmaInsertionJoinCase1}
\input{lemmaInsertionJoinCase2}
\input{lemmaInsertionJoinCase3}
\input{lemmaInsertionJoinCase4}

\end{document}